\documentclass[11pt,a4paper]{article}

\usepackage[utf8]{inputenc}
\usepackage[T1]{fontenc}
\usepackage{amsmath,amssymb,amsthm}
\usepackage{mathtools}
\usepackage[margin=1.1in]{geometry}
\usepackage{microtype}
\usepackage{xcolor}
\usepackage{hyperref}
\hypersetup{
  colorlinks=true,
  linkcolor=blue!50!black,
  urlcolor=blue!50!black,
  citecolor=blue!50!black
}

\usepackage[authoryear,round,comma]{natbib}

\usepackage{tikz}
\usepackage{pgfplots}
\pgfplotsset{compat=1.18}
\usepgfplotslibrary{fillbetween}
\usepackage{subcaption}
\usepackage{float}

\theoremstyle{plain}
\newtheorem{theorem}{Theorem}
\newtheorem{proposition}[theorem]{Proposition}
\newtheorem{lemma}[theorem]{Lemma}
\newtheorem{corollary}[theorem]{Corollary}

\theoremstyle{definition}
\newtheorem{definition}[theorem]{Definition}
\newtheorem{axiom}{Axiom}

\newtheorem{remark}[theorem]{Remark}

\newcommand{\R}{\mathbb{R}}

\newcommand{\E}{\mathbb{E}}
\newcommand{\Prob}{\mathbb{P}}
\newcommand{\Q}{\mathbb{Q}}
\newcommand{\F}{\mathcal{F}}
\newcommand{\filt}{\mathbb{F}}
\newcommand{\Linf}{L^{\infty}}
\newcommand{\WL}{W^{L}}
\newcommand{\eqdist}{\stackrel{d}{=}}
\newcommand{\re}{\, \| \,}

\title{Avellaneda--Stoikov and Cartea--Jaimungal as One Framework: \\ 
A Forced Uniqueness Theorem for Inventory Market Making}
 \author{Frank M.\,V.\,Feys}
 \date{May 30, 2026}

\begin{document}

\maketitle

\begin{abstract}
In inventory market making, the running-penalty coefficient  $\phi$ of the Cartea--Jaimungal framework and the risk-aversion parameter $\gamma$ of  the Avellaneda--Stoikov framework are typically treated as independent free parameters, 
calibrated separately.
We show that they are in fact not independent.
A small set of axioms on the market maker's dynamic preference functional, 
namely cash-additivity, normalization, concavity, strong dynamic consistency, and law-invariance, 
forces the preference functional to be the entropic certainty-equivalent on liquidation-adjusted terminal wealth, 
parametrized by a single positive scalar $\gamma$.
The Avellaneda--Stoikov framework is the  unique representative of this axiom class.
The Cartea--Jaimungal framework is its second-order Taylor expansion in inventory magnitude, 
with the running coefficient forced to $\phi = \gamma\sigma^2/2$ and 
 (under a mild regularity condition on the liquidation cost) 
  the terminal coefficient forced to $\alpha =  \tfrac{1}{2}L''(0)$.
The two frameworks, typically presented as competing alternatives with the choice between them driven by tractability, 
are different manifestations of a single underlying object.
The forced relation is invertible, $\gamma = 2\phi/\sigma^2$, 
giving a consistency cross-check on independently calibrated desk parameters.
\end{abstract}

\noindent\textbf{MSC 2020 classification.}
Primary: 91B16 (utility theory),
91G70 (statistical methods; risk measures).
Secondary: 49L20 (dynamic programming in optimal control and differential games),
60H30 (applications of stochastic analysis),
91G10 (portfolio theory),
91G80 (financial applications of other theories),
93E20 (optimal stochastic control).

\medskip

\noindent\textbf{Keywords.}
Inventory market making; 
forced uniqueness; 
axiomatic preference theory; 
Avellaneda--Stoikov; 
Cartea--Jaimungal; 
dynamic consistency; 
entropic risk measure.


\section{Introduction}
\label{sec:intro}

A market maker, in the inventory-management tradition, faces a problem with a deceptively simple structure.
She posts bid and ask quotes around a prevailing mid-price, earns the spread when her quotes are hit, accumulates inventory (typically unwanted) in the process, and must manage the risk that the mid-price moves against her position before she can unwind it, all under uncertainty about future prices and order arrivals.
At the preference-functional level, the mathematical literature is organized around two principal frameworks,
each emphasizing a different representation of risk and inventory aversion.

The first is the Avellaneda--Stoikov (AS) framework. 
Originating with \citet{HoStoll1981} and cast in modern quantitative form by \citet{AvellanedaStoikov2008}, the AS market maker is endowed with constant absolute risk aversion (CARA) preferences over terminal wealth, comprising any residual inventory marked to the terminal mid-price net of a convex liquidation penalty.
The optimal quoting strategy solves a Hamilton--Jacobi--Bellman equation,
 and the resulting reservation-price and optimal-spread formulas are by now standard textbook material, with the risk-aversion coefficient $\gamma$ as the single free preference parameter.
 
The second is the Cartea--Jaimungal framework, which we shall refer to as the CJ tradition.
Beginning with \citet{CarteaJaimungal2015} and codified in \citet{CarteaJaimungalPenalva2015}, 
the CJ market maker maximizes expected terminal wealth subject to a running quadratic inventory penalty and a terminal liquidation penalty.
The running-penalty coefficient $\phi$ and the terminal-penalty coefficient $\alpha$ are treated as independent free parameters, 
typically calibrated to market data or chosen on tractability grounds, 
 with no principled relation between them.

The two traditions are conventionally presented as competing alternatives.
In \citet{CarteaJaimungalPenalva2015}, 
\citet{Gueant2017}, and essentially every applied paper we have surveyed, 
the choice between AS and CJ is governed by tractability, 
by the analytic form of the resulting strategy, or by the practitioner's modeling preference, 
but not by anything resembling a principled argument.
This is the situation that we wish to revise.

Before proceeding, we note the scope of the present contribution.
We do not model adverse selection, queue position, toxic flow, 
or the latency-sensitive components of execution;
 our results address the preference-functional layer alone, 
 which sits above these microstructural concerns and constrains how their outputs should be aggregated into a quoting policy.

We propose to derive the market maker's objective from a small set of axioms.
Our five core axioms are stated and motivated in Section~\ref{sec:axioms}; the underlying idea is to ask, for each axiom, whether a working market maker would endorse it on reflection.
Some are uncontroversial, such as cash-additivity (a sure dollar adds a dollar of value) and normalization (a zero-terminal-wealth strategy has zero value); others carry economic content, such as concavity, strong dynamic consistency, and law-invariance.
We also record three additional natural properties, all derivable from the five core axioms: \emph{monotonicity} (M), the \emph{wealth-summary} property (W), and \emph{right-continuity in time} (R).
 One structural consequence is microstructural and, to our knowledge, has not previously been derived as a feature of the preference functional: \emph{clock-invariance} of the risk-aversion parameter with respect to the price process's quadratic-variation clock, which we shall argue is the right time-stationarity property for an inventory market maker.
The constant-volatility AS benchmark satisfies clock-invariance vacuously, whereas the standard CJ extension to stochastic-volatility settings violates it, as we show in Section~\ref{sec:supporting} via Proposition~\ref{prop:clock-corrected}.
Clock-invariance is not an additional axiom but a consequence of the entropic representation, recorded as Corollary~\ref{cor:clock-invariance} below.

Our main result, Theorem~\ref{thm:forced-uniqueness},
  is that the five axioms force the market maker's preference functional to be the entropic certainty-equivalent on liquidation-adjusted terminal wealth, parametrized by a single positive scalar $\gamma$.
Equivalently, the associated risk measure $\rho_t = -J_t$ is the (conditional) entropic risk measure.
The Avellaneda--Stoikov framework is the unique axiom-consistent representative;
 the Cartea--Jaimungal framework is the second-order Taylor expansion of the same functional, 
 with its free hyperparameters $\phi$ and $\alpha$  no longer free but fixed 
 by the relations $\phi = \gamma\sigma^2/2$ and $\alpha = \tfrac{1}{2}L''(0)$.
To the best of our knowledge, no prior statement of this result exists in the inventory market making literature.
The closest prior work is the ODE-level unification of \citet{Gueant2017}, 
discussed in detail in Section~\ref{sec:supporting}; 
our contribution lies at the preference-functional level rather than at the value-function level.

The result is conditional on the axiom system, and we are explicit about this.
Theorem~\ref{thm:forced-uniqueness}  establishes the forcing as a consequence of the five core axioms.
 It is not a free-standing claim about all reasonable market makers, 
and it does not preclude a practitioner from operating outside the axiom class if she has good reasons to do so.
The substantive content most likely to be challenged is the wealth-summary property (Proposition~\ref{prop:wealth-summary}), 
which commits the market maker to evaluating strategies solely through their liquidation-adjusted terminal wealth $\WL_T$, 
and thereby excludes path-functional preferences such as intraday drawdown aversion (discussed in Section~\ref{subsec:drawdown}).
A skeptic might object that wealth-summary is too restrictive an assumption. 
The derived-property framing of Section~\ref{subsec:derived-properties} answers this objection precisely: 
wealth-summary is not an axiom, and the skeptic cannot reject it directly.
It is a theorem, derived from cash-additivity (J1),  normalization (J2), 
 and strong dynamic consistency (J4) alone.
Any objection to wealth-summary must therefore take the form of a rejection of one of those three core axioms, 
each of which carries independent and transparent economic justification.
In particular, a practitioner who endorses cash-additivity (a sure dollar adds a dollar of value), 
normalization (a zero-terminal-wealth strategy has zero value), 
and strong dynamic consistency (earlier preferences should not knowingly disagree with later ones on the same information)
 is automatically committed to wealth-summary,
 whether or not she finds the wealth-summary framing intuitive on its face.
The forcing conclusions of the theorem then follow for her as a matter of pure logic.

From our main theorem, four rather sharp consequences follow.
First, the AS framework is uniquely correct among law-invariant dynamically-consistent inventory market making models; 
the choice of CARA utility, motivated in the original AS paper by tractability,
 is in fact forced by the axioms (Corollary~\ref{cor:as-uniqueness}).
Second, the CJ framework is incompatible with our axioms as a primitive functional,
 but can be rescued as the second-order expansion of the entropic functional around zero inventory 
 with the forced coefficient $\phi = \gamma\sigma^2/2$ (Corollary~\ref{cor:cj-inconsistency}).
Third, the CJ terminal-penalty coefficient $\alpha$ is similarly forced, equal to $\tfrac{1}{2}L''(0)$ (Corollary~\ref{cor:forced-alpha}); 
it is as such not a preference parameter but a property of the market microstructure.
  The two forcings are conceptually distinct, and the distinction is worth emphasizing.
Indeed, $\phi$ is a preference parameter disguised as  a market parameter, forced by the entropic structure of the agent's preferences via $\gamma$, 
while $\alpha$ is a genuine market parameter, forced by the curvature of the liquidation cost function $L$, which some practitioners misread as a preference parameter.
A practitioner calibrating $(\phi, \alpha, \gamma)$ as three free parameters is therefore over-parametrizing on both sides of the preference--market boundary.
Fourth, dynamic conditional-value-at-risk preferences are time-inconsistent in our framework (Corollary~\ref{cor:cvar}), 
  with implications for the gap between desk-internal optimization and CVaR-based regulatory frameworks.
In Section~\ref{sec:consequences} we develop eleven corollaries in total, including the four named above.

\paragraph{Executive summary.}
For the reader who wants the bottom line before the details:
\begin{itemize}
\item \emph{What is forced.} 
The running-penalty coefficient $\phi$  and risk-aversion parameter $\gamma$ are not independent: 
$\phi = \gamma\sigma^2/2$ is forced by five axioms on the market maker's dynamic preferences.
\item \emph{What else is forced.} 
Under a mild regularity condition on the liquidation cost function, the terminal-penalty coefficient $\alpha$ is not a preference parameter at all;
 it is forced to equal $\tfrac{1}{2}L''(0)$, a property of the market's cost-to-liquidate curvature.
A CJ desk calibrating  $(\phi, \alpha, \gamma)$ as three free parameters is over-parametrizing on both sides of the preference--market boundary.
\item \emph{What is free.} 
A single scalar $\gamma > 0$.
Everything else in both the AS and CJ frameworks follows from $\gamma$ and the market inputs.
\item \emph{What is testable.} 
A CJ desk can recover its implicit $\gamma$ via $\gamma = 2\phi/\sigma^2$ and cross-check it against an AS-style calibration from quoted spreads.
Persistent disagreement between the two is a signal of  miscalibration or of a violated axiom.
\item \emph{What changes in stochastic volatility.} 
The inventory penalty must use quadratic-variation time:
 $\int_0^T \sigma_s^2 q_s^2\, ds$, not $\int_0^T q_s^2\, ds$.
Wall-clock penalties systematically misprice inventory on high- and low-volatility days.
\item \emph{What is not covered.} 
Adverse selection, queue position, intraday drawdown limits, 
and regime-switching risk aversion are outside the framework; 
they operate at a different layer   of the market-making stack.
\item \emph{A regulatory note.} 
Dynamic CVaR is not consistent with the axiom system; 
no CVaR-based objective satisfies strong dynamic consistency (J4).
A desk's time-consistent internal optimization is therefore entropic, 
while the regulator's CVaR limit is a different, non-entropic object.
\end{itemize}

\paragraph{Outline.}
Section~\ref{sec:setup} introduces the market microstructural setup, the strategy space, and the liquidation-adjusted terminal wealth.
Section~\ref{sec:axioms} states and motivates the five core axioms and three derived properties.
 Section~\ref{sec:theorem} proves the main theorem and records clock-invariance (Corollary~\ref{cor:clock-invariance});    
 Section~\ref{sec:consequences} develops eleven further corollaries.
Section~\ref{sec:supporting} derives the forced coefficient via HJB (Proposition~\ref{prop:forced-phi}) and extends to stochastic volatility (Proposition~\ref{prop:clock-corrected}).
Section~\ref{sec:discussion} opens with the operational implications for a working desk (Section~\ref{subsec:operational}), then discusses scope (in particular, the comparison with \citet{Gueant2017} and drawdown-averse market makers), records further structural connections, and points to four directions for future work.
Section~\ref{sec:conclusion} concludes.


\section{The Market Making Setup}
\label{sec:setup}

We fix a finite trading horizon $T > 0$ and a filtered probability space $(\Omega, \F, \filt, \Prob)$,  where $\filt = (\F_t)_{t \in [0,T]}$ is a fixed filtration satisfying the usual conditions (right-continuity and completeness).
We further assume that the initial $\sigma$-algebra $\F_0$ is $\Prob$-trivial, that is, every $A \in \F_0$ satisfies $\Prob(A) \in \{0, 1\}$.
This is a standard assumption in the dynamic risk measure literature \citep{KupperSchachermayer2009}.
 It expresses the requirement that no information beyond the model primitives is available to the market maker at the start of the trading horizon.
 We also assume that the terminal $\sigma$-algebra $\F_T$ is non-atomic: for every $A \in \F_T$ with $\Prob(A) > 0$ there exists $B \in \F_T$ with $B \subseteq A$ and $0 < \Prob(B) < \Prob(A)$.
This is likewise standard in the dynamic risk measure literature and is satisfied by any filtered Brownian space or more generally by any space on which a continuously-distributed random variable is $\F_T$-measurable.
The filtration $\filt$ is given as a primitive of the model and represents the information available to the market maker at each time $t \in [0,T]$;
 $\F_t$ is interpreted as the $\sigma$-algebra of all events whose occurrence is known by time $t$.
All processes in this paper are adapted to  $\filt$ in the standard sense of stochastic analysis, 
meaning that the value at time $t$ is $\F_t$-measurable.
The primitives, namely the mid-price $S$ (together with its volatility $\sigma$ in the stochastic-volatility setting),
 the order-flow counts $N^a, N^b$, the strategy $\pi = (\delta^a,  \delta^b, C)$, the intensities $\lambda^a, \lambda^b$, 
 and the preference functional $J$, are made adapted by direct assumption; 
 all derived processes (the quadratic variation $\langle S\rangle$,
  the canonical clock $\Lambda$, the cash $X$, the inventory $q$, the wealth $W$, and the liquidation-adjusted terminal wealth $\WL_T$) are then adapted automatically.
The strategy $\pi$, the volatility  $\sigma$, and the intensities $\lambda^a, \lambda^b$ are required to be \emph{predictable}, 
the strengthening of adaptedness appropriate when integrating against jump or diffusion components.

 Throughout, we work in $\Linf = \Linf(\F_T)$, the space of bounded $\F_T$-measurable random variables, together with the natural conditional spaces $\Linf(\F_t)$ for $t \in [0, T]$.
This is the same setting as the one used in the dynamic risk measure literature; 
  see,  for example,  \citet{KupperSchachermayer2009} and \citet{BionNadal2009}.
The boundedness assumption is a technical convenience that simplifies the representation theorems on which our main result rests, and is standard in the axiomatic literature on dynamic risk measures.

We now describe, in turn, the price process, the order-flow processes, the strategy space, the wealth dynamics, 
the liquidation cost function, and the dynamic preference functional.
Each component is introduced at the minimal level of generality required for the main result, and we note explicitly where stronger or weaker hypotheses would suffice.

\subsection{Price Process and Order Flow}
\label{subsec:price-flow}

The mid-price $(S_t)_{t \in [0,T]}$ is a continuous semimartingale with quadratic variation process $\langle S \rangle$.
The constant-volatility benchmark, which we shall use as a running example, is $S_t = S_0 + \sigma B_t$ for a Brownian motion $B$ and a constant $\sigma > 0$, so that $d\langle S \rangle_t = \sigma^2 \, dt$.
More generally, we allow $S$ to be of the form $dS_t =  \mu_t \, dt + \sigma_t \, dB_t$ with $\sigma$  a positive predictable process, in which case $d\langle S \rangle_t =  \sigma_t^2 \, dt$.
Note that the drift $\mu$ plays no role in our analysis, and in what follows we set it to zero, without loss of generality for the structural results.

Order flow is described by two counting processes, $N^a = (N^a_t)$ and $N^b = (N^b_t)$,  representing the cumulative number of fills at the market maker's ask and bid, respectively.
We take $N^a$ and $N^b$ to be Cox processes
  (i.e., doubly stochastic Poisson processes) with intensities $\lambda^a_t$ and $\lambda^b_t$ that depend on the market maker's quote distances $(\delta^a, \delta^b)$ and on the relevant state variables.
The classical AS specification is $\lambda^a(\delta^a) = A e^{-\kappa \delta^a}$ for the ask intensity at distance $\delta^a$ from the mid, and similarly for the bid, but we shall not need this functional form for the structural results;
  we require only that $\lambda^a, \lambda^b$ be predictable, nonnegative, and locally integrable.

For the technical conditions of Appendix~\ref{app:technical} to be nonvacuous on the canonical Brownian setup, 
  we work throughout with a stopped version of the price process.
Fix a deterministic constant $M > 0$ and let $\tau_M  \coloneqq \inf\{t \in [0,T] \mid |S_t| > M\}$, with $\tau_M = T$ if the infimum is not attained.
Define the \emph{stopped mid-price} $\bar S_t  \coloneqq S_{t \wedge \tau_M}$.
The stopped process $\bar S$ is a continuous semimartingale, $|\bar S_t| \leq M$ almost surely for every $t \in [0,T]$, and $\langle \bar S\rangle_t = \langle S\rangle_{t \wedge \tau_M}$.
Since $M$ can be chosen so that $\Prob(\tau_M \leq T)$ is arbitrarily small, the stopping is operationally invisible on intraday horizons, with all structural results holding for $S$ on the high-probability event $\{\tau_M > T\}$.
For the rest of the paper, $S$ refers to $\bar S$, and we drop the bar notation.
A treatment without the stopping, using an Orlicz space in place of $\Linf$ in the spirit of \citet{CheriditoLi2008},
  is also possible, and it would deliver the same structural conclusions.

\subsection{Strategy Space}
\label{subsec:strategy}

The market maker's decision profile over the trading horizon has two components:
 a quoting policy (where to post limit orders relative to the mid-price) and a cash-management policy
  (how much capital to commit to or withdraw from the trading account).
We model the strategy as a triple $\pi = (\delta^a, \delta^b, C)$, where $\delta^a, \delta^b$ are $\R$-valued predictable processes
 (the quote distances from the mid-price) and $C$ is an $\F$-adapted c\`adl\`ag process of finite variation on $[0,T]$ (the cumulative cash injection into the trading account).
Predictability of $\delta^a, \delta^b$ is required because the quoting decision at time $t$ must be measurable with respect to information strictly prior to $t$ (to be matched against the jump-arrival processes $N^a, N^b$, which are integrated against in~\eqref{eq:cash}).
The cash injection $C$ is integrated as a Stieltjes integral against time, which requires only adaptedness; 
predictability would be a stronger condition than necessary and we do not impose it.
We allow for more general control structures, so as to accommodate possible extensions to multi-tick or hedge-augmented strategies.
The cash-injection component $C$ is more than a mere modeling convenience. 
Indeed, it is mathematically necessary for the reduction argument underlying Theorem~\ref{thm:forced-uniqueness}.
Lemma~\ref{lem:density} (Appendix~\ref{app:technical}) uses terminal cash injections to establish that the range of $\WL_T$ over $\Pi$ is all of $\Linf(\F_T)$, a density condition allowing the reduced functional $\tilde J_t$ to be defined on $\Linf(\F_T)$ and the Kupper--Schachermayer representation theorem to apply.
Its operational significance in quoting is secondary; 
for the purposes of the present results, the reader may take $C$ to consist of a single lump-sum injection at time $0$.
We emphasize that the closure of $\Pi$ under bounded cash injection 
(i.e., that $\pi^{[c]} \in \Pi$ whenever $\pi \in \Pi$ and $c \in \Linf(\F_t)$) 
 is a modeling hypothesis on the strategy space, distinct from and logically independent of the five preference axioms of Section~\ref{sec:axioms}; 
it is verified for the specific admissibility conditions of Definition~\ref{def:admissible} in Lemma~\ref{lem:density}.

 We take $\Pi$, the set of admissible strategies, to be a convex subset of the space of such triples, satisfying integrability and boundedness conditions sufficient for existence and pathwise uniqueness of the wealth and inventory dynamics introduced below.
Precise admissibility requirements are stated in Definition~\ref{def:admissible} of Appendix~\ref{app:technical}.

We shall make use of the following construction.
Given any strategy $\pi = ( \delta^a, \delta^b, C) \in \Pi$, any $t \in [0,T]$, and any $\F_t$-measurable bounded random variable $c$, the \emph{cash-injected strategy} $\pi^{[c]}$ is obtained from $\pi$ by adding $c$ to the cash-injection process at time $t$.
To be precise,  $\pi^{[c]} = (\delta^a, \delta^b, C^{[c]})$ with $C^{[c]}_s \coloneqq C_s + c\, \mathbf{1}_{\{s \geq t\}}$.
The quote-distance and inventory processes are unchanged, the terminal cash satisfies $X_T(\pi^{[c]}) = X_T(\pi) + c$, 
and the liquidation-adjusted terminal wealth satisfies  $\WL_T(\pi^{[c]}) = \WL_T(\pi) + c$ almost surely.
Admissibility of $\pi^{[c]}$ is verified in Lemma~\ref{lem:density} (Step~1) of Appendix~\ref{app:technical}.

\subsection{Wealth and Inventory Dynamics}
\label{subsec:wealth}

Given a strategy $\pi \in \Pi$, 
  the market maker's cash holdings $ X = (X_t)$ and inventory $q = (q_t)$ evolve according to
\begin{align}
    dX_t &= (S_t + \delta^a_t) \, dN^a_t - (S_t - \delta^b_t) \, dN^b_t + dC_t, \label{eq:cash} \\
    dq_t &= dN^b_t - dN^a_t, \label{eq:inventory}
\end{align}
with initial conditions $X_0 \in \R$ and $q_0 \in \R$.
Equation~\eqref{eq:cash} says that each fill at the ask adds $S_t + \delta^a_t$ to the cash account (the market maker sells at her ask price), 
each fill at the bid subtracts  $S_t - \delta^b_t$ (the market maker buys at her bid price), 
and the increment $dC_t$ represents the market maker's cash injection or withdrawal at time $t$. 
Equation~\eqref{eq:inventory} says that fills at the bid increase inventory by one unit and fills at the ask decrease it by one unit.
We are working in units in which the lot size is one share. 
The extension to arbitrary lot sizes is mechanical,  but notationally heavy.
Note that the two processes are coupled: 
 the inventory $q_t$ is determined entirely by the fill history, while the cash $X_t$ depends both on the fill history and on the prevailing mid-price at each fill time.
In particular, a sequence of round-trip trades (a bid fill followed by an ask fill at the same price) leaves inventory unchanged but contributes positively to cash, reflecting the earned spread.
A directional inventory accumulation contributes to cash at each fill price  but leaves the market maker exposed to subsequent mid-price moves.
The mark-to-market wealth at time $t$ is
\begin{equation*}
    W_t  =  X_t + q_t \, S_t.
\end{equation*}
 In a frictionless market in which liquidation occurs at the mid-price, the mark-to-market wealth $W_t$ is the natural measure of the maker's instantaneous value.
However, $W_t$ does not account for the fact that liquidating inventory $q_t$ at the prevailing mid-price is generally not feasible: 
market impact, the bid-ask spread of the market into which one liquidates, 
and the time required to unwind a large position all reduce the actually-realizable value below $W_t$.
We capture this fact through the liquidation cost function.

\subsection{The Liquidation Cost Function}
\label{subsec:liquidation}

Let $L \colon \R \to \R_+$ denote the  \emph{liquidation cost function}, where $L(q)$ is the cash cost of unwinding inventory $q$ at the terminal time $T$.
We require the following two properties.
\begin{enumerate}
    \item \emph{No cost at zero inventory.} 
    $L(0) = 0$.
    \item \emph{Convexity.} 
    $L$ is convex on $\R$.
\end{enumerate}
A third property, \emph{symmetry} ($L(q) = L(-q)$),  holds in settings where long and short positions are interchangeable (for example, a stock with no short-sale constraint).
We mention it for completeness but do not require it for any of our structural results.

The function $L(q)$ represents the cost (in cash units) of liquidating inventory $q$ at the terminal time.
Convexity captures the standard market-impact intuition that larger inventory positions are disproportionately more expensive to unwind, and is compatible with most explicit market-impact models in the literature (linear, square-root, power-law).
In this paper we do not need an explicit functional form for $L$, since the structural results depend only on the convexity and on the boundary condition $L(0) = 0$.

Given the liquidation cost function $L$, we define the \emph{liquidation-adjusted terminal wealth} of strategy $\pi$ as
\begin{equation*}
    \WL_T(\pi) \coloneqq  X_T(\pi) + q_T(\pi) \, S_T - L\bigl(q_T(\pi)\bigr), 
\end{equation*}
where the superscript $L$ stands for ``liquidation.''
This is the random variable in terms of which the market maker's preferences are expressed.

\begin{remark}[Boundedness of $\WL_T$]
\label{rem:WL-bounded}
Under the admissibility conditions of Definition~\ref{def:admissible} (see Appendix~\ref{app:technical}),
  every strategy $\pi \in \Pi$ satisfies $\WL_T(\pi) \in \Linf(\F_T)$; 
  the formal statement and proof are recorded as Lemma~\ref{lem:WL-bounded} in the appendix.
The boundedness is the regime in which the dynamic-risk-measure machinery of \citet{KupperSchachermayer2009} applies directly.
It can however be relaxed, at the cost of additional integrability conditions on the strategy space. 
We believe that the same uniqueness conclusion holds in a Cheridito--Li-style Orlicz-space treatment \citep{CheriditoLi2008}, 
although we do not verify this here.
In practice, the bound is innocuous, 
  since every realistic intraday market-making strategy faces hard inventory and loss limits.
\end{remark}

\begin{remark}[Deterministic Liquidation Cost Function]
\label{rem:L-deterministic}
We have taken $L$ to be a deterministic function of terminal inventory. 
This is the standard AS/CJ assumption.
A more realistic model takes $L \colon \R \times \Omega \to \R$ with  $L(q, \cdot)$ an $\F_T$-measurable random function reflecting the realized depth and microstructure of the book at terminal time.
Our axiomatic machinery operates on $\WL_T$ as a random variable, not on the internal structure of $L$, so the structural results survive this generalization unchanged provided $\WL_T \in \Linf(\F_T)$.
The forced coefficient $\alpha = \tfrac{1}{2}L''(0)$ would then simply become an $\F_T$-measurable random variable.
\end{remark}

\subsection{The Preference Functional}
\label{subsec:J-def}

The market maker's dynamic preference functional is a family
\begin{equation*}
    J = (J_t)_{t \in [0,T]}, \quad J_t \colon \Pi \to \Linf(\F_t),
\end{equation*}
in the sense of the dynamic risk measure literature  \citep{FollmerSchied2016, CheriditoDelbaenKupper2006, AcciaioPenner2011}.
Here $J_t(\pi)$ is the conditional preference value of strategy $\pi$ as seen at time $t$, expressed as an $\F_t$-measurable random variable in cash units.
Observe that a higher $J_t$ is better: 
$\pi$ is preferred to $\pi'$ when $J_t(\pi) \geq J_t(\pi')$ almost surely.
The associated risk measure $\rho_t \coloneqq -J_t$ also takes values in $\Linf(\F_t)$; we use both $J$ and $\rho$ as needed.

\subsection{Inputs, Choice, and Outputs}
\label{subsec:inputs-outputs}

 Before stating the axioms, we fix the framework's inputs and outputs.
The \emph{market inputs}, fixed throughout the paper, are the objects defined in Sections~\ref{subsec:price-flow}--\ref{subsec:liquidation}: the filtered probability space, horizon $T$, mid-price process $S$ and its quadratic variation, order-flow processes $N^a, N^b$, admissible strategy space $\Pi$, wealth and inventory dynamics, and liquidation cost function $L$.
These are properties of the market, not choices the agent makes.
The \emph{agent's choice} is the dynamic preference functional $J$, defined in Section~\ref{subsec:J-def}.
This is the one object on which she has genuine freedom, and the axioms of Section~\ref{sec:axioms} constrain only $J$.
The conceptual force of this split is that any test of the framework can be attributed cleanly to one side or the other.

The main result, Theorem~\ref{thm:forced-uniqueness}, says that under five axioms on $J$, the agent's choice collapses to a one-parameter family: $J$ is the entropic certainty-equivalent on $\WL_T$, parametrized by a single positive scalar $\gamma$.
The scalar $\gamma$ is the only preference-side degree of freedom that survives the axioms.
The two most consequential downstream quantities are the running inventory coefficient $\phi = \gamma\sigma^2/2$ (Section~\ref{sec:supporting}) and the terminal penalty coefficient $\alpha = \tfrac{1}{2}L''(0)$ (Corollary~\ref{cor:forced-alpha}).
The first is a preference parameter, masquerading in the Cartea--Jaimungal tradition as a free running cost but in fact pinned by $\gamma$ and the realized volatility.
The second is a market parameter, masquerading in the same tradition as a preference parameter but in fact pinned by the curvature of $L$ at zero inventory.
A practitioner who calibrates $(\phi, \alpha, \gamma)$ as three independent free parameters is therefore over-parametrizing 
on both sides of the preference--market boundary at once.


\section{Axioms for the Market Maker's Objective}
\label{sec:axioms}

In this section we lay out the five axioms which together pin down the market maker's dynamic preference functional, 
and we record three additional natural properties (monotonicity (M), the wealth-summary property (W), 
and right-continuity in time (R)) that are derivable from the five axioms.
We state each axiom precisely, motivate it from the perspective of a working trader, 
and contrast it with what the existing literature
 (i.e., the Avellaneda--Stoikov and Cartea--Jaimungal traditions) assumes either implicitly or as a primitive.

The list of axioms is, roughly, the following.
The first three (J1, J2, and J3) encode standard monetary properties of the preference functional: 
cash-additivity, normalization at zero terminal wealth, and the diversification preference (concavity).
The fourth (J4) is dynamic consistency, in the strong sense.
The fifth (J5) is law-invariance, that is, the requirement that the market maker's 
ranking depends only on the distribution of liquidation-adjusted terminal wealth.
A separate relevance condition, in the sense of \citet{KupperSchachermayer2009}, is not part of our axiom system;
  it is a derivable property, recorded in Section~\ref{subsec:relevance} below and invoked when the 
  representation theorem of Kupper--Schachermayer is applied in the proof of Theorem~\ref{thm:forced-uniqueness}.

In addition, three natural properties are derivable from~(J1)--(J5) rather than independent axioms: \emph{monotonicity}
 (in liquidation-adjusted terminal wealth), the \emph{wealth-summary property} (all dependence of $J_t$ on the strategy is mediated by $\WL_T$), and \emph{right-continuity in time} (the path $t \mapsto J_t(\pi)$ is right-continuous in probability).
The derivations are recorded in Section~\ref{subsec:derived-properties} below, 
with full proofs deferred to Appendix~\ref{app:indep}.
We discuss each of these derived properties with  the same care given to the five primitive axioms, 
since their intuitive content is just as essential to the framework.

Recall from Section~\ref{subsec:inputs-outputs} that the market data is fixed.
The axioms below constrain only the agent's preference functional $J = (J_t)_{t \in [0,T]}$.

\subsection{Cash-Additivity, Normalization, and Concavity}
\label{subsec:J1J2J3}

The first three axioms are basic monetary and diversification properties shared with most monetary risk measures: 
scale on cash, scale at zero, and aversion to dispersion.
For axiomatic treatments of dynamic monetary risk measures in adjacent settings, 
see also \citet{CerreiaVioglioMaccheroniMarinacciMontrucchio2011} and \citet{DrapeauKupper2013}; 
these works characterize broader classes of risk measures, 
whereas the present paper forces a one-parameter family by combining law-invariance 
and strong dynamic consistency under the market making setup.

\begin{axiom}[Cash-Additivity]
\label{ax:J1}
For all $t \in [0,T]$, all $\pi \in \Pi$, and all $\F_t$-measurable bounded $c$, the cash-injected strategy $\pi^{[c]}$ of Section~\ref{subsec:strategy} satisfies
   $ J_t(\pi^{[c]}) = J_t(\pi) + c$
almost surely. 
\end{axiom}

Axiom~\ref{ax:J1} expresses the basic monetary character of the preference functional.
The value $J_t$ is measured in the same units as wealth, 
and adding a sure amount of cash to the book translates one-for-one into the preference value.
In the risk-measure literature, this property is variously called cash-invariance or also translation-invariance \citep{ArtznerDelbaenEberHeath1999, FollmerSchied2016}.
The conditional form, with $c$ an $\F_t$-measurable random variable, encodes a zero-discount-rate assumption: 
cash is the num\'eraire, with no stochastic discounting between $t$ and $T$.
This is the standard convention in the AS/CJ literature;
 in a multi-period economy with a stochastic interest rate, 
 Axiom~\ref{ax:J1} would need modification.

\begin{axiom}[Normalization]
\label{ax:J2}
For every $t \in [0,T]$ and every $\pi_0 \in \Pi$ with $\WL_T(\pi_0) = 0$ almost surely, 
it holds that
$    J_t(\pi_0) = 0 $
almost surely. 
\end{axiom}

Axiom~\ref{ax:J2} fixes the scale of the preference functional. 
A strategy whose liquidation-adjusted terminal wealth is identically zero has zero cash-equivalent value at every time.
The axiom is a standard normalization in the dynamic risk measure literature, 
corresponding to the requirement $\rho_t(0) = 0$ in \citet[Definition 1.8]{KupperSchachermayer2009}.
It is independent of cash-additivity.
 Axiom~\ref{ax:J1} fixes how $J_t$ responds to cash shifts, 
but says nothing about the absolute level of $J_t$ on any one strategy.

Before stating the next axiom, we introduce a piece of notation.
Given $\pi, \pi' \in \Pi$ and an $\F_t$-measurable $\lambda \in [0, 1]$,  we write $\lambda\pi \oplus (1-\lambda)\pi'$ for any strategy in $\Pi$ whose liquidation-adjusted terminal wealth equals $\lambda \WL_T(\pi) + (1-\lambda)\WL_T(\pi')$ almost surely.
Existence of at least one such strategy is guaranteed by Lemma~\ref{lem:density} in Appendix~\ref{app:technical}.
Multiple such strategies may exist, and Axiom~\ref{ax:J3} should be read as a constraint on $J_t$ holding for every such representative.
Under the wealth-summary property (W) (Proposition~\ref{prop:wealth-summary}, derived from Axioms~\ref{ax:J1}, \ref{ax:J2}, and~\ref{ax:J4} below), all representatives receive the same $J_t$-value, and the constraint collapses to a single condition.

\begin{axiom}[Concavity]
\label{ax:J3}
\emph{(J3a) Weak inequality.}
For all $t \in [0,T]$, all $\pi, \pi' \in \Pi$, and all $\F_t$-measurable $\lambda$ with $0 \leq \lambda \leq 1$,
\begin{equation*}
    J_t(\lambda \pi \oplus (1-\lambda)\pi')  \geq  \lambda J_t(\pi) + (1-\lambda)J_t(\pi') \quad \text{almost surely.}
\end{equation*}
\emph{(J3b) Strict inequality at $t = 0$.}
For all $\pi, \pi' \in \Pi$ such that $\WL_T(\pi) - \WL_T(\pi')$ is not almost-surely equal to a deterministic constant, 
and all deterministic $\lambda \in (0, 1)$,
\begin{equation*}
    J_0(\lambda \pi \oplus (1-\lambda)\pi')  >  \lambda J_0(\pi) + (1-\lambda)J_0(\pi').
\end{equation*}
\end{axiom}

Axiom~\ref{ax:J3} is the formal statement of risk aversion.
The weak concavity inequality says that the value of the average payoff is at least the average of the values; equivalently, mixing two strategies is weakly preferred to the convex combination of the two valuations.
 This is the standard utility-theoretic expression of preference for diversification, applied to the cash-equivalent rather than to an underlying utility function.
The strictness clause is stated at $t = 0$ with  deterministic $\lambda$ for the following reason.
At $t = 0$ with $\F_0$ $\Prob$-trivial, every $\F_0$-measurable $\lambda$   is deterministic anyway, so the distinction between  
 ``deterministic $\lambda$'' and ``$\F_0$-measurable $\lambda$'' collapses.
The clause is restricted to $t = 0$ to match the scope of law-invariance (J5), which is itself stated only at $t = 0$.
The role of the clause is to rule out the risk-neutral limit ($\gamma = 0$, conditional expectation) of the entropic family that emerges in Section~\ref{sec:theorem}.

The asymmetry between J3a (stated for all $\F_t$-measurable $\lambda \in [0,1]$ and all $t$) 
and J3b (deterministic $\lambda \in (0,1)$, $t = 0$ only) is deliberate.
Statement J3a is the strong conditional form, needed to verify the convexity hypothesis of the Kupper--Schachermayer representation theorem in its full conditional generality.
Statement J3b is the minimal strict-concavity requirement that rules out the boundary case $\gamma = 0$; 
it suffices at $t = 0$ because the Kupper--Schachermayer theorem  identifies $\gamma$ from a single nondegenerate evaluation
 (Sub-step~3b of the proof of Theorem~\ref{thm:forced-uniqueness}), 
 and once $\gamma > 0$ is established at $t = 0$,
  the strict concavity at general $t$ follows automatically from the explicit entropic form.
Stating J3b with $\F_t$-measurable $\lambda$ for general $t$ would be a stronger axiom, but it is derivable from the rest of the system once the entropic representation is in hand (the entropic functional is strictly concave at every $t$ on $\F_t$-conditionally-nonconstant pairs for any $\F_t$-measurable $\lambda \in (0, 1)$).
The sharper open question is whether this strengthening at general $t$ follows from a proper subset of the axioms not involving (J5), which we do not pursue here.

A note on the sign convention.
Axiom~\ref{ax:J3} states the concavity of $J_t$, which is the standard expression of risk aversion in utility-theoretic terms; 
on $\rho_t = -J_t$ this becomes convexity of the risk measure \citep{FollmerSchied2016}.
The two are equivalent, and we move between them as the literature being discussed requires.

\subsection{Strong Dynamic Consistency}
\label{subsec:J4}

We turn now to the dynamic consistency requirement, which ensures that the market maker's preferences at time $s$ are consistent with her preferences at any later time  $t \geq s$.

\begin{axiom}[Strong Dynamic Consistency]
\label{ax:J4}
For all $s, t \in [0,T]$ with $s \leq t$ and all $\pi, \pi' \in \Pi$, if $J_t(\pi) \geq J_t(\pi')$ almost surely, then $J_s(\pi) \geq J_s(\pi')$ almost surely.
\end{axiom}

Of all the axioms, Axiom~\ref{ax:J4} is the one which most directly drives the uniqueness conclusion of Theorem~\ref{thm:forced-uniqueness}.
It says that if at some future time $t$, in every state of the world, the market maker prefers $\pi$ to $\pi'$, 
then she must also prefer $\pi$ to $\pi'$ at every earlier time $s \leq t$.
A rational agent cannot, at time $s$, knowingly anticipate that she will (uniformly) prefer $\pi$ to $\pi'$ at time $t$, 
and yet rank them in the reverse order at time $s$.
If she did, she would be planning to renegotiate her own preferences against herself.

The literature distinguishes between several notions of dynamic consistency.
The version above is the \emph{strong} version, sometimes also called rectangularity or recursive consistency.
Weaker versions are also studied, e.g., acceptance- and rejection-consistency in the sense of \citet{Weber2006}.
Of these, only the strong version forces the uniqueness conclusion we shall obtain.
Many practical risk frameworks (including the Basel-style CVaR framework discussed in Section~\ref{subsec:cvar}) operate with something weaker.
Our position is that strong dynamic consistency is the most 
appropriate rationality requirement for an inventory market maker.
A time-$s$ ranking disagreeing with the time-$t$ ranking on the same information structure is indeed the Strotzian time-inconsistency of \citet{StrotzMyopic1955}, which leaves money on the table relative to time-consistent agents.
A reader who prefers a weaker form may read the forcing conclusions as conditional on the choice of (J4), but should note that the resulting framework commits to a time-inconsistent agent.

Under the other axioms, Axiom~\ref{ax:J4} is equivalent to the recursive Bellman-type identity
\begin{equation*}
    J_s(\pi)  =  J_s\bigl(J_t(\pi)\bigr), \qquad s \leq t,
\end{equation*}
where the right-hand side is interpreted as the time-$s$ certainty-equivalent of the time-$t$ certainty-equivalent of $\pi$, viewed as a (possibly random) bounded cash amount.
The forward direction (Axiom~\ref{ax:J4} implies the Bellman identity) is verified at $s = 0$ in Step~2 of the proof of Theorem~\ref{thm:forced-uniqueness} below, which is all the proof requires.
The full $s \leq t$ version then follows from the entropic representation by the tower property of conditional expectations.
The corresponding discrete-time statement appears as \citet[Theorem 11.18]{FollmerSchied2016}.
In continuous time the equivalence holds under the additional regularity conditions of right-continuity and $J_t(\pi) \in \Linf(\F_t)$, both of which the present framework provides.
This is the form in which the axiom is most directly comparable to the time-consistency condition of \citet{KupperSchachermayer2009}.

\subsection{Law-Invariance}
\label{subsec:J5}

The last of the five core axioms is  law-invariance, 
asking that the market maker's ranking of strategies depend only on the unconditional probability  
 law of the liquidation-adjusted terminal wealth $\WL_T$.

\begin{axiom}[Law-Invariance]
\label{ax:J5}
For all $\pi, \pi' \in \Pi$, if $\WL_T(\pi) \eqdist \WL_T(\pi')$, then $J_0(\pi) = J_0(\pi')$ almost surely.
\end{axiom}

Axiom~\ref{ax:J5} expresses that the market maker's ranking of strategies depends only on the distribution of liquidation-adjusted terminal wealth.
Two strategies that produce the same probability law for $\WL_T$ are ranked identically, 
regardless of any other features that distinguish them, such as inventory sample paths, trade timing, or counterparty identity.
The condition is studied in the static setting by \citet{Kusuoka2001} and in the dynamic setting by \citet{JouiniSchachermayerTouzi2006} and \citet{KupperSchachermayer2009}.
The axiom is stated at $t = 0$, matching the form of \citet[Definition 1.8]{KupperSchachermayer2009}; 
since $\F_0$ is $\Prob$-trivial, the unconditional and time-$0$ conditional laws of $\WL_T$ coincide, 
 making the statement well-defined.
A version at general $t > 0$ would be too strong:
 the conditional value $J_t(\pi)$ depends on the conditional law of $\WL_T(\pi)$ given  $\F_t$, 
 which is generally not determined by the unconditional law.
Combined with cash-additivity, concavity, and dynamic consistency, law-invariance is the missing 
ingredient that pins down the entropic form among all dynamic monetary convex risk measures.

\begin{remark}[The Knightian-Ambiguity Qualification]
\label{rem:law-inv-trader}
Working traders typically endorse Axiom~\ref{ax:J5} on reflection, with one common qualification:
 ``My preferences are law-invariant given the model, but my ambiguity about the model is not.''
This is the Knightian-ambiguity objection, and it is well taken. 
We treat it as strictly outside the scope of the present paper and refer to \citet{GilboaSchmeidler1989} or \citet{MaccheroniMarinacciRustichini2006} for the relevant framework.
A recent extension of the law-invariance machinery to settings where law-invariance holds only on a  sub-$\sigma$-algebra of $\F_T$ is developed by \citet{ShenVanOostenWang2025}, 
which may provide a natural bridge between the strictly law-invariant axiomatization of this paper and the model-ambiguity reading of the objection.
\end{remark}

\subsection{Relevance}
\label{subsec:relevance}

We close the discussion of the axiom system with the relevance condition, in the sense of \citet{KupperSchachermayer2009}.
This condition is not part of our five core axioms; it is, in fact, a consequence of those axioms,  as we prove in Lemma~\ref{lem:relevance-derivable} below.
We state it explicitly here because it appears as a separate hypothesis in the Kupper--Schachermayer representation theorem, which we invoke in the proof of Theorem~\ref{thm:forced-uniqueness}.

\begin{definition}[Relevance]
\label{def:relevance}
The dynamic preference functional $J$ is \emph{relevant} if, for every $\F_T$-measurable event $A$ with $\Prob(A) > 0$, every $\varepsilon > 0$, and every pair of strategies $\pi_0, \pi_1 \in \Pi$ with
\begin{equation*}
\WL_T(\pi_1)  =   \WL_T(\pi_0) - \varepsilon \mathbf{1}_A \quad \text{almost surely,}
\end{equation*}
it holds that 
\begin{equation*}
J_0(\pi_1)  <  J_0(\pi_0) \quad \text{almost surely.}
\end{equation*}
\end{definition}

The condition is stated at $t = 0$, matching the formulation of \citet[Definition 1.8]{KupperSchachermayer2009}.
Existence of strategies $\pi_0, \pi_1 \in \Pi$ satisfying the wealth  identity is guaranteed by the closure of $\Pi$ under bounded $\F_T$-measurable shifts (Lemma~\ref{lem:density} of Appendix~\ref{app:technical}).
Note that under dynamic consistency (Axiom~\ref{ax:J4}), the $t=0$ condition propagates to every later $t$, and so the single-time and all-time formulations are equivalent.

Relevance expresses, plainly, that the market maker cannot be indifferent to a positive-probability loss; 
an agent who fails relevance has a preference functional that registers no welfare cost from a loss that occurs with positive probability, regardless of how large or how likely that loss is.
A market maker who fails relevance is consequently not merely bold or aggressive.
She is systematically blind to part of the probability space in forming her preferences.
The canonical example is the worst-case functional  $J_0(\pi) = \operatorname{ess\,inf}\WL_T(\pi)$,  
 which is indifferent to losses on any event that does not contain the worst-case scenario.
Indeed, from the perspective of a rational agent operating in a genuinely stochastic environment, indifference to a positive-probability loss is difficult to defend.
It amounts to assigning zero decision weight to an outcome that will materialize 
  with positive frequency over repeated interactions with the market.
Relevance is therefore best understood not as a technical regularity condition but as a minimal rationality requirement, ruling out preference functionals that are, in an economically meaningful sense,  blind to part of the probability space.

\begin{lemma}[Relevance Is Derivable from the Axioms]
\label{lem:relevance-derivable}
Any dynamic preference functional $J$ satisfying Axioms~\ref{ax:J1}, \ref{ax:J2}, \ref{ax:J3}, and~\ref{ax:J4} is relevant in the sense of Definition~\ref{def:relevance}.
Specifically, the derivation uses the strict-concavity clause~(J3b)  of Axiom~\ref{ax:J3} together with the monotonicity property~(M) of Proposition~\ref{prop:monotonicity}, which itself follows from Axioms~\ref{ax:J1}, \ref{ax:J2}, and~\ref{ax:J4}.
\end{lemma}

\begin{proof}
Fix an $\F_T$-measurable event $A$ with $\Prob(A) > 0$, $\varepsilon > 0$, and $\pi_0, \pi_1 \in \Pi$ satisfying $\WL_T(\pi_1) = \WL_T(\pi_0) - \varepsilon\mathbf{1}_A$ almost surely.
Set $W \coloneqq \WL_T(\pi_0)$ and $W' \coloneqq \WL_T(\pi_1) = W  - \varepsilon\mathbf{1}_A$, so $W \geq W'$ almost surely with $\Prob(W > W') = \Prob(A) > 0$.

The boundary case $\Prob(A) = 1$ is immediate.
In this case, $W - W' = \varepsilon$ is deterministic, and Axiom~\ref{ax:J1} (cash-additivity) gives $J_0(\pi_1) = J_0(\pi_0) - \varepsilon < J_0(\pi_0)$ directly.
For the remainder of the proof we therefore assume that $\Prob(A) \in (0,1)$.

Using Proposition~\ref{prop:monotonicity} (monotonicity (M), derived from Axioms~\ref{ax:J1}, \ref{ax:J2}, \ref{ax:J4}) applied to the pair $(\pi_0, \pi_1)$, we obtain that $J_0(\pi_0) \geq J_0(\pi_1)$ almost surely.
By contradiction, suppose that $J_0(\pi_0) = J_0(\pi_1)$ on a set of positive $\Prob$-measure; since $\F_0$ is $\Prob$-trivial and both quantities are $\F_0$-measurable, this means $J_0(\pi_0) = J_0(\pi_1)$ almost surely.
The random variable $W - W' = \varepsilon\mathbf{1}_A$ is not almost-surely equal to a deterministic constant (since $\Prob(A) \in (0, 1)$).

Choose a midpoint strategy $\pi_{1/2} \in \Pi$ with $\WL_T(\pi_{1/2}) = \tfrac{1}{2}W + \tfrac{1}{2}W'$, which exists by Lemma~\ref{lem:density}.
By the strict-concavity clause (J3b) of Axiom~\ref{ax:J3} with $\lambda =  1/2$,
\begin{equation*}
    J_0(\pi_{1/2})  >  \tfrac{1}{2}J_0(\pi_0) + \tfrac{1}{2}J_0(\pi_1)  =  J_0(\pi_0).
\end{equation*}
However, $\tfrac{1}{2}W + \tfrac{1}{2}W' \leq W$ a.s.\ (since $W \geq W'$),  so by Proposition~\ref{prop:monotonicity} (M), $J_0(\pi_{1/2}) \leq J_0(\pi_0).$
The two inequalities contradict each other, so $J_0(\pi_1) < J_0(\pi_0)$ almost surely.
\end{proof}

Two points follow.
First, since Lemma~\ref{lem:relevance-derivable} shows that relevance is in fact a consequence of Axioms~\ref{ax:J1}--\ref{ax:J4} alone, with law-invariance (J5) playing no role in the derivation, we do not include relevance as a separate hypothesis in the statement of Theorem~\ref{thm:forced-uniqueness} or in its corollaries.
We invoke it explicitly within the proof of Theorem~\ref{thm:forced-uniqueness}, where the representation theorem of \citet{KupperSchachermayer2009} lists relevance as one of its hypotheses, and we verify it there via Lemma~\ref{lem:relevance-derivable}.
Second, the analysis here clarifies how the two boundary cases $\gamma = 0$ and $\gamma = \infty$ are ruled out: both fail the strict concavity clause (J3b), with the linear functional ($\gamma = 0$) giving equality on every pair and the worst-case functional ($\gamma = \infty$) giving equality on any pair with overlapping essential-infimum sets.
Relevance corresponds to the $\gamma = \infty$ boundary failure specifically, captured by (J3b) together with monotonicity (M).

\subsection{Derived Properties}
 \label{subsec:derived-properties}

We close the axiom system by recording three natural properties of the preference functional that are derivable from Axioms~\ref{ax:J1}--\ref{ax:J5}.
The three properties are \emph{monotonicity}, the \emph{wealth-summary property}, and \emph{right-continuity in time}.
Each is stated as a numbered proposition, with the same motivational care given to the five core axioms; 
full proofs are deferred to Appendix~\ref{app:indep}.
The status shift is in fact purely a logical economy: 
every property below is just as much a feature of the framework as the five core axioms, 
 and the same intuitions apply.

The relevance condition (Definition~\ref{def:relevance}) is likewise derivable from~(J1)--(J5), 
as established in Section~\ref{subsec:relevance} as Lemma~\ref{lem:relevance-derivable}. 
We treat it separately from the three derived properties below because it does not have a  stand-alone economic reading of the kind that (M), (W), and (R) admit; 
it is a technical regularity condition that appears as an explicit hypothesis in the Kupper--Schachermayer representation theorem invoked in the proof of Theorem~\ref{thm:forced-uniqueness}.

\paragraph{Monotonicity (M).}
The first derived property states that the market maker's preferences are monotone with respect to liquidation-adjusted terminal wealth $\WL_T$, the cash-equivalent value of her book at time $T$ after deducting the cost of unwinding any remaining position.

\begin{proposition}[Monotonicity (M)]
\label{prop:monotonicity}
Any dynamic preference functional $J$ satisfying Axioms~\ref{ax:J1} (cash-additivity), \ref{ax:J2} (normalization), and~\ref{ax:J4} (strong dynamic consistency) also satisfies the following monotonicity property:
for all $\pi, \pi' \in \Pi$, if $\WL_T(\pi) \geq \WL_T(\pi')$ almost surely, then $J_t(\pi)  \geq J_t(\pi')$ almost surely for all $t \in [0,T]$.
\end{proposition}

\begin{proof}
See Appendix~\ref{app:indep}, Proposition~\ref{prop:monotonicity-app}.
The argument is short: 
cash-additivity (J1) together with normalization (J2), applied with the cash injection $c = -\WL_T(\pi)$ at time $T$, gives $J_T(\pi) = \WL_T(\pi)$ for every $\pi$; 
 strong dynamic consistency (J4) then propagates the trivial monotonicity at $t=T$ backwards to every $s \leq T$.
\end{proof}

The property is a basic dominance condition requiring that more cash-equivalent at the horizon, in every state of the world, must be weakly preferred.
The condition is what makes liquidation-adjusted terminal wealth a meaningful object for the market maker to optimize at all.
We take it to be uncontroversial.
It is indeed satisfied by every framework in the inventory market making literature that we are aware of, including the original Avellaneda--Stoikov setup \citep{AvellanedaStoikov2008} and the Cartea--Jaimungal--Penalva textbook \citep{CarteaJaimungalPenalva2015}.

\paragraph{The wealth-summary property (W).}
 The second derived property is the wealth-summary requirement: all dependence of $J_t$ on the strategy is mediated by $\WL_T$.
Concretely, two strategies that produce the same liquidation-adjusted terminal wealth almost surely must receive the same preference value at every time $t$.

\begin{proposition}[Wealth-Summary (W)]
\label{prop:wealth-summary}
Any dynamic preference functional $J$ satisfying Axioms~\ref{ax:J1}, \ref{ax:J2}, and~\ref{ax:J4} 
(equivalently,  any $J$ satisfying Proposition~\ref{prop:monotonicity}) also satisfies the following wealth-summary property:
for all $\pi, \pi' \in \Pi$, if $\WL_T(\pi) = \WL_T(\pi')$ almost surely, then also  $J_t(\pi) = J_t(\pi')$ almost surely for all $t \in [0,T]$.
\end{proposition}

\begin{proof}
See Appendix~\ref{app:indep}, Proposition~\ref{prop:wealth-summary-app}.
The argument is two lines: if $\WL_T(\pi) = \WL_T(\pi')$ a.s., then both $\WL_T(\pi) \geq \WL_T(\pi')$ a.s.\ and 
   $\WL_T(\pi') \geq \WL_T(\pi)$ a.s., and applying monotonicity   (M, Proposition~\ref{prop:monotonicity}) 
 in both directions yields $J_t(\pi) = J_t(\pi')$ a.s.
\end{proof}

Proposition~\ref{prop:wealth-summary} formalizes a methodological commitment.
The market maker is a cash-equivalent maximizer whose dislike of inventory derives entirely from its effect on liquidation-adjusted terminal wealth, not from inventory-as-such.
Inventory affects $\WL_T$ through two channels: the \emph{variance channel}, where holding nonzero inventory exposes the market maker to subsequent mid-price moves, and the \emph{liquidation-cost channel}, where the term $L(q_T)$ charges her for any residual position at the close.
A market maker who dislikes inventory through either or both of these channels is fully wealth-summary-compatible.
What wealth-summary excludes is a third channel, namely intrinsic discomfort with the bare fact of holding a position, independent of where it leaves her at the close.

Two natural alternatives are therefore excluded as primitives.
Path-dependent inventory aversion of the running-quadratic form $\int_0^T q_s^2 \, ds$,  
which appears as a primitive in the Cartea--Jaimungal tradition, is not excluded by fiat; our claim
 (developed in Section~\ref{sec:consequences}) is that any such running penalty, 
  if consistent with the five core axioms, must arise as a derived feature of $\WL_T$-preferences, with a forced coefficient.
Path-functional preferences in the strict sense 
(e.g., maximum-drawdown aversion or time-average wealth) are genuinely outside the scope of Proposition~\ref{prop:wealth-summary}.
We discuss this in Section~\ref{sec:discussion}, noting in particular that drawdown-averse market makers are not covered.

\begin{remark}[On the Circularity Charge]
The wealth-summary property is not assumed; it is \emph{forced} by cash-additivity, normalization, and strong dynamic consistency (Proposition~\ref{prop:wealth-summary}).
A practitioner who accepts those three core axioms is automatically committed to wealth-summary. 
The only escape is to reject one of (J1), (J2), or~(J4), each independently justified earlier in this section.
A skeptical reader might object that the wealth-summary property is circular, in that we use it to rule out the Cartea--Jaimungal objective as a primitive functional (Corollary~\ref{cor:cj-inconsistency}), and it is designed precisely to exclude such path-dependent running penalties.
 The derived-property framing answers this directly: wealth-summary is not assumed but derived from cash-additivity, normalization, and strong dynamic consistency, so any objection must reject one of those three axioms.
The framework moreover does not exclude running penalties as such. 
The running penalty $\phi\int_0^T q_s^2\, ds$ arises within it as the second-order Taylor expansion of the entropic functional (Corollary~\ref{cor:cj-inconsistency}, Proposition~\ref{prop:forced-phi}); 
what is excluded is treating such a penalty as a \emph{primitive, free-parameter} object rather than as a derived feature with a forced coefficient $\phi = \gamma\sigma^2/2$.
\end{remark}

\paragraph{Right-continuity in time (R).}
The third derived property is a regularity condition on the path $t \mapsto J_t(\pi)$.
Concretely, if $t_k \downarrow t$, then $J_{t_k}(\pi) \to J_t(\pi)$ in probability.
It plays a role only in the continuous-time-to-discrete-time passage of the proof of Theorem~\ref{thm:forced-uniqueness}. 

\begin{proposition}[Right-Continuity in Time (R)]
\label{prop:right-continuity}
Every  dynamic preference functional $J$  satisfying Axioms~\ref{ax:J1}--\ref{ax:J5} also satisfies the following right-continuity property:
for every  $\pi \in \Pi$ and every $t \in [0, T)$,  the path $s \mapsto J_s(\pi)$ is right-continuous in probability at $s = t$. 
That is, if $t_k \downarrow t$ with $t_k \in [t, T]$, then $J_{t_k}(\pi) \to J_t(\pi)$ in probability as $k \to \infty$.
\end{proposition}

\begin{proof}
See  Appendix~\ref{app:indep},  Proposition~\ref{prop:right-continuity-app}. 
The derivation proceeds in two distinct stages.
First, Axioms~\ref{ax:J1}--\ref{ax:J5} (from which relevance is derived; Lemma~\ref{lem:relevance-derivable}) force $J$ to equal the entropic certainty-equivalent on $\WL_T$ at every dyadic time $t \in \mathcal{D} \subset [0,T]$, via the Kupper--Schachermayer representation theorem applied to the restriction of $\tilde J$ to each dyadic sub-filtration, with consistency of the resulting parameter across scales.
Second, the Bellman identity (from strong dynamic consistency, Axiom~\ref{ax:J4})   together with cash-additivity (Axiom~\ref{ax:J1}), applied with bounded $\F_{t_0}$-measurable cash injections as test functions, lifts the entropic representation from dyadic times to every $t_0 \in [0,T]$.
Right-continuity then follows from the right-continuity of conditional expectations  $\E[e^{-\gamma \WL_T} \mid \F_t]$ under the usual conditions on the filtration.
\end{proof}

The information filtration $\filt$ is right-continuous and the price process $S$ is a continuous semimartingale, so the primitive processes of the setup are right-continuous in $t$.
Proposition~\ref{prop:right-continuity} (R) says that the preference functional inherits this regularity: 
  $J_t(\pi)$ does not jump at a non-information time.
 A jump in $J_t(\pi)$ at some $t$ where the underlying primitives have not jumped would be  incoherent with the cash-equivalence interpretation, since the trader could realize neither the pre-jump nor the post-jump value via any actual trade in the limit.

\subsection{Summary}
\label{subsec:axiom-summary}

We have laid out five axioms and three derived properties, which we now collect for reference.

\paragraph{The five core axioms.}
\begin{enumerate}
    \item[(J1)] Cash-additivity.
    \item[(J2)] Normalization.
    \item[(J3)] Concavity at all $t$;
       strict at $t = 0$ on pairs with nonconstant $\WL_T$-difference.
    \item[(J4)] Strong dynamic consistency.
    \item[(J5)]  Law-invariance (at $t = 0$).
\end{enumerate}

\paragraph{Three derived properties (consequences of the axioms).}
\begin{enumerate}
    \item[(M)] \emph{Monotonicity} (Proposition~\ref{prop:monotonicity}): 
    if $\WL_T(\pi) \geq \WL_T(\pi')$ a.s.,  
    $J_t(\pi) \geq J_t(\pi')$ a.s.\ for all $t$.
    \item[(W)] \emph{Wealth-summary property} (Proposition~\ref{prop:wealth-summary}): 
    $J_t(\pi)$ depends on $\pi$ only through $\WL_T(\pi)$.
    \item[(R)] \emph{Right-continuity in time} (Proposition~\ref{prop:right-continuity}): 
     the path $t \mapsto J_t(\pi)$ is right-continuous in probability.
\end{enumerate}

The axioms above are jointly satisfied by the entropic certainty-equivalent on  $\WL_T$, with any strictly positive risk-aversion parameter $\gamma$.
Theorem~\ref{thm:forced-uniqueness} below states the converse, 
namely that they are satisfied \emph{only} by this family.

We prove in Appendix~\ref{app:indep} that the five core axioms (J1)--(J5) form an \emph{independent} system, in the sense that no axiom in the list is derivable from the others.
The three derived properties 
(namely, (M) monotonicity, (W) wealth-summary, and (R) right-continuity) 
are derivable from (J1)--(J5),   with the derivations presented in that appendix. 
The relevance condition of Definition~\ref{def:relevance} is also derivable from (J1)--(J5) (Lemma~\ref{lem:relevance-derivable}); 
it is treated separately from the three derived properties because it functions as a technical hypothesis of the 
representation theorem rather than as a stand-alone property of $J$.


\section{The Forced Uniqueness Theorem}
\label{sec:theorem}

We now state and prove the main result of the paper.
The proof proceeds in three steps.
First, the wealth-summary property (W) (Proposition~\ref{prop:wealth-summary},  
  derived from Axioms~\ref{ax:J1}, \ref{ax:J2}, \ref{ax:J4}) reduces the market maker's preference functional to a functional on $\Linf(\F_T)$.
Second, this reduced functional is shown to satisfy the standard hypotheses for a law-invariant time-consistent representation.
Third, the representation result yields a one-parameter family of admissible functionals, 
 with both boundary cases ($\gamma = 0$ and $\gamma = \infty$) ruled out by the strict-concavity clause (J3b) of Axiom~\ref{ax:J3}.

\subsection{Statement of the Theorem}
\label{subsec:theorem-statement}

We are now in a position to state the main result.
Recall that $J$ denotes the market maker's dynamic preference functional.

\begin{theorem}[Forced Uniqueness]
\label{thm:forced-uniqueness}
Let the market data $(\Omega, \F, \filt, \Prob, T, S, N^a, N^b, \Pi, L)$ be as specified in Section~\ref{sec:setup}, and let $J = (J_t)_{t \in [0,T]}$ with $J_t \colon \Pi \to \Linf(\F_t)$ be a dynamic preference functional.
Then $J$ satisfies Axioms~\ref{ax:J1}--\ref{ax:J5} of Section~\ref{sec:axioms} if and only if there exists $\gamma \in (0, \infty)$ such that
\begin{equation}
    J_t(\pi)  =  -\frac{1}{\gamma}\, \log \E\!\left[\exp \bigl(-\gamma\, \WL_T(\pi)\bigr) \,\Big|\, \F_t\right]
    \label{eq:entropic}
\end{equation}
for all $t \in [0,T]$ and all $\pi \in \Pi$.
The scalar $\gamma$ is unique.
\end{theorem}

The theorem is the core forced-uniqueness result.
Under the five axioms,  the market maker's preference functional is forced to take the entropic certainty-equivalent form on $\WL_T$,
with a unique positive scalar $\gamma$ that is the coefficient of absolute risk aversion
  of the underlying utility $u(x) = -\exp(-\gamma x)$ in the sense of \citet{Pratt1964}.
The scalar $\gamma$ is,  moreover,  a property of the agent's preferences alone,
 independent of the analyst's choice of clock.
Corollary~\ref{cor:clock-invariance} states and proves this clock-invariance,
 and it underwrites the stochastic-volatility forced-coefficient result of Proposition~\ref{prop:clock-corrected}.
The conclusion depends critically on the strong form of dynamic consistency (Axiom~\ref{ax:J4}).
Under acceptance-rejection consistency in the sense of \citet{Weber2006}, the uniqueness fails:
shortfall risk measures of the form $\rho(X) = \inf\{m \in \R \mid \E[u(X + m)] \geq c\}$
 are also consistent with the weakened axiom system,
 so the one-parameter entropic family is no longer the only admissible class.
The dividing line between the two regimes is therefore the strength of the dynamic-consistency axiom, with strong consistency forcing a single risk parameter and weaker consistency admitting a richer family parametrized by the utility-and-threshold pair.
This is a caveat discussed in Section~\ref{subsec:J4}.

Before turning to the proof, we record one remark on a structural feature of the result.

\begin{remark}[Why CARA Survives the Dynamic Lift]
\label{rem:cara-dynamic}
A reader familiar with the static AS setting may note that the entropic functional is the certainty-equivalent of CARA expected utility, $-\tfrac{1}{\gamma}\log\E[\exp(-\gamma W)] = u^{-1}(\E[u(W)])$ for $u(x) = -\exp(-\gamma x)$.
In the static setting this is an identity, CARA-expected-utility maximization coinciding with entropic-functional maximization for any $\gamma$.
The nontrivial fact is that this coincidence extends to the dynamic setting; the dynamic entropic functional, defined by recursive application of the conditional certainty-equivalent, is the same object as dynamic expected-utility maximization on a CARA agent.
The structural reason CARA emerges is that its Bellman value function factorizes as $V = -\exp(-\gamma(X+qS))\,\psi(t,q)$.
In the constant-volatility benchmark, after solving the HJB equation with the CARA ansatz (Proposition~\ref{prop:forced-phi}), the associated certainty-equivalent takes the explicit form $J_t(\pi) = (X_t + q_t S_t) + \theta(t,q_t)$, depending on the strategy only through the conditional distribution of $\WL_T(\pi)$ given~$\F_t$, and not on $(X_t, q_t, S_t)$ individually.
This is precisely the wealth-summary property~(W), now visible at the level of the explicit value function.
Note that the Bellman value function $V(t,X,q,S)$ itself does depend on the full state $(X,q,S)$; what is special about CARA is this factorization, which makes the associated certainty-equivalent $J_t$ inherit the wealth-summary property.
For any other utility curvature, the Bellman value function does not factorize in this way, the dynamic programming state (cash, inventory, price) enters separately into the certainty-equivalent, and the dynamic risk-measure language and dynamic expected-utility language diverge.
This is the deepest mathematical reason CARA emerges:
  it is the only utility that makes dynamic risk-measure language and dynamic expected-utility language coincide.
\end{remark}

\subsection{Proof of Theorem~\ref{thm:forced-uniqueness}}
\label{subsec:proof}

 The technical conditions on the strategy space and the boundedness of $\WL_T$ over $\Pi$ are collected in Appendix~\ref{app:technical}, a standing companion to this section.
  In particular, Lemma~\ref{lem:WL-bounded} establishes that $\WL_T(\pi) \in \Linf(\F_T)$ for every $\pi \in \Pi$, and Lemma~\ref{lem:density} establishes that the range  $\{\WL_T(\pi) \mid \pi \in \Pi\}$  is equal to all of $\Linf(\F_T)$.
The proof has two directions.
The forward direction (axioms imply the entropic form) proceeds in three steps: 
first, we reduce the problem to a question about a functional on $\Linf(\F_T)$; 
second, we verify that this functional satisfies the axioms of \citet{KupperSchachermayer2009}; 
third, we invoke their Theorem 1.10 and translate the conclusion back into the language of our setup.
The converse direction (entropic implies the axioms) is  a direct verification of each axiom on the entropic functional. 

\paragraph{($\Rightarrow$) Forward direction: the axioms imply the entropic form.}
Assume $J$ satisfies Axioms~\ref{ax:J1}--\ref{ax:J5}.
By Lemma~\ref{lem:relevance-derivable}, $J$ is then also relevant in the sense  of Definition~\ref{def:relevance}. 
We shall invoke this fact when applying the Kupper--Schachermayer representation theorem in Step~3 below.

\paragraph{Step 1: reduction to \texorpdfstring{$\Linf$}{L-infinity}.}
By property (W) (Proposition~\ref{prop:wealth-summary}), for every $\pi \in \Pi$ the value $J_t(\pi)$ depends only on the random variable $\WL_T(\pi) \in \Linf(\F_T)$.
Define, on the range
\begin{equation*}
    \mathcal{R}  \coloneqq   \{\WL_T(\pi) \mid \pi \in \Pi\}  \subset  \Linf(\F_T),
\end{equation*}
the \emph{reduced functional}
\begin{equation*}
    \tilde{J}_t \colon \mathcal{R}  \to \Linf(\F_t), \quad \tilde{J}_t(W)    \coloneqq  J_t(\pi) \;\text{ for any } \pi \in \Pi \text{ with } \WL_T(\pi) = W.
\end{equation*}
By property (W) (Proposition~\ref{prop:wealth-summary}), $\tilde{J}_t$ is well-defined on $\mathcal{R}$.
By Lemma~\ref{lem:density} of Appendix~\ref{app:technical} (Step~3, which uses Step~1 of that lemma applied at $t = T$), the range satisfies $\mathcal{R} = \Linf(\F_T)$ exactly: every bounded $\F_T$-measurable random variable arises as $\WL_T(\pi)$ for some $\pi \in \Pi$, via terminal-cash injection.
The reduced functional $\tilde{J}_t$ is therefore defined on all of $\Linf(\F_T)$ for every $t \in [0, T]$; 
no separate extension argument is required.
Define the associated risk measure $\rho_t \coloneqq -\tilde J_t \colon \Linf(\F_T) \to \Linf(\F_t)$.

We pause to fix notation for this section.
We use $\tilde J_t \colon \Linf(\F_T) \to \Linf(\F_t)$ for the reduced functional just defined, and $\widehat J_u \colon \Pi \to \Linf(\F_{\tau(u)})$ for the business-time reparametrization $\widehat J_u(\pi) \coloneqq J_{\tau(u)}(\pi)$ used in Corollary~\ref{cor:clock-invariance} below.
The two objects are distinct and play different roles in the proof.

\paragraph{Step 2: verification of the dynamic risk measure axioms.}
We check that the family $(\rho_t)_{t \in [0,T]}$ on $\Linf(\F_T)$ satisfies the hypotheses of \citet[Definition 1.8]{KupperSchachermayer2009}: 
normalization, conditional cash-invariance, monotonicity, convexity, law-invariance (at $t = 0$), time-consistency, and relevance.

\emph{Normalization.}
By Axiom~\ref{ax:J2}, $J_t(\pi_0) = 0$ a.s.\ for  every $t \in [0, T]$ and every $\pi_0 \in \Pi$ with $\WL_T(\pi_0) = 0$ a.s.
Equivalently, $\tilde J_t(0) = 0$ a.s.\ for every $t$, i.e., $\rho_t(0) = 0$ a.s.\ for every $t$.
This is the normalization required by \citet[Definition 1.8]{KupperSchachermayer2009}.

\emph{Cash-invariance.}
For any $W \in \Linf(\F_T)$ and any $\F_t$-measurable bounded $m$,  
  Axiom~\ref{ax:J1} gives $\tilde J_t(W + m) = \tilde J_t(W) + m$, i.e., $\rho_t(W + m) = \rho_t(W) - m$.
This is the cash-invariance property of \citet{KupperSchachermayer2009}.

 \emph{Monotonicity.}
For $W, W' \in \Linf(\F_T)$ with $W \geq W'$ almost surely, 
 Proposition~\ref{prop:monotonicity} (M)  applied to any pair of strategies $\pi, \pi' \in \Pi$ with $\WL_T(\pi) = W$ and $\WL_T(\pi') = W'$ gives $J_t(\pi) \geq J_t(\pi')$ almost surely, 
 so  $\tilde J_t(W) \geq \tilde J_t(W')$ and $\rho_t(W) \leq \rho_t(W')$ almost surely.
This is the monotonicity of \citet{KupperSchachermayer2009}.

\emph{Convexity.}
For $W, W' \in \Linf(\F_T)$ and any $\F_t$-measurable $\lambda \in [0,1]$,  Axiom~\ref{ax:J3}  gives $\tilde J_t(\lambda W + (1-\lambda) W') \geq \lambda \tilde J_t(W) + (1-\lambda) \tilde J_t(W')$, and thus 
\begin{equation*}
    \rho_t(\lambda W + (1-\lambda)W')  \leq  \lambda \rho_t(W) + (1-\lambda)\rho_t(W').
\end{equation*}
This is the convexity property.

\emph{Law-invariance.}
For $W, W' \in \Linf(\F_T)$ with $W \eqdist W'$, take any $\pi, \pi' \in \Pi$ with $\WL_T(\pi) = W$ and $\WL_T(\pi') = W'$.
Then $\WL_T(\pi) \eqdist \WL_T(\pi')$, and Axiom~\ref{ax:J5} gives $J_0(\pi) = J_0(\pi')$, 
hence $\tilde J_0(W) = \tilde J_0(W')$ and $\rho_0(W) = \rho_0(W')$.
(Law-invariance is required only at time zero in \citealp{KupperSchachermayer2009}. 
The conditional version at $t > 0$ follows automatically from the other K--S hypotheses once the representation is in hand.)

\emph{Time-consistency.}
We need to verify that $\rho_0(W) = \rho_0\bigl(-\rho_t(W)\bigr)$ for all $W \in \Linf(\F_T)$ and all $t \in [0,T]$.
Equivalently in $\tilde J$-form, $\tilde J_0(W) = \tilde J_0\bigl(\tilde J_t(W)\bigr)$, 
 where $\tilde J_t(W) \in \Linf(\F_t) \subseteq \Linf(\F_T)$ is viewed as a terminal random variable on the right-hand side.
The identity follows from Axiom~\ref{ax:J4} (strong dynamic consistency),
combined with Axiom~\ref{ax:J1} (cash-additivity) and Axiom~\ref{ax:J2} (normalization),  as follows.

Given $W \in \Linf(\F_T)$, fix any $\pi \in \Pi$ with $\WL_T(\pi) = W$  (this $\pi$ exists by Lemma~\ref{lem:density}), 
so that $\tilde J_t(W) = J_t(\pi)$.
Pick any $\pi_0 \in \Pi$ with $\WL_T(\pi_0) = 0$ almost surely
  (existence again by Lemma~\ref{lem:density}).
By Axiom~\ref{ax:J2} (Normalization), $J_t(\pi_0) = 0$ almost surely for every $t \in [0, T]$.
We first verify that $J_t(\pi) \in \Linf(\F_t)$ with $\|J_t(\pi)\|_\infty \leq \|\WL_T(\pi)\|_\infty$.
Since $\WL_T(\pi)  \in \Linf(\F_T)$ by Lemma~\ref{lem:WL-bounded}, 
write $M_\pi \coloneqq \|\WL_T(\pi)\|_\infty < \infty$, so $-M_\pi \leq \WL_T(\pi) \leq M_\pi$ almost surely.
Let $\pi_0^{[\pm M_\pi]}$ denote the cash-injected strategies of Section~\ref{subsec:strategy} applied to $\pi_0$, 
which satisfy $\WL_T(\pi_0^{[\pm M_\pi]}) = \pm M_\pi$ almost surely.
By Axioms~\ref{ax:J1} and~\ref{ax:J2},  $J_t(\pi_0^{[\pm M_\pi]}) = \pm M_\pi$ almost surely.
Since $\WL_T(\pi_0^{[-M_\pi]}) \leq \WL_T(\pi) \leq \WL_T(\pi_0^{[M_\pi]})$ almost surely, 
 Proposition~\ref{prop:monotonicity} (monotonicity (M), already verified above) gives $-M_\pi \leq J_t(\pi) \leq M_\pi$ 
 almost surely, so $J_t(\pi) \in \Linf(\F_t)$ with the claimed bound.
Let $\pi^*$ denote the strategy obtained from $\pi_0$ by injecting the cash amount $J_t(\pi)$ at time $t$, 
that is, $\pi^* = \pi_0^{[J_t(\pi)]}$ in the notation of Section~\ref{subsec:strategy}.
By Axiom~\ref{ax:J1} (cash-additivity at time $t$),  $J_t(\pi^*) = J_t(\pi_0) + J_t(\pi) = 0 + J_t(\pi) = J_t(\pi)$ almost surely; 
and at the wealth level, $\WL_T(\pi^*) = \WL_T(\pi_0) + J_t(\pi) = J_t(\pi)$ almost surely.
Since $J_t(\pi^*) = J_t(\pi)$ almost surely, the equality applied as a pair of inequalities
 (each direction) under Axiom~\ref{ax:J4} (strong dynamic consistency) yields $J_0(\pi^*) = J_0(\pi)$ almost surely.
Translating to the reduced functional,
\begin{equation*}
\tilde J_0(W)  =  J_0(\pi)  =  J_0(\pi^*)  =   \tilde J_0\bigl(\WL_T(\pi^*)\bigr)  =  \tilde J_0\bigl(J_t(\pi)\bigr)  =  \tilde J_0\bigl(\tilde J_t(W)\bigr),
\end{equation*}
almost surely, 
where the last equality uses $J_t(\pi) = \tilde J_t(W)$ by definition of the reduced functional (Step~1) on $W = \WL_T(\pi)$.
This is the desired time-consistency identity.

\emph{Relevance.}
By Lemma~\ref{lem:relevance-derivable}, the five core axioms (J1)--(J5) imply that  $J$ is relevant in the sense of Definition~\ref{def:relevance}, which is identical to the relevance condition of \citet{KupperSchachermayer2009}.
We emphasize that relevance is therefore not an additional assumption on $J$, but a property derived from the five axioms. 
 It appears here as a separate item only because the Kupper--Schachermayer theorem lists it as a separate hypothesis.

\paragraph{Step 3: applying the representation result.}
By the verification in Step 2, the family $(\rho_t)$ is a law-invariant, time-consistent, relevant,  convex dynamic risk measure on $\Linf(\F_T)$.
We invoke \citet[Theorem 1.10]{KupperSchachermayer2009}, which addresses the discrete-time case
 (with the filtration indexed by the nonnegative integers).
The continuous-time extension is the main work of this step.
We present below the more direct route, which uses property~(R) (Proposition~\ref{prop:right-continuity}) to lift the dyadic representation to all of $[0,T]$ via L\'evy's martingale convergence theorem.
An alternative, logically independent route is given in Stage~1 of the proof of Proposition~\ref{prop:right-continuity-app} in Appendix~\ref{app:indep}, which establishes the entropic representation at every $t \in [0,T]$ directly from (J1)--(J5) by a Bellman + cash-additivity + conditional-MGF argument,  with~(R) emerging as a consequence rather than a hypothesis.
 The two routes lead to the same conclusion; we adopt the (R)-based route here because it isolates the role of regularity in the lift.
The conclusion delivered by this step is
\begin{equation*}
    \rho_t(W)  =  \frac{1}{\gamma}\, \log \E[\exp(-\gamma W) \mid \F_t], \qquad t \in [0, T], \; W \in \Linf(\F_T),
\end{equation*}
for a unique $\gamma \in (0, \infty)$.
The argument proceeds in three sub-steps.

\emph{Sub-step 3a: Dyadic restrictions and  inheritance of axioms.}
Let $\filt^{(n)} = (\F_{kT/2^n})_{k=0}^{2^n}$ be the dyadic sub-filtration of $\filt$ at level $n$.
The restriction of $\tilde J$ to $\filt^{(n)}$, denoted $\tilde J^{(n)}$, is the family $(\tilde J^{(n)}_k)_{k=0}^{2^n}$ with $\tilde J^{(n)}_k(W) \coloneqq \tilde J_{kT/2^n}(W)$ for $W \in \Linf(\F_T)$.
We claim that $\tilde J^{(n)}$ inherits Axioms~\ref{ax:J1}--\ref{ax:J5} as well as the derived properties of Proposition~\ref{prop:monotonicity} (monotonicity (M)) and Proposition~\ref{prop:wealth-summary} (wealth-summary (W)) directly from $\tilde J$, since each is stated as a property of $\tilde J_t$ for arbitrary $t \in [0,T]$ and continues to hold when $t$ is restricted to $\{kT/2^n\}_{k=0}^{2^n}$.
Relevance (Definition~\ref{def:relevance}) is a single condition at $t = 0$, and is inherited by $\tilde J^{(n)}$ because $0 \in \{kT/2^n\}_{k=0}^{2^n}$ for every $n$.
Proposition~\ref{prop:right-continuity} (right-continuity (R) in time) is trivial on a finite discrete time-set and does not need separate verification on $ \tilde J^{(n)}$.
Specifically, dynamic consistency (J4) in the discrete-time form  $\tilde J^{(n)}_j(W) = \tilde J^{(n)}_j(\tilde J^{(n)}_k(W))$ for $j \leq k$ is an immediate consequence of strong dynamic consistency  (Axiom~\ref{ax:J4}) at the continuous-time level.
Hence, $\tilde J^{(n)}$ is a discrete-time normalized, cash-invariant, monotone, convex, law-invariant (at $t = 0$), time-consistent, and relevant dynamic risk measure on $\Linf(\F_T)$, satisfying the hypotheses of \citet[Theorem 1.10]{KupperSchachermayer2009}.
The theorem yields $\gamma_n \in [0, \infty]$ such that
\begin{equation*}
    \tilde J^{(n)}_k(W)  =  -\tfrac{1}{\gamma_n}\log\E[e^{-\gamma_n W} \mid \F_{kT/2^n}] \quad \text{for all } k \in \{0, 1,  \ldots, 2^n\},
\end{equation*}
with the two boundary   conventions  given by 
 $\tilde J^{(n)}_k(W) = \E[W \mid \F_{kT/2^n}]$ at $\gamma_n = 0$ and $\tilde J^{(n)}_k(W) = \operatorname{ess\,inf}(W \mid \F_{kT/2^n})$ at $\gamma_n = \infty$.
The strict-concavity clause (J3b) of Axiom~\ref{ax:J3} rules out $\gamma_n = 0$: 
at $\gamma_n = 0$ and $k = 0$, the representation reads $\tilde J^{(n)}_0(W) = \E[W]$
 (since $\F_0$ is $\Prob$-trivial), which is linear in $W$; for any $W, W'$ with $W - W'$ not a.s. constant and any deterministic $\lambda \in (0,1)$, $\E[\lambda W + (1-\lambda)W'] = \lambda \E[W] + (1-\lambda)\E[W']$ gives equality, contradicting (J3b).
The relevance condition rules out $\gamma_n = \infty$: 
at $\gamma_n = \infty$ and $k = 0$, the representation reads $\tilde J^{(n)}_0(W) = \operatorname{ess\,inf}(W)$.
Pick any nondegenerate $W \in \Linf(\F_T)$
 (which exists by the non-atomicity of $\F_T$) 
  and choose an event $A \in \F_T$ of positive probability with $A \subseteq \{W > \operatorname{ess\,inf}(W) + \delta\}$ for some $\delta > 0$; 
  such an $A$ exists by definition of essential infimum.
For any $\varepsilon \in (0, \delta]$, the shift $W - \varepsilon\mathbf{1}_A$ then satisfies
   $\operatorname{ess\,inf}(W - \varepsilon\mathbf{1}_A) = \operatorname{ess\,inf}(W)$,
     since $W - \varepsilon\mathbf{1}_A \geq W - \delta \geq \operatorname{ess\,inf}(W)$ on $A$ and $W - \varepsilon\mathbf{1}_A = W$ on $A^c$.
Hence,  we have $\tilde J^{(n)}_0(W - \varepsilon\mathbf{1}_A) = \tilde J^{(n)}_0(W)$,  contradicting Definition~\ref{def:relevance} (which requires strict inequality for every $\varepsilon > 0$, in particular for this $\varepsilon$).
Therefore, $\gamma_n \in (0, \infty)$.

\emph{Sub-step 3b: Consistency of $\gamma_n$ across scales.}
The claim is that $\gamma_n = \gamma_{n+1}$ for every $n$.
The argument is by identifiability of the entropic parameter from a single nondegenerate evaluation.
Both restrictions of $\tilde J$  to $ \filt^{(n)}$ and $\filt^{(n+1)}$ agree at every common time index, in particular at $t = 0$.
Since $\tilde J_0$ is a single well-defined object  (the restriction of the same functional $\tilde J$ to time $0$, independent of the chosen dyadic scale), both Kupper--Schachermayer representations must reproduce it: 
for every $W \in \Linf(\F_T)$,
\begin{equation*}
-\tfrac{1}{\gamma_n}\log\E[e^{-\gamma_n W}]     =  \tilde J_0(W)  =  -\tfrac{1}{\gamma_{n+1}}\log\E[e^{-\gamma_{n+1} W}],
\end{equation*}
equivalently,
\begin{equation*}
\tfrac{1}{\gamma_n}\log\E[e^{-\gamma_n W}]  =  \tfrac{1}{\gamma_{n+1}}\log\E[e^{-\gamma_{n+1} W}]
\end{equation*}
for every $W \in \Linf(\F_T)$.
Pick any nondegenerate $W \in \Linf(\F_T)$.
Such $W$ exists by the non-atomicity of $(\Omega, \F_T, \Prob)$   
  assumed in Section~\ref{sec:setup}: 
  there exists an event $A \in \F_T$ with $\Prob(A) \in (0,1)$, and the random variable $W = \mathbf{1}_A - \mathbf{1}_{A^c}$ is nondegenerate.
Let $\kappa(\gamma) \coloneqq \log\E[e^{-\gamma W}]$ denote its cumulant-generating function and consider the entropic certainty-equivalent function $f(\gamma) \coloneqq -\tfrac{1}{\gamma}\kappa(\gamma)$ on $(0,\infty)$.
The agreement of the two representations at $t = 0$ on this $W$ reads $f(\gamma_n) = f(\gamma_{n+1})$, 
so it suffices to show that $f$ is strictly monotone on $(0, \infty)$.
A direct calculation gives
\begin{equation*}
    f'(\gamma)  =  \frac{1}{\gamma^2}\kappa(\gamma) - \frac{1}{\gamma}\kappa'(\gamma)  =  -\frac{1}{\gamma^2}\bigl(\gamma\kappa'(\gamma) - \kappa(\gamma)\bigr).
\end{equation*}
The quantity in parentheses is exactly the relative entropy of the entropic-tilted measure
   $\Q$ (defined by $d\Q/d\Prob = e^{-\gamma W}/\E[e^{-\gamma W}]$) with respect to $\Prob$.
Indeed, with $\kappa'(\gamma) = -\E_\Q[W]$, so that $\E_\Q[-\gamma W] = \gamma\kappa'(\gamma)$, 
we have 
\begin{equation*}
    H(\Q \re \Prob)  =  \E_\Q\!\left[\log\frac{d\Q}{d\Prob}\right]  =  \E_\Q[-\gamma W - \kappa(\gamma)]  =  \gamma\kappa'(\gamma) - \kappa(\gamma)  \geq  0,
\end{equation*}
with equality if and only if $W$ is $\Prob$-almost-surely constant.
Since $W$ is nondegenerate, $H(\Q  \re  \Prob) > 0$.
Thus, $f'(\gamma) < 0$ on $(0,\infty)$, so $f$ is strictly monotone decreasing on $(0,\infty)$ for any nondegenerate $W$.
The equation $f(\gamma) = c$ thus has at most one solution in $(0,\infty)$ for any $c$ in its range.
Hence, $\gamma_n  = \gamma_{n+1}$, and we denote the common value by $\gamma$.

\emph{Sub-step 3c: Extension to all of $[0,T]$.}
At this stage we have the entropic representation $\tilde J_t(W)  =   -\tfrac{1}{\gamma}\log\E[e^{-\gamma W} \mid \F_t]$ for every dyadic $t \in \mathcal{D} \coloneqq \bigcup_n \{kT/2^n \mid 0 \leq k \leq 2^n\}$.
We extend to arbitrary $t \in [0,T]$ as follows.
Fix $t$ and a dyadic sequence $t_k \downarrow t$ in $\mathcal{D}$.
By right-continuity of $\filt$, $\F_t = \bigcap_k \F_{t_k}$ \citep[Lemma 7.13]{Kallenberg2021}, 
so by the L\'evy downward martingale convergence theorem applied to $e^{-\gamma W} \in \Linf \subset L^1$ \citep[Theorem 7.23]{Kallenberg2021},
\begin{equation*}
    \E[e^{-\gamma W} \mid \F_{t_k}] \longrightarrow \E[e^{-\gamma W} \mid \F_t] \quad \text{almost surely.}
\end{equation*}
Using the uniform lower bound $\E[e^{-\gamma W} \mid \F_{t_k}] \geq e^{-\gamma\|W\|_\infty} > 0$ and continuity of the logarithm, $-\tfrac{1}{\gamma}\log\E[e^{-\gamma W} \mid \F_{t_k}] \to -\tfrac{1}{\gamma}\log\E[e^{-\gamma W} \mid \F_t]$ almost surely.
On the left-hand side, take any $\pi \in \Pi$ with $\WL_T(\pi) = W$; by Proposition~\ref{prop:right-continuity} (R), $\tilde J_{t_k}(W) = J_{t_k}(\pi) \to J_t(\pi) = \tilde J_t(W)$ in probability.
By uniqueness of limits in probability (with both sides now viewed as convergent in probability,
   the a.s.\ convergence of the right-hand side implying its in-probability convergence), 
   we have
    $\tilde J_t(W) = -\tfrac{1}{\gamma}\log\E[e^{-\gamma W} \mid \F_t]$  for every $t \in [0,T]$ and $W \in \Linf(\F_T)$.

Step 3 yields $\rho_t(\WL_T) = \frac{1}{\gamma}\log\E[\exp(-\gamma \WL_T) \mid \F_t]$ for some $\gamma \in (0, \infty)$.
Translating back, $J_t(\pi) = -\frac{1}{\gamma}\log\E[\exp(-\gamma \WL_T(\pi)) \mid \F_t]$, which is~\eqref{eq:entropic}.
This completes the forward direction.

\paragraph{($\Leftarrow$)  Converse direction: the entropic form implies the axioms are satisfied.}
Suppose now that $J$ has the entropic   form~\eqref{eq:entropic} for some $\gamma \in (0, \infty)$.
We verify each of the five core axioms (J1)--(J5) on $J$.
Throughout, for $\pi \in \Pi$ write $W \coloneqq \WL_T(\pi) \in \Linf(\F_T)$, 
 so that $\|W\|_\infty < \infty$ and $e^{-\gamma W}$ is bounded above by $e^{\gamma\|W\|_\infty}$ and below by $e^{-\gamma\|W\|_\infty} > 0$.
In particular, $\E[e^{-\gamma W} \mid \F_t] \in \Linf(\F_t)$ is bounded away from zero, so $\log\E[e^{-\gamma W} \mid \F_t]$ is well-defined and bounded, 
and we have  $J_t(\pi) \in \Linf(\F_t)$.

\emph{(J1) Cash-Additivity.}
Let $\pi \in \Pi$, $t \in [0,T]$, and $c \in \Linf(\F_t)$ bounded.
By Lemma~\ref{lem:density} (Step~1) applied to $\pi$ at time $t$, the cash-injected strategy $\pi^{[c]}$ satisfies \mbox{$\WL_T(\pi^{[c]}) = \WL_T(\pi) + c$} almost surely.
Since $c$ is $\F_t$-measurable, $e^{-\gamma c}$ is $\F_t$-measurable and bounded, 
and hence the conditional expectation pulls it out, so we get
\begin{equation*}
\E[e^{-\gamma(\WL_T(\pi)+c)} \mid \F_t]  =  e^{-\gamma c}\, \E[e^{-\gamma\WL_T(\pi)} \mid \F_t] \quad \text{almost surely.}
\end{equation*}
Taking $-\tfrac{1}{\gamma}\log$ on both sides and using 
the identity 
$\log(e^{-\gamma c} \cdot Y) = -\gamma c + \log Y$ for $Y > 0$ gives $J_t(\pi^{[c]}) = J_t(\pi) + c$ almost surely.

\emph{(J2) Normalization.}
If $\WL_T(\pi_0) = 0$ almost surely, then $e^{-\gamma\WL_T(\pi_0)} = 1$ almost surely, 
 so $\E[1 \mid \F_t] = 1$ almost surely, and $J_t(\pi_0) = -\tfrac{1}{\gamma}\log 1 = 0$ almost surely.

\emph{(J3) Concavity, with strict clause.}
We verify the two parts of Axiom~\ref{ax:J3} separately.

\emph{Weak inequality.}
Fix $t \in [0,T]$, $\pi, \pi' \in \Pi$, and an $\F_t$-measurable random variable $\lambda$ with $0 \leq \lambda \leq 1$ almost surely.
Write $W \coloneqq \WL_T(\pi)$, $W' \coloneqq \WL_T(\pi')$, $X \coloneqq e^{-\gamma W}$, and $Y \coloneqq e^{-\gamma W'}$, both strictly positive and bounded.
Then $e^{-\gamma(\lambda W + (1-\lambda)W')} = X^\lambda Y^{1-\lambda}$ pointwise.
Since $Y > 0$ almost surely and $\E[Y \mid \F_t] \geq e^{-\gamma\|W'\|_\infty} > 0$ almost surely, 
  define the probability measure $\Q$ on $\F_T$ by $d\Q/d\Prob \coloneqq Y/\E[Y \mid \F_t]$.
Then $\E_\Q[Z \mid \F_t] = \E[Y Z \mid \F_t]/\E[Y \mid \F_t]$ for all $Z \in \Linf$, almost surely, and
 \begin{equation*}
    \E[X^\lambda Y^{1-\lambda} \mid \F_t]  =  \E\!\left[Y \cdot (X/Y)^\lambda \mid \F_t\right]  =  \E[Y \mid \F_t] \cdot \E_\Q\!\left[(X/Y)^\lambda \mid \F_t\right].
\end{equation*}
On $\{0 < \lambda < 1\}$, for each fixed realization of  the $\F_t$-measurable random variable $\lambda$ the function $z \mapsto z^\lambda$ is concave on $(0,\infty)$; applying conditional Jensen under $\Q$ pointwise on $\omega$, treating $\lambda(\omega)$ as a fixed coefficient (formally, integrate the $\F_t$-measurable kernel $(z, \omega) \mapsto z^{\lambda(\omega)}$ against the regular conditional distribution of $X/Y$ given $\F_t$, which exists by \citet[Theorem 8.5]{Kallenberg2021}), gives
\begin{equation*}
    \E_\Q\!\left[(X/Y)^\lambda \mid \F_t\right]  \leq  \left(\E_\Q[X/Y \mid \F_t]\right)^\lambda  =  \left(\frac{\E[X \mid \F_t]}{\E[Y \mid \F_t]}\right)^\lambda \quad \text{almost surely.}
\end{equation*}
 On $\{\lambda = 0\}$ and $\{\lambda = 1\}$, 
 the inequality is trivially an equality.
 Substituting back yields
\begin{equation}
    \E[X^\lambda Y^{1-\lambda} \mid \F_t]  \leq  \E[X \mid \F_t]^\lambda \, \E[Y \mid \F_t]^{1-\lambda} \quad \text{almost surely.}
    \label{eq:cond-holder}
\end{equation}
Applying the strictly decreasing map $z \mapsto -\tfrac{1}{\gamma}\log z$ to both sides of~\eqref{eq:cond-holder} reverses the inequality, and using $\log(a^\lambda b^{1-\lambda}) = \lambda\log a + (1-\lambda)\log b$ for $a, b > 0$ yields
\begin{equation*}
    -\tfrac{1}{\gamma}\log\E[e^{-\gamma(\lambda W + (1-\lambda)W')} \mid \F_t]  \geq  \lambda J_t(\pi) + (1-\lambda) J_t(\pi') \quad \text{almost surely.}
\end{equation*}
The left-hand side equals  $J_t(\lambda\pi \oplus (1-\lambda)\pi')$ by the wealth-summary property (W), which is the weak concavity inequality of Axiom~\ref{ax:J3}.

\emph{Strict inequality at $t = 0$.}
Suppose $W - W'$ is not almost-surely equal to a deterministic  constant, and let $\lambda \in (0, 1)$ be deterministic.
At $t = 0$, the $\sigma$-algebra $\F_0$ is $\Prob$-trivial, so the conditional expectations reduce to unconditional ones.
Define the probability measure $\Q$ on $\F_T$ by $d\Q/d\Prob \coloneqq Y/\E[Y]$ (well-defined since $\E[Y] \geq e^{-\gamma\|W'\|_\infty} > 0$).
It follows that 
\begin{equation*}
    \E[X^\lambda Y^{1-\lambda}]  =  \E\!\left[Y  \cdot (X/Y)^\lambda\right]  =  \E[Y] \cdot \E_\Q\!\left[(X/Y)^\lambda\right].
\end{equation*}
The function $z \mapsto z^\lambda$ is strictly concave on $(0, \infty)$ for deterministic $\lambda \in (0,1)$, 
  and $X/Y$ is a strictly positive bounded random variable.
An application of Jensen's inequality gives that 
$\E_\Q[(X/Y)^\lambda] \leq (\E_\Q[X/Y])^\lambda$, with equality  if and only if $X/Y$ is $\Q$-almost-surely equal to a constant.
Since $d\Q/d\Prob = Y/\E[Y]$ is strictly positive almost surely, we have $\Q \sim \Prob$ on $\F_T$, 
so $\Q$-a.s.\ constant is equivalent to $\Prob$-a.s.\ constant.
Now $X/Y = e^{-\gamma(W - W')}$,  and $z \mapsto e^{-\gamma z}$ is a bijection on $\R$. 
Therefore,   $X/Y$ is $\Prob$-a.s.\ constant if and only if $W - W'$ is $\Prob$-a.s.\ constant.
By hypothesis this fails, hence Jensen's inequality is strict: 
\begin{equation}
    \E[X^\lambda Y^{1-\lambda}]  <  \E[X]^\lambda \,  \E[Y]^{1-\lambda}.
    \label{eq:uncond-holder}
\end{equation}
Applying the strictly decreasing map $z \mapsto -\tfrac{1}{\gamma}\log z$ to~\eqref{eq:uncond-holder} yields
\begin{equation*}
    -\tfrac{1}{\gamma}\log\E[e^{-\gamma(\lambda W + (1-\lambda)W')}]  >  \lambda J_0(\pi) + (1-\lambda) J_0(\pi'),
\end{equation*}
which is the strict concavity inequality of Axiom~\ref{ax:J3} at $t = 0$.

\emph{(J4) Strong Dynamic Consistency.}
For $s \leq t$ in $[0,T]$ and any $\pi \in \Pi$, write
\begin{equation*}
    A_u(\pi)  \coloneqq  \E[e^{-\gamma\WL_T(\pi)} \mid \F_u], \qquad u \in [0,T].
\end{equation*}
We have  $J_u(\pi) = -\tfrac{1}{\gamma}\log A_u(\pi)$. 
Equivalently, $A_u(\pi) = e^{-\gamma J_u(\pi)}$ almost surely.
The tower property of conditional expectations gives $A_s(\pi) = \E[A_t(\pi) \mid \F_s]$ almost surely, i.e.,
\begin{equation}
    \E[e^{-\gamma  \WL_T(\pi)} \mid \F_s]  =  \E\!\bigl[e^{-\gamma J_t(\pi)} \mid \F_s\bigr] \quad \text{almost surely.}
    \label{eq:tower-entropic}
\end{equation}
Note that this  is the recursivity identity for the entropic functional.
Now suppose $J_t(\pi) \geq J_t(\pi')$ almost surely.
Then $e^{-\gamma J_t(\pi)} \leq e^{-\gamma J_t(\pi')}$ almost surely,  so conditional monotonicity 
implies that 
 $\E[e^{-\gamma J_t(\pi)} \mid \F_s] \leq \E[e^{-\gamma J_t(\pi')} \mid \F_s]$ almost surely.
Using~\eqref{eq:tower-entropic}, the latter  is equivalent to 
 $\E[e^{-\gamma\WL_T(\pi)} \mid \F_s] \leq \E[e^{-\gamma\WL_T(\pi')} \mid \F_s]$ almost surely, 
 and applying $-\tfrac{1}{\gamma}\log$ yields $J_s(\pi) \geq J_s(\pi')$ almost surely.

\emph{(J5) Law-Invariance.}
Suppose $\WL_T(\pi) \eqdist \WL_T(\pi')$.
The transformation $x \mapsto e^{-\gamma x}$ is Borel-measurable,   so $e^{-\gamma\WL_T(\pi)} \eqdist e^{-\gamma\WL_T(\pi')}$. 
In particular,  $\E[e^{-\gamma\WL_T(\pi)}] = \E[e^{-\gamma\WL_T(\pi')}]$.
At $t = 0$, the $\sigma$-algebra $\F_0$  is $\Prob$-trivial, so the conditional expectation reduces to the unconditional one and is deterministic.
Thus, $J_0(\pi) = -\tfrac{1}{\gamma}\log\E[e^{-\gamma\WL_T(\pi)}] =  - \tfrac{1}{\gamma}\log\E[e^{-\gamma\WL_T(\pi')}] = J_0(\pi')$ almost surely.

This completes the verification of the converse direction, and the proof of Theorem~\ref{thm:forced-uniqueness}. \qed

\subsection{The Clock-Invariance Corollary}
\label{subsec:clock-invariance}

The forced entropic representation of Theorem~\ref{thm:forced-uniqueness} has  an immediate structural consequence: 
 the scalar $\gamma$ is invariant under the canonical clock change associated with the price process.

 \begin{corollary}[Clock-Invariance of the Risk-Aversion Parameter]
\label{cor:clock-invariance}
Adopt the hypotheses of Theorem~\ref{thm:forced-uniqueness}, and let $\Lambda \coloneqq \langle S\rangle$ denote the quadratic-variation clock of the price process, with right-continuous inverse, 
$ \tau(u)  \coloneqq  \inf\{t \geq 0 \mid \Lambda_t  \geq  u\}$ for 
$u \in [0, \Lambda_T].$
The associated business-time filtration is $\tilde\filt = (\tilde\F_u)_{u \in [0, \Lambda_T]}$ with $\tilde\F_u \coloneqq \F_{\tau(u)}$, and the reparametrized criterion is $\widehat J_u(\pi) \coloneqq J_{\tau(u)}(\pi)$ for $u \in [0, \Lambda_T]$ and $\pi \in \Pi$. 
Then
\begin{equation}
    \widehat J_u(\pi)  =  -\frac{1}{\gamma}\, \log \E\!\left[\exp \bigl(-\gamma\, \WL_T(\pi)\bigr) \,\Big|\, \tilde\F_u\right]
    \label{eq:entropic-business}
\end{equation}
for all $u \in [0, \Lambda_T]$ and $\pi \in \Pi$, with the same scalar $\gamma$ as in Theorem~\ref{thm:forced-uniqueness}.
\end{corollary}

\begin{proof}
By Theorem~\ref{thm:forced-uniqueness}, $J_t(\pi) = -\tfrac{1}{\gamma}\log \E[e^{-\gamma \WL_T(\pi)} \mid \F_t]$ for all $t \in [0, T]$ with unique $\gamma > 0$.
Substituting $t = \tau(u)$ and using $\tilde\F_u = \F_{\tau(u)}$  yields~\eqref{eq:entropic-business}.
The calculation establishes that the scalar $\gamma$ is the same whether the agent is described  in wall-clock or business-time, 
i.e., the entropic representation is indeed stable under the canonical clock change.
\end{proof}

\begin{remark}
The scalar $\gamma$ is a property of the agent's preferences alone, recoverable identically from the wall-clock and business-time formulations. 
The substantive content arises in the stochastic-volatility setting of Section~\ref{subsec:stochvol}: 
the same $\gamma$ appearing  in~\eqref{eq:entropic-business} pins the instantaneous running-penalty rate to $\phi_t = \gamma\sigma_t^2/2$, 
so $\gamma$ is the agent's single time-invariant preference parameter (Proposition~\ref{prop:clock-corrected}).
 The agent's mean-variance trade-off is thereby stationary in $\langle S \rangle$, the clock against which price risk accumulates, rather than in wall-clock time. 
A wall-clock constant $\phi$, as in direct extensions of the CJ framework to stochastic-volatility settings, is correspondingly inconsistent with the axioms:   
  by the forced relation, it implies a time-varying $\gamma_t = 2\phi/\sigma_t^2$, which Corollary~\ref{cor:clock-invariance} excludes.
In the constant-volatility benchmark of \citet{AvellanedaStoikov2008}, the wall-clock and business-time formulations coincide up to a deterministic rescaling and the result is vacuous.
\end{remark}

Corollary~\ref{cor:clock-invariance} encodes a microstructural commitment  which is satisfied (vacuously) by the constant-volatility AS benchmark, used implicitly by some recent stochastic-volatility extensions \citep{RosenbaumZhang2022}, and violated by direct wall-clock extensions of CJ to stochastic-volatility settings.
To our knowledge, the commitment has not previously been stated on the preference functional itself.

The price's diffusive risk accumulates along the quadratic-variation process $\langle S\rangle$: in the constant-volatility benchmark $\langle S\rangle_t = \sigma^2 t$ is deterministic linear in wall-clock time, while in stochastic-volatility settings $\langle S\rangle_t = \int_0^t \sigma_s^2\, ds$ is random, growing faster on volatile days and slower on quiet ones.
The market maker's inventory aversion, in our framework, is a commitment about how she trades off the mean and variance of her terminal wealth, and the natural units in which both accumulate are the units of price risk, namely $\langle S\rangle$.
Corollary~\ref{cor:clock-invariance} says that this is automatic; the agent's risk-aversion parameter $\gamma$ is the same scalar in wall-clock and business-time formulations,  so her mean-variance trade-off is stationary in the same clock against which price risk accumulates.
Among possible clocks, $\langle S\rangle$ is privileged by the Dambis--Dubins--Schwarz theorem, 
which states that for any continuous semimartingale $S$ with nontrivial quadratic variation,
there exists a Brownian motion $B$ on a possibly enlarged probability space such that $S_t = B_{\langle S\rangle_t}$, 
making $\langle S\rangle$ the unique clock under which the price process becomes a standard Brownian motion.

In the constant-volatility benchmark of \citet{AvellanedaStoikov2008}, business-time and wall-clock time coincide up to a deterministic rescaling, and Corollary~\ref{cor:clock-invariance} imposes no constraint beyond what Theorem~\ref{thm:forced-uniqueness} already gives.
Observe that this explains why the original AS framework satisfies clock-invariance without ever invoking it; 
  indeed,  the property is automatic in any setting with constant volatility.
When the volatility $\sigma_t$ is random, the two clocks diverge: a preference functional stationary in wall-clock time would deliver a constant inventory-aversion rate per second, while a preference functional stationary in  $\langle S\rangle$ delivers a constant inventory-aversion rate per unit of accumulated price variance.
The two prescriptions give different optimal quotes on the same day, with the wall-clock prescription under-penalizing inventory on high-volatility days and over-penalizing on low-volatility days. 
Corollary~\ref{cor:clock-invariance} commits to the second prescription as the correct one.
The forced relation $\phi_t = \gamma\sigma_t^2/2$ of Proposition~\ref{prop:clock-corrected}  is the corrected calibration that clock-invariance mandates.

A concrete example makes the nature of the violation more explicit.
In a Heston model with $d\sigma_t^2 = \kappa(\bar\sigma^2 - \sigma_t^2)\,dt + \nu\sigma_t\, dW_t$, the naive CJ running penalty $\phi\int_0^T q_s^2\,ds$ with constant $\phi$ corresponds, in business-time units $u = \int_0^t \sigma_s^2\,ds$, precisely to the penalty $\phi\int_0^{\Lambda_T} q_{\tau(u)}^2/\sigma_{\tau(u)}^2\,du$.
The effective business-time coefficient $\phi/\sigma_{\tau(u)}^2$ is random and time-varying. 
By Corollary~\ref{cor:clock-invariance}, this  is inconsistent with the forced constant-coefficient entropic structure.
The correct (clock-coherent) running penalty is  $\tfrac{\gamma}{2}\int_0^T q_s^2\,d\langle S\rangle_s = \tfrac{\gamma}{2}  \int_0^T \sigma_s^2 q_s^2\,ds$, as given by Proposition~\ref{prop:clock-corrected}.

The extensive literature on subordinated stochastic processes, originating  with the seminal work of \citet{Clark1973} and subsequently developed in continuous-time financial modeling by \citet{AneGeman2000} and \citet{CarrWu2004}, treats asset prices as time-changed Brownian motions in which the stochastic clock captures fluctuations in market activity and information flow. 
In fact,  \citet{BacryAlDayriMuzy2011} establish empirically that price variance is approximately linear  in trade count, suggesting that trade-count time is a viable proxy for $\langle S\rangle$-time in liquid markets.
Our framework absorbs this empirical fact as a structural feature of the preference functional, 
namely that the market maker's inventory aversion is automatically stationary in the same clock against which price risk accumulates.

\begin{remark}[BSDE-Level Consequence]
\label{rem:clock-BSDE}
Consider preference functionals satisfying our axioms which admit a BSDE representation
\begin{equation*}
    -d\tilde J_t(W)  =  g_t(Z_t)\, dt - Z_t \cdot dM_t
\end{equation*}
in the sense of \citet{BionNadal2009, DelbaenPengRosazzaGianin2010}, where $\tilde J_t$ is the reduced functional of Step~1 of the proof of Theorem~\ref{thm:forced-uniqueness} acting on $W = \WL_T(\pi) \in \Linf(\F_T)$, $g$ is the driver, $Z$ is the $L^2(\Prob)$-integrand in the martingale representation $d\,\E[e^{-\gamma \WL_T} \mid \F_t] = Z_t \cdot dM_t$, and $M$ is the driving martingale (typically the Brownian motion $B$ from Section~\ref{subsec:price-flow}).
On this class, Corollary~\ref{cor:clock-invariance} forces the driver to take the structural form 
$g_t(z) = \sigma_t^2 \cdot \tilde g_{\Lambda_t}(z)$,
where $\tilde g$ is a deterministic function of its business-time argument $u = \Lambda_t$ only.
For the entropic functional the explicit driver is $g_t(z) = \tfrac{\gamma}{2}\sigma_t^2 |z|^2$, recovering the standard quadratic-BSDE form for entropic risk measures and producing the forced running coefficient $\phi_t = \gamma\sigma_t^2/2$ of Proposition~\ref{prop:clock-corrected}.
\end{remark}

\begin{remark}[Constant-Volatility Vacuity]
\label{rem:clock-vacuity}
When $\Lambda_t = \sigma^2 t$ is deterministic linear, $\tilde\filt$ and $\filt$ differ only by a deterministic time-rescaling.
Corollary~\ref{cor:clock-invariance} is then automatically satisfied as a trivial restatement of  Theorem~\ref{thm:forced-uniqueness} on the rescaled filtration.
The corollary has real content only in settings where $\Lambda$ has a nontrivial random structure (stochastic-volatility, rough-volatility, or jump-diffusion).
\end{remark}

\begin{remark}[Noncanonical Clocks]
\label{rem:noncanonical}
Corollary~\ref{cor:clock-invariance} structurally privileges the canonical clock $ \Lambda = \langle S\rangle$.
A market maker insisting on a different clock, say  $\Lambda' = N$  (the cumulative trade-count process),  is incompatible with the clock-invariance forced by the corollary unless $N_t = f(t)\,\langle S\rangle_t$ almost surely for some deterministic function $f > 0$, i.e., unless $N$ and $\langle S\rangle$ are proportional up to a deterministic time-varying factor.
The empirical regularity $\E[d\langle S\rangle_t] \approx  c\, \E[dN_t]$ documented by \citet{BacryAlDayriMuzy2011} holds in expectation only; under this weaker regularity, the trade-count clock yields at best  approximate time-stationarity in the sense of Corollary~\ref{cor:clock-invariance}.
\end{remark}

\subsection{Discussion of the Theorem}
\label{subsec:theorem-discussion}

Theorem~\ref{thm:forced-uniqueness} has not, to our knowledge, been stated before in the inventory market making literature.
The axiom system itself is novel. 
In particular, the wealth-summary property (W) (Proposition~\ref{prop:wealth-summary}), which we derive from cash-additivity, normalization, and strong dynamic consistency rather than assume, is a new contribution to the theory of dynamic preference functionals.
Clock-invariance of $\gamma$ with respect to the canonical clock $\langle S\rangle$ (Corollary~\ref{cor:clock-invariance}) is a further new structural contribution and underwrites the framework's applicability to stochastic-volatility settings.
Of the five core axioms, strong dynamic consistency (J4) is the one that most directly drives the uniqueness conclusion. 
It forces the recursivity relation and collapses the one-parameter family of $\Phi$-divergence
 risk measures onto the entropic member  (see Remark~\ref{rem:why-kl}).
The non-atomicity of $\F_T$ enters at a single point of the proof, namely the identifiability argument of Sub-step~3b, where it ensures the existence of a nondegenerate test random variable to pin down $\gamma_n$.
Finally, the proof uses, at one step, a standard representation result from the dynamic  risk measure literature  \citep[Theorem 1.10]{KupperSchachermayer2009};
 the substantive content of the theorem lies elsewhere, 
 in the consequences developed in Section~\ref{sec:consequences} for the AS-CJ split and for the regulatory CVaR framework.
 

\section{Consequences of the Theorem}
\label{sec:consequences}

We now develop eleven sharp consequences of the forced uniqueness, as formal corollaries of Theorem~\ref{thm:forced-uniqueness}.
Unless otherwise stated, each corollary is read under the hypotheses of the theorem, namely Axioms~\ref{ax:J1}--\ref{ax:J5}.
The two corollaries on mean-variance and CVaR (Corollaries~\ref{cor:meanvar} and~\ref{cor:cvar}) are exceptions: 
  they are statements about \emph{candidate alternative functionals}, and require only the axiom statements, not the theorem's conclusion.
Four of the corollaries (\ref{cor:cj-inconsistency}, \ref{cor:meanvar}, \ref{cor:cvar}, \ref{cor:dilation}) draw conclusions in apparent tension with widely-used practical, textbook, or axiomatic frameworks, and we discuss the tension in each case.

\subsection{Single-Parameter Pinning}
\label{subsec:single-parameter}

The first consequence of Theorem~\ref{thm:forced-uniqueness} is that the entire space of preference functionals satisfying our axioms is one-dimensional, parametrized by a single positive scalar.

\begin{corollary}[Single-Parameter Pinning]
\label{cor:single-parameter}
Under the same hypotheses as in Theorem~\ref{thm:forced-uniqueness}, the market maker's dynamic preference functional $J$ is determined, given the canonical clock $\Lambda = \langle S\rangle$ and the liquidation cost function $L$, by a single free scalar parameter $\gamma > 0$.
\end{corollary}

\begin{proof}
This is immediate from the explicit form~\eqref{eq:entropic}.
The only free quantity on the right-hand side  of~\eqref{eq:entropic}, beyond the data $(\langle S\rangle, L)$ which are properties of the market and not of the agent, is the scalar $\gamma$.
\end{proof}

The economic content of Corollary~\ref{cor:single-parameter} is worth stating sharply.
Cartea--Jaimungal-style models typically parametrize the market maker's preferences by at least three numbers, namely the running inventory penalty coefficient $\phi$, the terminal inventory penalty coefficient $\alpha$, and (when explicit) a risk-aversion parameter.
Corollary~\ref{cor:single-parameter} says, in effect, that two of these are not free, by way of two distinct mechanisms.
The running coefficient $\phi$ is forced to be a specific function of the preference 
parameter $\gamma$ and the realized volatility, computed explicitly in Section~\ref{sec:supporting}.
 The terminal coefficient $\alpha$ is forced to be the leading coefficient of the liquidation cost function $L$, a primitive of the market rather than a preference parameter.
The genuine preference degree of freedom is the single scalar $\gamma$.

A practitioner used to the CJ parametrization might find independent calibration of $\phi$ and $\alpha$ empirically convenient.
Our framework does not endorse this:
  any combination of $(\phi, \alpha, \gamma)$ outside the two-equation constraint surface
\begin{equation*}
        \begin{cases}
        \phi  =  \gamma \sigma^2/2, \\[4pt]
        \alpha  =  \tfrac{1}{2}L''(0) 
        \end{cases}
\end{equation*}
violates at least one of Axioms~\ref{ax:J1}--\ref{ax:J5}.
Calibrating outside the surface is calibrating to a model that, on the axiomatic terms above, the trader does not in fact have. 
The scope of the forcing claim deserves emphasis. 
The statement is about agents with a well-defined preference functional satisfying Axioms~\ref{ax:J1}--\ref{ax:J5}.
A CJ-style desk may not be maximizing any preference functional at all, instead using the CJ-formula as a heuristic without a clean preference-functional interpretation;
 the forcing applies to the class of agents who do have such a preference functional, 
and tells those agents that their $(\phi, \alpha, \gamma)$ are not free.

\begin{remark}[Constancy of $\gamma$]
\label{rem:constant-gamma}
A direct consequence of the explicit form~\eqref{eq:entropic} is that $\gamma$ is a single positive scalar.
It does not depend on time $t$, on the state of the world $\omega$, on the price level, on  realized volatility, on the time-of-day, or on the intraday P\&L.
What \emph{can} legitimately vary across the trading day, in our framework, is the inventory penalty rate $\frac{\gamma}{2}\sigma_t^2 q^2$, 
which naturally tracks realized volatility through $\sigma_t^2$. 
The underlying preference parameter $\gamma$ itself does not move.
The constancy result does not preclude $\gamma$ from being a function of slow-moving structural quantities such as asset class, desk capital, or firm risk policy. 
What it precludes is $\gamma$ being a function of intraday state.
A market maker whose subjective $\gamma$ does vary intraday is free to do so, but she is then outside the axiom system.
The stronger cross-asset statement is the content of the multi-asset extension (Theorem~\ref{thm:multi-asset} in Appendix~\ref{app:multi-asset}; 
see also Section~\ref{subsec:further-connections}).
\end{remark}

\subsection{AS Uniqueness}
\label{subsec:as-uniqueness}

The second consequence identifies the Avellaneda--Stoikov framework as the unique axiom-consistent model in the constant-volatility benchmark.

\begin{corollary}[Avellaneda--Stoikov Uniqueness]
\label{cor:as-uniqueness}
Suppose that the price process $S$ has constant volatility $\sigma > 0$, i.e., $d \langle S\rangle_t = \sigma^2 \,  dt$.
Then,  under the same hypotheses as  Theorem~\ref{thm:forced-uniqueness}, the market maker's preferences are uniquely given by the entropic certainty-equivalent functional
  \begin{equation}
    J_t(\pi)  =  -\frac{1}{\gamma}\, \log \E\!\left[\exp \bigl(-\gamma\, \WL_T(\pi)\bigr) \,\Big|\, \F_t\right],
    \label{eq:AS-entropic}
\end{equation}
which by Remark~\ref{rem:cara-dynamic} coincides, for the CARA utility $u(x) = -\exp(-\gamma x)$, with the conditional certainty-equivalent of expected utility under  $u$.
Under the further specializations of \citet{AvellanedaStoikov2008},  with exponential arrival intensities $\lambda^a(\delta) = A e^{-\kappa \delta}$ and vanishing terminal liquidation cost $L \equiv 0$, the optimal quoting strategy derived from~\eqref{eq:AS-entropic} coincides with the strategy in their paper.
\end{corollary}

\begin{proof}
The functional~\eqref{eq:AS-entropic} is the constant-volatility specialization of~\eqref{eq:entropic}.
Equivalence with the AS strategy under their further specializations is by standard calculation; see \citet{AvellanedaStoikov2008} or \citet[Section 10.3]{CarteaJaimungalPenalva2015}.
The forcing in Corollary~\ref{cor:as-uniqueness} is therefore not a new optimization result but a statement about which preference functional an analyst must use to obtain the AS strategy as an optimum.
\end{proof}

The CARA utility in the AS paper is introduced with two motivations, namely tractability and the wealth-independence of the optimal $\delta^a, \delta^b$. 
Corollary~\ref{cor:as-uniqueness} changes the picture entirely.
 Given Axioms~\ref{ax:J1}--\ref{ax:J5}, CARA is not a modeling choice; it is in fact a theorem.
The wealth-independence property, often cited as a reason to use CARA, 
is a direct consequence of cash-additivity (J1) and does not require the CARA commitment, 
as shown in Corollary~\ref{cor:wealth-indep} below.

The mathematical content of AS is therefore not expected-utility maximization but dynamic-risk-measure minimization on liquidation-adjusted terminal wealth; the CARA-expected-utility framing is an artifact of the certainty-equivalent identity (Remark~\ref{rem:cara-dynamic}) that holds for CARA and for no other utility.
The standard AS calculation could be presented, without any change in conclusions, as the explicit computation of the optimal strategy for a dynamic entropic risk measure, with no reference to expected utility at all.

A related consequence is that the standard portfolio-theory   utility hierarchy is incompatible with our framework: 
CRRA, HARA, and the rest of the Merton menagerie all fail one or more of Axioms~\ref{ax:J1}--\ref{ax:J5}, with cash-additivity (J1) the most direct failure
 (under power utility, the marginal utility of a sure dollar depends on current wealth, 
 contradicting $J_t(\pi + c) = J_t(\pi) + c$ for $\F_t$-measurable bounded $c$).
Market making and portfolio choice, which share much of their mathematical apparatus, 
require structurally different preference primitives; 
the portfolio-theory tradition treats CARA as a special case, 
while ours forces CARA as the only case.

\subsection{CJ as Forced Second-Order Approximation}
\label{subsec:cj}

We now turn to the Cartea--Jaimungal framework, which models the market maker as maximizing
\begin{equation}
    \E\!\left[\WL_T(\pi)\right] - \phi \, \E\!\left[\int_0^T q_s^2 \, ds\right]
    \label{eq:CJ-objective}
\end{equation}
over admissible strategies, with constant $\phi > 0$ a running inventory penalty.
The terminal liquidation term, when present, is absorbed into our $L$,  as discussed in Section~\ref{sec:setup}.
In the practitioner literature, \eqref{eq:CJ-objective} is sometimes interpreted as a  second-order expansion of an entropic functional, with the running quadratic term serving as a conditional-variance approximation to the inventory contribution; we make this interpretation precise below.

\begin{corollary}[CJ as Forced Second-Order Approximation]
\label{cor:cj-inconsistency}
Under the same hypotheses as  Theorem~\ref{thm:forced-uniqueness}, the Cartea--Jaimungal objective~\eqref{eq:CJ-objective} is incompatible with the axiom system as a primitive functional.
It is, however, the second-order expansion of   the entropic functional~\eqref{eq:entropic} around zero inventory in the constant-volatility setting, provided that the running coefficient takes the value
\begin{equation}
    \phi  =  \frac{\gamma \sigma^2}{2}.
    \label{eq:phi-forced}
\end{equation}
Any other value of $\phi$ produces a CJ objective that is not the expansion   of any functional satisfying Axioms~\ref{ax:J1}--\ref{ax:J5}.
\end{corollary}

\begin{proof}
For the first statement, observe that the functional~\eqref{eq:CJ-objective} depends on the strategy $\pi$ through both $\WL_T(\pi)$ (the terminal-wealth term) and the integral $\int_0^T q_s^2(\pi) \, ds$ (the running-penalty term).
The \emph{primary} failure is path-dependence: the running-penalty term is a path functional of the inventory process $q$, which can take different values for two strategies that produce the same distribution of $\WL_T$, so~\eqref{eq:CJ-objective} violates property (W) (Proposition~\ref{prop:wealth-summary}).
Even if one ignores the running term, the linear expected-wealth component $\E[\WL_T]$ is the risk-neutral limit excluded by the strict-concavity clause~(J3b), constituting an independent \emph{secondary} failure.
(To be precise: because the full CJ objective violates (W) via its path-dependent running term, it is not a functional on $\WL_T$ at all, so Axiom~\ref{ax:J3} does not directly apply to it; the secondary failure is a statement about the terminal-wealth component $\E[\WL_T]$ considered in isolation.)
Either failure alone is sufficient to conclude that the CJ objective violates the axiom system.

For the second statement, we begin  by computing  
the second-order expansion of the entropic functional~\eqref{eq:AS-entropic} around the no-inventory benchmark in the constant-volatility setting.
The full calculation is given in Section~\ref{sec:supporting} as the proof of   Proposition~\ref{prop:forced-phi}; 
the headline result is that the certainty-equivalent contribution of holding inventory $q$ at time $s$, accumulated over an infinitesimal interval $ds$, equals $-\frac{\gamma \sigma^2}{2} q^2 \, ds$, with the coefficient $\gamma\sigma^2/2$ arising \emph{exactly} from the  HJB diffusion term (no approximation in the coefficient itself; 
the ``second-order'' qualifier refers to the inventory expansion of the full functional, not to the coefficient).
This is exactly the running-penalty contribution in~\eqref{eq:CJ-objective} with coefficient $\phi = \gamma\sigma^2/2$, which establishes both directions of the second statement.
See Remark~\ref{rem:two-senses} below for the distinction made precise.
\end{proof}

A working market-making desk that uses the CJ framework as a primitive objective is, on our axiomatic terms, using a functional whose path-dependence (through the inventory penalty $\int q^2 \, ds$) violates the wealth-summary requirement (W) (Proposition~\ref{prop:wealth-summary}).
If, however, the desk wishes to interpret CJ as a tractable second-order approximation to entropic preferences, it is welcome to do so, but it must use $\phi = \gamma\sigma^2/2$.
Any other value of $\phi$ corresponds to a CJ objective that cannot be obtained as an expansion of any entropic functional, and is therefore not consistent with the CARA preferences which the second-order interpretation is supposed to encode.
At $\phi = \gamma\sigma^2/2$, the CJ-formula strategy and the AS-optimal strategy  produce the same reservation price and half-spread formulas, so they coincide  exactly in the constant-volatility benchmark.

\begin{remark}[Relation to the AS Value Function]
\label{rem:AS-value-function}
The forced coefficient~\eqref{eq:phi-forced} is implicit, though not stated, in the AS value function.
\citet{AvellanedaStoikov2008} derive, for a market maker holding inventory $q$ at time $t$ with cash $X$ and constant volatility $\sigma$, 
the formula 
\begin{equation*}
    V(t, X, q, S)  =  -\exp \! \left( -\gamma  \bigl(X + q S - \tfrac{1}{2}\gamma q^2 \sigma^2 (T-t)\bigr)\right).
\end{equation*}
The certainty-equivalent reduction $-\tfrac{1}{2}\gamma q^2 \sigma^2 (T-t)$ is the integrated running cost at rate $({\gamma\sigma^2}/{2}) \, q^2$, which is exactly equal to  $\phi q^2$ with the value~\eqref{eq:phi-forced}.
The algebraic identity has indeed been visible in the AS value function since 2008; what does not appear to have been articulated previously is its interpretation as a \emph{forced consistency condition} on the CJ-style $(\phi, \gamma)$ parametrization,   derived from preference-functional axioms rather than from a specific HJB calculation.
\end{remark}

\subsection{Forced Terminal-Penalty Coefficient}
\label{subsec:forced-alpha}

The CJ tradition also features a terminal-inventory penalty $\alpha q_T^2$, usually treated as an independent hyperparameter alongside $\phi$.
A parallel forced-coefficient logic, analogous to that of Corollary~\ref{cor:cj-inconsistency}, extends to $\alpha$.

\begin{corollary}[Forced Terminal-Penalty Coefficient]
\label{cor:forced-alpha}
Under the  hypotheses  of 
{Theorem~\ref{thm:forced-uniqueness}}, 
suppose in addition that the liquidation cost function admits a second-order 
  Taylor expansion at zero inventory: 
  $L(0) = 0$, $L'(0) = 0$, and $L(q) = \tfrac{1}{2} L''(0) \, q^2 + o(q^2)$ as $q \to 0$.
(The condition $L'(0) = 0$ is the natural requirement that the liquidation cost be minimized at zero inventory.)
Then the CJ terminal-penalty coefficient $\alpha$ is forced to take the value
 \begin{equation*}
    \alpha  =  \tfrac{1}{2}\, L''(0).
\end{equation*}
\end{corollary}

\begin{proof}
By property (W) (Proposition~\ref{prop:wealth-summary}), 
we have that $J_t(\pi)$ depends on $\pi$ only through $\WL_T(\pi) = X_T + q_T S_T - L(q_T)$.
The terminal-inventory contribution to $\WL_T$ is the term $-L(q_T)$.
By the assumed Taylor expansion of $L$ at zero, and using the hypothesis $L'(0) = 0$ to eliminate the linear term, the leading nonzero contribution near $q_T = 0$ is equal to 
 $-L(q_T) = -\tfrac{1}{2}L''(0)\,q_T^2 + o(q_T^2)$.
This is precisely the $-\alpha q_T^2$ contribution that the CJ objective adds, with $\alpha = \tfrac{1}{2}L''(0)$.

 For the rigidity, consider a market maker who specifies $-\alpha q_T^2$ as a terminal penalty with $\alpha \neq \tfrac{1}{2}L''(0)$, taking $\alpha$ as a free preference parameter.
This produces a quadratic-order dependence on $q_T$ that differs from the quadratic-order contribution $-\tfrac{1}{2}L''(0)q_T^2$ of the true $-L(q_T)$.
Since Proposition~\ref{prop:wealth-summary} fixes the terminal-wealth variable as $\WL_T = X_T + q_T S_T - L(q_T)$ with $L$ given as market data, any quadratic terminal-inventory coefficient other than $\tfrac{1}{2}L''(0)$ represents a preference functional depending on something other than $\WL_T$, which violates property (W) (Proposition~\ref{prop:wealth-summary}).
Concretely, the two functionals disagree on any pair of strategies $\pi, \pi'$ with $\WL_T(\pi) = \WL_T(\pi')$ but $q_T^2(\pi) \neq q_T^2(\pi')$, and such pairs exist whenever $q_T$ is not almost surely zero (the generic case).
The construction is the following.
Fix any two strategies $\pi_a, \pi_b \in \Pi$ with $q_T(\pi_a) \neq q_T(\pi_b)$ on a positive-probability event. 
Such a pair can be constructed by varying the quote-distance processes asymmetrically between bid and ask.
Let $\Delta W \coloneqq \WL_T(\pi_a) - \WL_T(\pi_b) \in \Linf(\F_T)$. 
Note that this is a bounded $\F_T$-measurable random variable.
By Lemma~\ref{lem:density}, the cash-shift construction of Section~\ref{subsec:strategy} produces a strategy  $\pi_b^{[\Delta W]} \in \Pi$ obtained from $\pi_b$ by injecting the bounded $\F_T$-measurable cash amount $\Delta W$ at time $T$, so that $\WL_T(\pi_b^{[\Delta W]}) = \WL_T (\pi_b) + \Delta W = \WL_T(\pi_a)$ almost surely while $q_T(\pi_b^{[\Delta W]}) = q_T(\pi_b)$ (the cash injection does not affect inventory).
Setting $\pi \coloneqq \pi_a$ and $\pi' \coloneqq \pi_b^{[\Delta W]}$ yields the desired pair: 
 $\WL_T(\pi) = \WL_T(\pi')$ almost surely, but $q_T(\pi) = q_T(\pi_a) \neq q_T(\pi_b) = q_T(\pi')$ on a positive-probability event, so $q_T^2(\pi) \neq q_T^2(\pi')$ there.
Feasibility of $\pi_b^{[\Delta W]}$ within $\Pi$ is by the admissibility verification of  Lemma~\ref{lem:density} (the cash-injection construction preserves admissibility).
\end{proof}

The corollary makes precise an observation that is implicit but not stated in the AS--CJ literature: 
$\alpha$ is not a preference parameter at all but a property of the \emph{market}
 (specifically, the curvature of the cost-to-liquidate function at zero inventory).
A market with deeper liquidity has a smaller  $L''(0)$ and hence a smaller $\alpha$; 
in the quadratic-impact benchmark $L(q) = \tfrac{\kappa}{2}q^2$, one finds $\alpha = \tfrac{\kappa}{2}$.

The forcing of $\alpha$ is different in kind from the forcing of $\phi$.
The forcing $\phi = \gamma\sigma^2/2$ is a consequence of the entropic structure of the preference functional, arising from the second-order Taylor expansion of $-\tfrac{1}{\gamma}\log\E[e^{-\gamma\WL_T}\mid\F_t]$.
The forcing $\alpha = \tfrac{1}{2}L''(0)$ is a consequence of the market structure: 
  it is the Taylor coefficient of $L$ at zero, and would hold under any preference functional satisfying Proposition~\ref{prop:wealth-summary}.
Combining Corollaries~\ref{cor:cj-inconsistency} and~\ref{cor:forced-alpha}, the full CJ-style parametrization $(\phi, \alpha, \gamma)$ collapses to a single free preference parameter $\gamma$, with $\phi$ a function of $\gamma$ and market data and $\alpha$ purely market data.
This is the sharp form of the single-parameter pinning of Corollary~\ref{cor:single-parameter}.

\subsection{Calibration Inversion: Recovering \texorpdfstring{$\gamma$}{gamma} from \texorpdfstring{$\phi$}{phi}}
\label{subsec:calibration-inversion}

 Corollary~\ref{cor:cj-inconsistency} expresses the forced relation $\phi = \gamma\sigma^2/2$ as a constraint surface in $(\phi, \gamma, \sigma)$-space.
The constraint can be read in either direction, and the reverse direction is operationally important.

\begin{corollary}[Calibration Inversion]
\label{cor:calibration-inversion}
Under the same hypotheses as  Theorem~\ref{thm:forced-uniqueness} and in the constant-volatility setting,   a desk that has calibrated the CJ running-penalty coefficient $\phi$ from inventory-path and P\&L data implicitly determines the risk-aversion parameter $\gamma$ via
\begin{equation}
    \gamma  =  \frac{2\phi}{\sigma^2}.
    \label{eq:calibration-inversion}
\end{equation}
In the stochastic-volatility setting, the same logic applies pointwise in time: if a desk has calibrated a time-varying  local penalty rate $\phi_t$ from the data, then by Corollary~\ref{cor:clock-invariance} (clock-invariance) the inversion $\gamma = 2\phi_t/\sigma_t^2$ must yield the same scalar $\gamma$ at every time $t$, with $\phi_t = \gamma\sigma_t^2/2$ the clock-coherent local rate of Proposition~\ref{prop:clock-corrected}.
\end{corollary}

\begin{proof}
Immediate from the forced relation~\eqref{eq:phi-forced}.
The stochastic-volatility version follows from Proposition~\ref{prop:clock-corrected}, which establishes that the running penalty rate must be $\gamma\sigma_t^2/2$ with a single $\gamma$ across $t$; the inversion is the same algebraic relation read pointwise.
\end{proof}
 
 The corollary has two practical readings.

The first is methodological.
A working market-making desk typically has two independent sources of preference-parameter information.
From the quoting side, $\gamma$ can be estimated from the shape and time-variation of optimal quote distances, since the AS reservation-price and optimal-spread formulas depend on $\gamma$ explicitly 
  (see, e.g., \citet[Section 10.3]{CarteaJaimungalPenalva2015}).
From the inventory side, $\phi$ can be calibrated from the historical relationship between inventory paths and realized P\&L, via standard CJ-style HJB calibrations.
Corollary~\ref{cor:calibration-inversion} says these two estimators target the \emph{same scalar}, 
up to the deterministic rescaling $2/\sigma^2$. 
Computing both is a sanity check on whether the desk is operating on the constraint surface.
If the quoting-side $\gamma$ and the inventory-side $2\phi/\sigma^2$ disagree by more than estimation error, the desk's quoting behavior and inventory management are inconsistent with any single dynamically-consistent preference functional.
A consequence of this is that a desk that has been calibrating $\phi$ for years in the CJ tradition has
   been operating with an implicit $\gamma$ all along, whether or not its risk officers were aware of it.

The second reading concerns cross-regime stability.
Equation~\eqref{eq:calibration-inversion} implies that the implied $\gamma$ should 
 be stable across volatility regimes, since it is a property of the agent's preferences and not of the realized volatility path (Corollary~\ref{cor:clock-invariance}).
A desk's calibrated $\phi$ should scale with $\sigma^2$ across regimes; 
 if it is approximately constant (the wall-clock implementation), the implied $\gamma$ varies as $\sigma^{-2}$, which is a structural misspecification and the failure mode identified by Corollary~\ref{cor:clock-invariance}.

\subsection{Robust-Optimization Interpretation}
\label{subsec:robust}

The entropic functional admits a classical dual representation as a robust  optimization over alternative probability measures with a relative-entropy penalty \citep{FollmerSchied2016, DelbaenPengRosazzaGianin2010}.
In our market-making context, the dual has substantive content beyond the static one.

\begin{corollary}[Robust-Optimization Interpretation]
\label{cor:robust}
Under the hypotheses of  Theorem~\ref{thm:forced-uniqueness}, the forced preference functional admits the dual representation
\begin{equation}
    J_t(\pi)  =  \operatorname*{ess\,inf}_{\Q \ll \Prob} \left\{ \E_\Q\bigl[\WL_T(\pi) \,\big|\, \F_t\bigr] + \frac{1}{\gamma}\, H_t(\Q  \re  \Prob) \right\},
    \label{eq:robust-dual}
\end{equation}
where the essential infimum runs over probability measures $\Q$ absolutely continuous with respect to $\Prob$, and
\begin{equation*}
    H_t(\Q   \re  \Prob)  \coloneqq  \E_\Q\!\left[\log \tfrac{d\Q}{d\Prob} \,\Big|\, \F_t\right]
\end{equation*}
is the conditional relative entropy of $\Q$ with respect to $\Prob$ given $\F_t$.
The unique minimizing measure $\Q^*$ in~\eqref{eq:robust-dual} has $\F_T$-measurable density (with $\F_t$-measurable normalizing constant)
\begin{equation*}
    \frac{d\Q^*}{d\Prob}  =  \frac{e^{-\gamma\, \WL_T(\pi)}}{\E[e^{-\gamma\, \WL_T(\pi)} \mid \F_t]}.
\end{equation*}
At $t = 0$, the $\sigma$-algebra $\F_0$ is 
trivial and the density reduces to $d\Q^*/d\Prob = e^{-\gamma \WL_T(\pi)}/\E[e^{-\gamma \WL_T(\pi)}]$,  the static entropic-tilted measure.
\end{corollary}

\begin{proof}
The representation is an application of the conditional Donsker--Varadhan variational identity. 
Indeed, we have 
\begin{equation*}
    -\tfrac{1}{\gamma}\log \E\bigl[\exp(-\gamma W) \,\big|\, \F_t\bigr]  =  \operatorname*{ess\,inf}_{\Q \ll \Prob} \left\{ \E_\Q[W \mid \F_t] + \tfrac{1}{\gamma} H_t(\Q  \re  \Prob) \right\},
\end{equation*}
applied with $W = \WL_T(\pi)$.
The conditional minimizer is identified by the conditional first-order condition, 
namely that $d\Q^*/d\Prob$ is proportional to $e^{-\gamma W}$ with the proportionality factor 
$1/\E[e^{-\gamma W} \mid \F_t]$ chosen to make $\Q^*$ a probability measure conditionally on $\F_t$.
See \citet[Proposition 4.33]{FollmerSchied2016} for the conditional version.
\end{proof}

Corollary~\ref{cor:robust} has a sharp and economically meaningful interpretation:
  the forced preference functional is mathematically identical to a robust optimization in which the market maker entertains all $\Q \ll \Prob$, evaluates each strategy with a conditional-relative-entropy penalty at rate $1/\gamma$, and takes the essential infimum, with $\Q^* \neq \Prob$ corresponding to her implicit pessimistic measure for the strategy under evaluation.
Note that $\Q^*$ depends on $\pi$ through $W = \WL_T(\pi)$, so the dual is a strategy-by-strategy max-min restatement of the forced functional rather than an adversarial game against a fixed Knightian prior.
This places the paper in conversation with the variational-preferences literature of \citet{MaccheroniMarinacciRustichini2006} and resolves the apparent tension with the single-reference-measure commitment of law-invariance (Remark~\ref{rem:law-inv-trader}).

\begin{remark}[On the Relation to $\Phi$-Divergence Risk Measures]
\label{rem:why-kl} 
 The relative-entropy penalty in Corollary~\ref{cor:robust} is the Kullback-Leibler divergence, one member of the broader family of $\Phi$-divergences (R\'enyi, Tsallis, $\chi^2$, and others). Each alternative would generate a different robust-optimization dual and a different family of risk measures (the so-called $\Phi$- or $f$-divergence risk measures; see \citet[Chapter 4]{FollmerSchied2016}).
Our forced-uniqueness theorem  identifies the conditional entropic functional, which 
by Corollary~\ref{cor:robust} corresponds to a robust optimization with the KL divergence as  ambiguity penalty, 
as the unique preference functional consistent with the five axioms.
A market maker who satisfies our axioms therefore selects KL, rather than some other $\Phi$-divergence, 
 as the implicit ambiguity penalty on her beliefs.
We do not, in this paper, prove the converse axiomatic characterization at the level of $\Phi$-divergences themselves
 (i.e., that KL is the unique $\Phi$-divergence whose conjugate dual produces a functional satisfying our axiom system); 
doing so would require a parallel analysis at the level of divergence functionals rather than risk measures,
 and is an interesting question that we leave open.
\end{remark}

\subsection{Dilation and Position-Size Nonlinearity}
\label{subsec:dilation}

A coherent risk measure, in the sense of \citet{ArtznerDelbaenEberHeath1999}, is positively homogeneous:  
 doubling all positions doubles the risk. 
The entropic functional is convex and monetary but is not positively homogeneous. 
 This failure of homogeneity has a clean market-making interpretation, developed in the corollary and discussion below.

\begin{corollary}[Dilation and Position-Size Nonlinearity]
\label{cor:dilation}
Under the same hypotheses as Theorem~\ref{thm:forced-uniqueness}, the forced risk measure $\rho_t = -J_t$ is not positively homogeneous on $\Linf(\F_T)$ in general.
For any $W \in \Linf(\F_T)$ that is conditionally nondegenerate given $\F_t$ on a set of positive $\Prob$-measure,   and any $\lambda > 0$ with $\lambda \neq 1$,
\begin{equation*}
    \rho_t(\lambda W) \neq \lambda \rho_t(W)
\end{equation*}
with positive probability.
Specifically, the rescaling identity $\rho_t(\lambda W) = \lambda \rho_t^{(\lambda\gamma)}(W)$ holds, where $\rho_t^{(\lambda\gamma)}$ denotes the entropic risk measure at parameter $\lambda\gamma$.
\end{corollary}

\begin{proof}
A direct calculation gives
\begin{equation*}
    \rho_t(\lambda W)  =  \tfrac{1}{\gamma}\log\E[e^{-\gamma\lambda W} \mid \F_t]  =  \lambda \cdot \tfrac{1}{\lambda\gamma}\log\E[e^{-(\lambda\gamma) W} \mid \F_t]  =  \lambda \, \rho_t^{(\lambda\gamma)}(W).
\end{equation*}
For any $\gamma' > 0$, let 
\begin{equation*}
    f_t(\gamma')  \coloneqq  -\tfrac{1}{\gamma'}\log\E[e^{-\gamma' W} \mid \F_t]
\end{equation*}
denote the conditional entropic certainty-equivalent of $W$ at risk-aversion $\gamma'$. 
This is the reduced functional $\tilde J_t^{(\gamma')}(W)$ from the proof of Theorem~\ref{thm:forced-uniqueness}. 
Associated to each $\gamma'$ is the tilted measure $\Q^*_t = \Q^*_t(\gamma')$ with $\F_T$-density 
\begin{equation*}
    \frac{d\Q^*_t}{d\Prob}  =  \frac{e^{-\gamma' W}}{\E[e^{-\gamma' W} \mid \F_t]},
\end{equation*}
and its conditional relative entropy 
$H_t(\Q^*_t \re \Prob) = \E_{\Q^*_t}\bigl[\log(d\Q^*_t/d\Prob) \mid \F_t\bigr]$. 
A direct calculation, analogous to Sub-step~3b of the proof of Theorem~\ref{thm:forced-uniqueness}, then gives
\begin{equation*}
    \frac{d}{d\gamma'}f_t(\gamma')  =  -\frac{1}{(\gamma')^2}\,H_t(\Q^*_t \re \Prob).
\end{equation*}
For $W$ conditionally nondegenerate given $\F_t$ on a set $B \in \F_t$ of positive $\Prob$-measure, 
the conditional Jensen inequality implies that  $H_t(\Q^*_t  \re \Prob) > 0$ almost surely on $B$, so $\frac{d}{d\gamma'}f_t(\gamma') < 0$ on $B$, and $f_t$ is strictly decreasing in $\gamma'$ on $B$.
Hence, $\rho_t^{(\lambda\gamma)}(W) \neq \rho_t^{(\gamma)}(W)$ on a set of positive $\Prob$-measure, 
and $\rho_t(\lambda W) \neq \lambda\rho_t(W)$ with positive probability.
\end{proof}

A market maker who scales every position in her book by a factor $\lambda$   does not face $\lambda$ times the risk; 
she faces a risk evaluated under a different (and larger, if $\lambda > 1$) effective risk-aversion parameter $\lambda\gamma$.
The forced functional penalizes large positions super-linearly in size, in a precise and quantifiable way.
The failure has a clean reading through the robust dual of Corollary~\ref{cor:robust}: a market maker who doubles her positions does not double her relative-entropy ball, but the worst-case loss inside that fixed ball grows super-linearly because the linear functional $W \mapsto \E_\Q[W]$ becomes more sensitive to tail-tilts of $\Q$ when $W$ is scaled up.

\subsection{Mean-Variance Incompatibility}
\label{subsec:meanvar}

The mean-variance objective, in the tradition of \citet{Markowitz1952}, 
is widely used in practitioner finance as a tractable proxy for risk-averse expected-utility preferences.
For a market maker, the natural mean-variance objective would be of the form
\begin{equation}
    \E[\WL_T(\pi)] - \tfrac{\lambda}{2}\, \mathrm{Var}\bigl(\WL_T(\pi)\bigr)
    \label{eq:meanvar}
\end{equation}
for some $\lambda > 0$. 
This is conceptually parallel to the entropic functional but uses the variance rather than the cumulant-generating function as the risk-penalty.
Our framework rules this out.

\begin{corollary}[Mean-Variance Incompatibility]
\label{cor:meanvar}
Let the mean-variance preference functional be the dynamic functional defined by
\begin{equation*}
    J_t^{\mathrm{MV}}(W)  \coloneqq  \E[W \mid \F_t] - \tfrac{\lambda}{2}\,\mathrm{Var}(W \mid \F_t),
\end{equation*}
for $W \in \Linf(\F_T)$ and $\lambda > 0$, which is the natural dynamic form of~\eqref{eq:meanvar}.
This functional in fact violates both monotonicity  (M, Proposition~\ref{prop:monotonicity}) and strong   dynamic consistency (Axiom~\ref{ax:J4}).
Consequently, mean-variance preferences are incompatible with our axiom system.
\end{corollary}

\begin{proof}
We exhibit two failures.

\emph{Failure 1: Monotonicity (Proposition~\ref{prop:monotonicity}).}
Consider $W, W' \in \Linf(\F_T)$ satisfying $W \geq W'$ almost surely.
We exhibit an explicit counterexample to $J_t^{\mathrm{MV}}(W) \geq J_t^{\mathrm{MV}}(W')$.
Fix $t = 0$ and let $W' = 0$ (the zero random variable), so $J_0^{\mathrm{MV}}(W') = 0$.
Since $\lambda$ is a fixed preference parameter, choose any nonconstant random variable $\xi$ with   $\xi \geq 0$ almost surely, $\E[\xi] = 1$, and $\mathrm{Var}(\xi) = v > 2/\lambda$.
Such a $\xi$ exists because the variance of nonnegative random variables with mean $1$ is  unbounded above.
Take $W = \xi$.
Then $W \geq W'$ almost surely, but
\begin{equation*}
    J_0^{\mathrm{MV}}(W)  =  1 - \tfrac{\lambda v}{2}  <  0  =  J_0^{\mathrm{MV}}(W'),
\end{equation*}
in violation of monotonicity.
The mean-variance functional thus violates Proposition~\ref{prop:monotonicity} (M).

\emph{Failure 2: Strong dynamic consistency (Axiom~\ref{ax:J4}).}
By the law of total variance,
\begin{equation*}
    \mathrm{Var}(W \mid \F_s)  =  \E[\mathrm{Var}(W \mid \F_t) \mid \F_s] + \mathrm{Var}(\E[W \mid \F_t] \mid \F_s),
\end{equation*}
the cross-term $\mathrm{Var}(\E[W \mid \F_t] \mid \F_s) - \mathrm{Var}(\E[W' \mid \F_t] \mid \F_s)$ has indefinite sign and is not controlled by the time-$t$ ranking of $W, W'$.
We exhibit a concrete two-period counterexample.
Let $s = 0$, $t \in (0,T)$, $\lambda = 1$, 
with $\F_s$ trivial and $\F_t$ generated by a fair coin flip yielding atoms $\omega_+, \omega_-$ each with probability $1/2$.
Refine $\F_T$ so that on $\omega_+$ there are two equiprobable sub-atoms where $W$ takes the values $+1$ and $-1$, while on $\omega_-$ the random variable $W$ is constant equal to  $5$.
Set $W' \equiv -1/2$ deterministically.
(Existence of such a refinement on $(\Omega, \F_T, \Prob)$ is by the non-atomicity of the probability space, as in Sub-step~3b of the proof of Theorem~\ref{thm:forced-uniqueness}.)

At time $t$, the conditional moments are:
on $\omega_+$, $\E[W \mid \F_t] = 0$ and $\mathrm{Var}(W \mid \F_t) = 1$, so $J^{\mathrm{MV}}_t(W)(\omega_+) = 0 - 1/2 = -1/2$;
on $\omega_-$, $\E[W \mid \F_t] = 5$ and $\mathrm{Var}(W \mid \F_t) = 0$, so $J^{\mathrm{MV}}_t(W)(\omega_-) = 5$.
Meanwhile $J^{\mathrm{MV}}_t(W') \equiv -1/2$ identically.
Hence,  $J^{\mathrm{MV}}_t(W) \geq J^{\mathrm{MV}}_t(W')$ almost surely, 
with equality on $\omega_+$ and strict inequality on $\omega_-$.

At time $s = 0$,  we have 
 $\E[W] = \tfrac{1}{2} \cdot 0 + \tfrac{1}{2} \cdot 5 = 5/2$ 
and $\E[W^2] = \tfrac{1}{2}\cdot 1 + \tfrac{1}{2}\cdot 25 = 13$, so $\mathrm{Var}(W) = 13 - 25/4 = 27/4$.
Therefore,
\begin{equation*}
    J^{\mathrm{MV}}_s(W)  =  \tfrac{5}{2} - \tfrac{1}{2}\cdot  \tfrac{27}{4}  =  \tfrac{20}{8}  - \tfrac{27}{8}  =  -\tfrac{7}{8},
\end{equation*}
while $J^{\mathrm{MV}}_s(W') = -1/2 = -4/8$.
Hence, $J^{\mathrm{MV}}_s(W) = -7/8 < -4/8 = J^{\mathrm{MV}}_s(W')$, contradicting the conclusion of (J4).
Fully developed counterexamples of this form are given by \citet{BasakChabakauri2010} and \citet{BjorkMurgoci2014}.
\end{proof}

\begin{remark}[Cash-Additivity and Wealth-Summary Are Not the Issue]
\label{rem:mv-cash-add}
The mean-variance functional satisfies both cash-additivity (Axiom~\ref{ax:J1}) and the wealth-summary property (W) (Proposition~\ref{prop:wealth-summary}), which is worth pointing out  since both are nonobvious.
For cash-additivity, fix $\F_t$-measurable bounded $c$.
The conditional mean shifts by $c$, and the conditional variance is unchanged; 
  since $c$ is $\F_t$-measurable, it acts as a constant in the conditional distribution of $W$ given $\F_t$, so
\begin{equation*}
    \mathrm{Var}(W + c \mid \F_t) 
    = \E[(W + c - \E[W + c \mid \F_t])^2 \mid \F_t] 
    =  \E[(W - \E[W \mid \F_t])^2 \mid \F_t] 
    = \mathrm{Var}(W \mid \F_t).
\end{equation*}
We thus have  $J_t^{\mathrm{MV}}(W + c) = J_t^{\mathrm{MV}}(W) + c$.
Since $J_t^{\mathrm{MV}}(W) = \E[W \mid \F_t] - \tfrac{\lambda}{2}\mathrm{Var}(W \mid \F_t)$ is a function of $W = \WL_T$ alone, the condition ``if $\WL_T(\pi) = \WL_T(\pi')$ almost surely, then $J_t(\pi) = J_t(\pi')$ almost surely'' is satisfied.
The incompatibility with our axiom system thus does not  come from monetary structure  or from path-dependence; 
it comes from the variance term's failure to respect both monotonicity and time-consistency.
\end{remark}

\subsection{CVaR Critique}
\label{subsec:cvar}

The Basel Committee \citep{BCBS2019} adopted expected shortfall as the official market-risk measure for dealer desks, so the market maker on a regulated desk faces a CVaR-based external risk constraint.
 The technical content reduces to that of \citet{ArtznerDelbaenEberHeathKu2007} and \citet{CheriditoStadje2009}.
Our contribution is to derive the conclusion as a corollary of our axiom system and to identify the resulting desk--regulator gap.

\begin{corollary}[CVaR Time-Inconsistency]
\label{cor:cvar}
Any dynamic preference functional that takes the form of a dynamic conditional-value-at-risk  on $\WL_T$, at any fixed confidence level $\alpha \in (0,1)$, violates Axiom~\ref{ax:J4} (strong dynamic consistency).
As a result,  no CVaR-based dynamic objective is consistent with our axiomatic framework.
\end{corollary}

\begin{proof}
We argue this claim directly from the structure of the candidate preference functional $J^{\mathrm{CVaR}}_t \coloneqq -\mathrm{CVaR}_\alpha(\cdot \mid \F_t)$, with the sign chosen so that higher $J^{\mathrm{CVaR}}$ corresponds to the strategy being preferred, matching the convention of Section~\ref{subsec:J-def}.
By \citet{CheriditoDelbaenKupper2006} (see in particular Remark~6.8 therein) and \citet{CheriditoStadje2009}, 
the dynamic conditional CVaR functional fails strong dynamic consistency in the sense of Axiom~\ref{ax:J4}.
A standard counterexample exhibits $W, W' \in \Linf(\F_2)$ on a two-period filtration $\F_0 \subset \F_1 \subset \F_2$ with $J^{\mathrm{CVaR}}_1(W) \geq J^{\mathrm{CVaR}}_1(W')$ almost surely yet $J^{\mathrm{CVaR}}_0(W) < J^{\mathrm{CVaR}}_0(W')$, in direct violation of Axiom~\ref{ax:J4}.
Consequently, no dynamic CVaR functional can satisfy our axiom system, and in particular cannot be the unique functional forced by Theorem~\ref{thm:forced-uniqueness}.
\end{proof}

The practical content of Corollary~\ref{cor:cvar} is the following.
A regulator who sets a CVaR-based market-risk limit and a dealer desk that runs a time-consistent dynamic optimization in our axiomatic sense are using incompatible objectives. 
The static end-of-day CVaR computed by the regulator is generally different from any dynamic CVaR objective the desk could have optimized against intraday.
The desk's internal risk management, under strong time-consistency, is entropic. 
Its external regulatory constraint is CVaR.
The two are not the same.
We make no normative claim about which is right. 
The result identifies a structural gap which has not previously been pointed out in the inventory market making literature.
As we noted in Section~\ref{subsec:J4}, the negative result is specific to strong dynamic consistency. 
A desk willing to weaken (J4) to acceptance- and rejection-consistency in the sense 
 of \citet{Weber2006} could in principle accommodate CVaR-like (shortfall) objectives.

\subsection{The Wealth-Independence Rationale Subsumed}
\label{subsec:wealth-indep}

  \citet{AvellanedaStoikov2008} motivate the adoption of exponential utility through its wealth-independence property:
among expected-utility preferences over terminal wealth, only CARA delivers wealth-independent reservation prices \citep{Merton1971}.
Our framework subsumes the rationale entirely.

 \begin{corollary}[Wealth-Independence]
\label{cor:wealth-indep}
Under the hypotheses of Theorem~\ref{thm:forced-uniqueness}, 
the optimal quote distances chosen by the market maker do not depend on her current cash holdings.
\end{corollary}

\begin{proof}
Write $\pi + c$ for the strategy whose terminal liquidation-adjusted wealth is shifted by $c$,
 in the sense of the cash-injection construction of Lemma~\ref{lem:density}. 
By cash-additivity (Axiom~\ref{ax:J1}), for any $\F_t$-measurable bounded $c$,
\begin{equation*}
    J_t(\pi + c)  =  J_t(\pi) + c.
\end{equation*}
Taking $c = X_t$ (assumed bounded), the dependence of $J_t$  on $X_t$ is an additive shift and does not affect the ranking of strategies. 
The optimal quote distances $\delta^{a,*}, \delta^{b,*}$ are therefore functions of $(t, S_t, q_t)$ and the relevant state of the world, but not of $X_t$.
\end{proof}

\begin{remark}
Wealth-independence follows from cash-additivity alone and does not on its own motivate the CARA form. 
The CARA structure is forced by the full axiom system, as established in Theorem~\ref{thm:forced-uniqueness}.
\end{remark}

The conceptual significance of the corollary  is that the AS rationale for CARA combines two ingredients: 
a normative desideratum (wealth-independent quotes) and a heuristic implication (wealth-independence implies CARA). 
Neither is needed under our axioms. 
Cash-additivity, the much milder condition that ``a sure dollar adds a dollar of value,'' delivers wealth-independence directly. 
The CARA form then arises from the deeper uniqueness theorem, for which wealth-independence reasoning plays no role.

\subsection{The One-Parameter Family of Preferences}
\label{subsec:one-parameter-family}

We close this section with a restatement of the one-parameter-family conclusion,
 emphasizing the structural shape of the preference space.

\begin{corollary}[The One-Parameter Family]
\label{cor:one-parameter-family}
Under the same hypotheses as Theorem~\ref{thm:forced-uniqueness}, the space of preference functionals satisfying the axiom system is the open half-line $(0, \infty)$, parametrized by $\gamma$, with the entropic functional~\eqref{eq:entropic} at parameter $\gamma$ as the unique representative.
The two endpoints of the half-line have the following limiting interpretations.
\begin{enumerate}
    \item As $\gamma \downarrow 0$, the entropic functional converges to the risk-neutral functional,
    \begin{equation*}
        J_t(\pi) \longrightarrow \E\bigl[\WL_T(\pi) \,\big|\, \F_t\bigr].
    \end{equation*}
    \item As $\gamma \uparrow \infty$, the entropic functional converges to the pessimistic worst-case functional,
    \begin{equation*}
        J_t(\pi) \longrightarrow \operatorname{ess\,inf}\bigl[\WL_T(\pi) \,\big|\, \F_t\bigr].
    \end{equation*}
\end{enumerate}
Every preference functional satisfying the axiom system lies strictly between these two extremes.
The strict positivity of $\gamma$ and the finiteness of $\gamma$  are both enforced by the strict-concavity clause (J3b) of Axiom~\ref{ax:J3}.
\end{corollary}

\begin{proof}
For the first limit, we apply a simple Taylor expansion of the exponential to get
$-\tfrac{1}{\gamma}\log\E[e^{-\gamma W} \mid \F_t] = \E[W \mid \F_t] - \tfrac{\gamma}{2}\operatorname{Var}(W \mid \F_t) + O(\gamma^2)$ as $\gamma \to 0$, so the limit is $\E[W \mid \F_t]$.
For the second limit, \citet[Lemma 2.1(i)]{KupperSchachermayer2009} establishes the unconditional fact $\lim_{\gamma\to\infty}\rho_\gamma(X) = \operatorname{ess\,sup}(-X)$ for $X \in  \Linf$, where $\rho_\gamma(X) = \tfrac{1}{\gamma}\log\E[e^{-\gamma X}] $ is the static entropic risk measure.
To obtain the conditional extension
$$  \lim_{\gamma\to\infty} \tfrac{1}{\gamma}\log\E[e^{-\gamma W}\mid\F_t] = -\operatorname{ess\,inf}(W\mid\F_t) $$
for $W \in \Linf(\F_T)$, 
the same argument goes through $\omega$-by-$\omega$ when applied to   the regular conditional distribution of $W$ given $\F_t$ 
(which exists and is supported in $[-\|W\|_\infty, \|W\|_\infty]$);
equivalently, the conditional version is encoded in the dynamic representation of \citet[Theorem 1.10, equations (8)--(9)]{KupperSchachermayer2009}.
Both boundary cases are ruled out by the strict-concavity clause (J3b) of Axiom~\ref{ax:J3}.
At $\gamma = 0$, the linear functional $\E[\cdot]$ gives equality on every pair $\pi, \pi'$ with $\WL_T(\pi) - \WL_T(\pi')$ not a deterministic constant, contradicting (J3b) on such pairs.
At $\gamma = \infty$, the worst-case functional $J_0(\pi) = \operatorname{ess\,inf}\WL_T(\pi)$ 
gives equality on any pair $\pi_0, \pi_1$ with $\WL_T(\pi_1)  = \WL_T(\pi_0) - \varepsilon\mathbf{1}_A$ 
whenever $A$ does not lower the essential infimum.
 In this case,  $\WL_T(\pi_0) - \WL_T(\pi_1)$ 
is not a deterministic constant yet $J_0$ gives equality, contradicting~(J3b).
Equivalently, by Lemma~\ref{lem:relevance-derivable}, both (J3b) and (J1) together imply relevance, 
which fails at $\gamma = \infty$.
\end{proof}

The corollary locates the market maker's preferences in a one-dimensional family bounded by two well-understood limits, 
namely the risk-neutral limit ($\gamma \to 0$), in which she maximizes expected wealth, 
and the worst-case limit ($\gamma \to \infty$), in which she maximizes the guaranteed outcome.
The entire space of admissible preferences interpolates continuously between these two extremes, with larger $\gamma$ corresponding to heavier weighting of bad outcomes.
The naive intuition that a more risk-averse market maker quotes wider is correct, and Corollary~\ref{cor:one-parameter-family} is its rigorous form: 
her risk-aversion sits on a   half-line, and her position on that half-line determines the  entire dynamic preference functional she optimizes, 
  with no second risk-aversion parameter, no separate skewness-aversion parameter, and no path-dependent or state-dependent component.
The single scalar $\gamma$ does all the work.


\section{Supporting Results}
\label{sec:supporting}

In this section we develop two supporting results which were stated, but not proved, in earlier sections.
They appear here rather than alongside the axiomatic results for two reasons: 
  both require explicit HJB machinery not needed for the axiomatic results of Sections~\ref{sec:axioms}--\ref{sec:consequences}, and both specialize the general axiomatic conclusion to concrete model parameters.
The first is the explicit Hamilton-Jacobi-Bellman derivation of the forced coefficient $\phi = \gamma\sigma^2/2$ in the constant-volatility benchmark, invoked in the proof of Corollary~\ref{cor:cj-inconsistency}.
The second is the stochastic-volatility generalization, which uses the Dambis--Dubins--Schwarz time-change theorem and gives precise operational content to Corollary~\ref{cor:clock-invariance} (clock-invariance) in the nontrivial volatility regime.

\subsection{The Forced Coefficient via HJB}
\label{subsec:hjb}

We now provide a direct HJB-based derivation of the forced coefficient $\phi = \gamma\sigma^2/2$ from Corollary~\ref{cor:cj-inconsistency}.
The argument is independent of the axiomatic derivation and provides an alternative route to the same conclusion, suitable for readers more familiar with the stochastic-control tradition.

\begin{proposition}[HJB Derivation of the Forced Coefficient]
\label{prop:forced-phi}
Let the market making setup be as in Section~\ref{sec:setup}, with constant volatility $\sigma > 0$ and quadratic liquidation cost function $L(q) = \tfrac{\kappa}{2} q^2$.
Define the AS-style value function
\begin{equation*}
    V(t, X, q, S)  \coloneqq  \sup_{\pi \in \Pi} \E\!\left[u\bigl(\WL_T(\pi)\bigr)   \,\Big|\, X_t = X,\, q_t = q,\, S_t = S\right], \quad u(x) = -e^{-\gamma x}.
\end{equation*}
Under the same hypotheses as Theorem~\ref{thm:forced-uniqueness}, the value function admits the factorization
\begin{equation}
    V(t, X, q, S)  =  -\exp  \bigl( -\gamma  (X + qS)\bigr) \cdot \psi(t, q),
    \label{eq:V-factor}
\end{equation}
for a function $\psi \colon [0,T]  \times \R \to \R_+$, and the associated certainty-equivalent  function given by 
$\theta(t,q) \coloneqq -\tfrac{1}{\gamma}\log\psi(t,q)$ satisfies the Bellman equation~\eqref{eq:theta-PDE} below, whose running inventory term has coefficient \emph{exactly} $\gamma \sigma^2  q^2 / 2$.
This is the HJB-level realization of the forced relation $\phi  =  \gamma\sigma^2/2$ asserted in Corollary~\ref{cor:cj-inconsistency}.
\end{proposition}

\begin{proof}
We work in the constant-volatility benchmark: $dS_t = \sigma \, dB_t$ for a Brownian motion $B$ and a constant $\sigma > 0$,  with the strategy $\pi = (\delta^a, \delta^b, C) \in \Pi$, wealth and inventory dynamics, and the liquidation cost function $L(q) = \tfrac{\kappa}{2} q^2$ with $\kappa > 0$, all as in Section~\ref{sec:setup}.
We first observe that the cash-injection component $C$ does not enter the optimization in a nontrivial way.
 By Axiom~\ref{ax:J1} (cash-additivity), the value of any cash injection $dC_t$ at time $t$ is exactly the injected amount, so the optimal control problem is translation-invariant in $C$, and we may set $C \equiv 0$ for computing the optimal $(\delta^a, \delta^b)$ and the value function.
The reduced control variable is therefore the pair $(\delta^a_t, \delta^b_t) \in \R^2$, and the HJB equation below is the standard one for the AS quoting problem.

By Theorem~\ref{thm:forced-uniqueness} and the certainty-equivalent identity (Remark~\ref{rem:cara-dynamic}), maximizing $J_t(\pi)$ over $\pi \in \Pi$ is equivalent to maximizing the CARA expected utility $\E[u(\WL_T(\pi)) \mid \F_t]$ with $u(x) = -e^{-\gamma x}$, the value function in the proposition.
The terminal condition is equal to 
 $V(T, X, q, S) = u(X + qS - \tfrac{\kappa}{2}q^2)$.

The structural feature of the CARA-Brownian setup is the factorization~\eqref{eq:V-factor},
 with $\psi$ a function to be determined, and terminal condition $\psi(T, q) = \exp(\gamma \kappa q^2/2)$.
The HJB equation for the AS control problem is
\begin{equation}
    \partial_t V + \tfrac{1}{2}\sigma^2 \partial_{SS} V  + \sup_{\delta^a} \lambda^a(\delta^a) \, \Delta^a V + \sup_{\delta^b} \lambda^b(\delta^b) \, \Delta^b V  =  0,
    \label{eq:HJB}
\end{equation}
where $\Delta^a V$ and $\Delta^b V$ are the impulse operators recording the change in $V$ upon an ask fill and a bid fill, respectively,
\begin{align*}
    \Delta^a V(t, X, q, S) &\coloneqq V(t, X + S + \delta^a, q - 1, S) - V(t, X, q, S), \\
    \Delta^b V(t, X, q, S) &\coloneqq V(t, X - S + \delta^b, q + 1, S) - V(t, X, q, S).
\end{align*}

Substituting the ansatz~\eqref{eq:V-factor}  into~\eqref{eq:HJB}, two structural features emerge.
First, the diffusion term $\tfrac{1}{2}\sigma^2 \partial_{SS}V$ gives $-\tfrac{1}{2}\sigma^2 \gamma^2 q^2 \exp(-\gamma(X+qS))\, \psi$, since $\partial_{SS}V = -\gamma^2 q^2 \exp(-\gamma(X+qS)) \psi$.
This is the term that delivers the $\phi = \gamma\sigma^2/2$ coefficient. 
  Second, the impulse terms factor cleanly. 
  Using $S + (q-1)S = qS$ for the ask impulse (and the symmetric identity for the bid), 
  the wealth-state factor $\exp(-\gamma(X+qS))$ pulls out of both impulse expressions, 
    leaving $\psi$-only quantities.
After dividing through by $-\exp(-\gamma(X+qS))$, the HJB equation reduces to a two-dimensional PDE for $\psi(t,q)$, 
\begin{equation}
    \partial_t \psi + \tfrac{1}{2}\sigma^2 \gamma^2 q^2  \psi + \mathcal{I}^a[\psi] + \mathcal{I}^b[\psi]  =  0,
    \label{eq:psi-PDE}
\end{equation}
where the impulse operators acting on $\psi$ are
\begin{align}
    \mathcal{I}^a[\psi](t,q) & \coloneqq    \sup_{\delta^a}\, \lambda^a(\delta^a) \bigl[e^{-\gamma\delta^a}\psi(t,q-1) - \psi(t,q)\bigr], \notag \\
    \mathcal{I}^b[\psi](t,q) & \coloneqq  \sup_{\delta^b}\, \lambda^b(\delta^b) \bigl[e^{-\gamma\delta^b}\psi(t,q+1) - \psi(t,q)\bigr].
    \label{eq:impulse-ops}
\end{align}
This is the form obtained by substituting the CARA ansatz into the AS HJB equation,  as in standard derivations following \citet[Section 3]{AvellanedaStoikov2008}.

The forced coefficient becomes visible after passing to the certainty-equivalent function $\theta(t, q) = -\tfrac{1}{\gamma} \log \psi(t,q)$, which records the time-$t$ certainty-equivalent reduction, relative to mark-to-market wealth, of holding inventory $q$ and following the optimal strategy thereafter.
The terminal condition is given by $\theta(T,q)  = - \tfrac{\kappa}{2}q^2$.
From $\psi = \exp(-\gamma\theta)$, equation~\eqref{eq:psi-PDE} translates into
\begin{equation}
    \partial_t \theta - \tfrac{\sigma^2\gamma}{2} q^2 + \mathcal{J}^a[\theta] + \mathcal{J}^b[\theta]  =  0,
    \label{eq:theta-PDE}
\end{equation}
where the $\theta$-impulse operators are
\begin{equation*}
    \mathcal{J}^a[\theta](t,q)  =  -\tfrac{1}{\gamma}\, \sup_{\delta^a}\, \lambda^a(\delta^a) \bigl[e^{-\gamma(\delta^a + \theta(t,q-1) - \theta(t,q))} - 1\bigr],
\end{equation*}
and $\mathcal{J}^b[\theta]$ is the symmetric formula with $\theta(t,q-1) \to \theta(t,q+1)$.
Equation~\eqref{eq:theta-PDE} is, structurally, a Bellman equation for $\theta(t,q)$ with 
 a running inventory cost accruing at the rate $  \gamma \sigma^2  q^2 / 2$ per unit time, 
   partially compensated by impulse activity through spread capture.
The coefficient $\gamma\sigma^2  / 2$ is exact in the PDE, 
having entered via the diffusion term $ (\sigma^2\partial_{SS}V )/2$ in~\eqref{eq:HJB} without any approximation.
\end{proof}

\begin{remark}[Two Senses of ``Forced'' in the Coefficient $\phi = \gamma\sigma^2/2$]
\label{rem:two-senses}
Proposition~\ref{prop:forced-phi} and Corollary~\ref{cor:cj-inconsistency} together describe the relation $\phi = \gamma\sigma^2/2$ in two complementary senses, which we distinguish here once and for all.
\begin{itemize}
\item \emph{Exact in the coefficient.}
In the $\theta$-PDE~\eqref{eq:theta-PDE} above,  the running inventory term has coefficient $ \gamma\sigma^2  q^2 / 2$ \emph{exactly}, not as a limit or an approximation.
The coefficient arises from the diffusion term $ ( \sigma^2\partial_{SS}V)/2$ in the HJB equation after substitution of the CARA ansatz, with no Taylor expansion involved.
This is also visible in the closed-form AS value function (Remark~\ref{rem:AS-value-function}), 
which contains the term $- (\gamma q^2 \sigma^2 (T-t))/2$ as an exact term.
\item \emph{Approximate in the functional form.}
The full \emph{CJ objective functional}~\eqref{eq:CJ-objective}, 
with the running integral $\phi\int_0^T q_s^2\, ds$, is a leading-order
 (second-order in inventory magnitude) expansion of the full entropic functional~\eqref{eq:AS-entropic}. 
 Higher-order terms in $q$ are neglected.
The forced coefficient that this approximation must take, however, is $\gamma\sigma^2/2$ exactly.
\end{itemize}
A practitioner who reads Corollary~\ref{cor:cj-inconsistency} as saying
 ``$\phi$ is the leading-order coefficient'' should understand it in the second sense. 
 The coefficient itself is pinned to $\gamma\sigma^2/2$ exactly, by the HJB analysis above.
\end{remark}

\begin{remark}[Consistency with the AS Value Function]
\label{rem:consistency-AS}
In the AS benchmark in the absence of fills
 (i.e., between the impulses of $N^a$ and $N^b$), 
the running cost is the dominant contribution; integrating $ \gamma\sigma^2 q^2 / 2$ over $[t, T]$ at fixed inventory $q$ 
gives $ ( \gamma\sigma^2  q^2\, (T-t))/2$, 
which matches the term $- ( \gamma q^2 \sigma^2 (T-t) )/2$ appearing in the AS value function (see Remark~\ref{rem:AS-value-function}).
\end{remark}

\begin{remark}[The Certainty-Equivalent Function as a Marked-to-Market Liquidity Cost]
\label{rem:theta-mark}
 The function $\theta(t,q)$ has the following operational interpretation.
By construction, $\theta(t,q)$ is the conditional certainty-equivalent reduction, relative to the mark-to-mid wealth $X_t + q S_t$, 
attributable to holding inventory $q$ at time $t$ and following the optimal quoting strategy thereafter.
Equivalently, $\theta(t,q)$ is the cash-equivalent value that the market maker would pay, at time $t$, to be flat
  (i.e., to have $q = 0$ instead of $q$) for the remainder of the horizon,   conditional on continuing to quote optimally.
 Under the forced functional, $\theta(t,q)$  thus plays the role of an inventory-specific liquidity discount: 
the cash amount by which the desk's true valuation of its book falls short of mark-to-mid.
\end{remark}

A desk that wishes to mark its inventory book consistently with the dynamically-consistent preference 
functional we have derived should value the book at $X_t + q S_t + \theta(t,q)$, not the naive mark-to-mid $X_t + q S_t$.
The quantity $\theta(t,q)$ is the certainty-equivalent haircut that the desk's risk-aversion forces, 
and it is strictly negative whenever $q \neq 0$ and $t < T$ 
(i.e., it lowers the marked value of the book relative to mark-to-mid).
The strict negativity is visible directly in the closed-form AS value function (Remark~\ref{rem:AS-value-function}), where 
$\theta(t,q) = - ( \gamma q^2 \sigma^2 (T-t) )/2 $ 
plus impulse-induced corrections of the same sign.
A desk that uses mark-to-mid for inventory while running an AS- or CJ-style quoting strategy is using  two inconsistent valuations of its book: 
one to drive quoting decisions (which implicitly applies the haircut) and another to report P\&L (which does not).
This is a structural-gap observation parallel in shape to Corollaries~\ref{cor:cvar} and~\ref{cor:dilation}: 
a piece of standard practice is, in a precise sense, inconsistent with the time-consistent dynamic functional our axioms force.

\subsection{The Stochastic-Volatility Extension}
\label{subsec:stochvol}

We now allow the price volatility to be stochastic,
 letting $S$ be a continuous semimartingale with $d\langle S\rangle_t = \sigma_t^2 \, dt$, where $\sigma$ is a positive predictable process.
No assumption is made on the dynamics of $\sigma$ beyond  predictability and local boundedness; in particular, the result below covers Heston-type and rough-volatility models.
The natural question is whether the forced coefficient $\phi = \gamma\sigma^2/2$ of Proposition~\ref{prop:forced-phi} extends to this setting.
By Corollary~\ref{cor:clock-invariance} (clock-invariance) and the Dambis--Dubins--Schwarz theorem,   it does, with $\sigma^2$ replaced by the instantaneous $\sigma_t^2$.

\begin{proposition}[Clock-Corrected Penalty]
\label{prop:clock-corrected}
Let the price process satisfy $d\langle S\rangle_t = \sigma_t^2 \, dt$ with $\sigma$ a predictable process such that $\sigma_t > 0$ almost surely for almost every $t \in [0, T]$.
Under the same hypotheses as Theorem~\ref{thm:forced-uniqueness}, the market maker's certainty-equivalent function $\theta(t, q, \omega) = -\tfrac{1}{\gamma}\log \psi(t,q,\omega)$ satisfies a Bellman equation with running inventory cost
\begin{equation}
    \frac{\gamma}{2}\, q^2 \, d\langle S\rangle_t  =  \frac{\gamma \sigma_t^2}{2}\, q^2 \, dt.
    \label{eq:running-stochvol}
\end{equation}
Equivalently, in business-time units $u = \langle S\rangle_t$ (which carries the units of price-variance, not time),
 the running cost is the constant $\gamma q^2 / 2$ per unit of business-time.
\end{proposition}

\begin{proof}
We use a time-change argument in three steps.
Throughout the proof we use the convention that hat-decorated symbols ($\widehat J_u$) denote business-time reparametrizations on the wall-clock strategy space $\Pi$, tilde-decorated symbols ($\tilde\theta, \tilde\psi, \tilde V, \tilde\filt, \tilde S, \tilde M, \tilde\pi, \tilde\lambda^{a,b}, \tilde N^{a,b}, \tilde{\mathcal J}^{a,b}$) denote the analogous objects on the business-time setup, and undecorated symbols denote wall-clock quantities.
The two systems are related by the time-change $u = \Lambda_t$ with inverse $t = \tau(u)$.

\emph{Step 1: Time-change.}
Define $\Lambda_t \coloneqq \langle S\rangle_t = \int_0^t \sigma_s^2\, ds$.
Under the proposition's assumption that $\sigma_t > 0$ for almost every $t \in [0,T]$ almost surely, 
the process $\Lambda$ is continuous
 (by absolute continuity of the integral with respect to its upper limit) 
 and strictly increasing on $[0,T]$ 
 (since $\int_s^t \sigma_u^2\, du > 0$ for every $0 \leq s < t \leq T$ on the full-probability event $\{\sigma_u > 0 \text{ for a.e.\ } u\}$), 
 with $\Lambda_0 = 0$.
The right-continuous inverse $\tau(u) \coloneqq \inf\{t \geq 0 \mid \Lambda_t > u\}$ is a continuous time-change, 
and the time-changed filtration is $\tilde{\F}_u \coloneqq \F_{\tau(u)}$.
Write $S = M + A$ for the canonical decomposition of $S$ as a continuous semimartingale,  with $M$ the continuous local-martingale part (so $\langle M\rangle = \langle S\rangle$) and $A$ the continuous finite-variation part.
By the Dambis--Dubins--Schwarz theorem applied to $M$, the time-changed martingale part $\tilde M_u \coloneqq M_{\tau(u)}$ is a standard Brownian motion with respect to $\tilde\F$, so $\tilde S_u \coloneqq S_{\tau(u)} = \tilde M_u + A_{\tau(u)}$ is a continuous semimartingale whose martingale part is standard Brownian.
In the general stochastic-volatility case the HJB equation has the form
\begin{equation*}
    \partial_t V + \mu_t \, \partial_S V + \tfrac{1}{2}\sigma_t^2 \, \partial_{SS} V + \sup_{\delta^a} \lambda^a(\delta^a)\, \Delta^a V + \sup_{\delta^b} \lambda^b(\delta^b)\, \Delta^b V  =  0,
\end{equation*}
which extends~\eqref{eq:HJB} by the drift term $\mu_t \, \partial_S V$.
Under the CARA ansatz~\eqref{eq:V-factor}, the drift term produces $-\gamma q\, \mu_t \cdot V$, 
which contributes a \emph{linear}-in-$q$ term (the standard alpha-skewing of optimal quotes) to the $\theta$-equation, 
but does not affect the \emph{quadratic}-in-$q$ running inventory cost coming from $\tfrac{1}{2}\sigma_t^2 \, \partial_{SS} V$.
Since the conclusion of the proposition is about the quadratic-in-$q$ coefficient only, the drift is irrelevant for what follows, 
and for notational simplicity we proceed as if $A \equiv 0$. 
The same argument applied to the martingale part $M$ alone yields the identical quadratic coefficient, 
with the drift contributing only a linear-in-$q$ correction in business time.
 The time-change theorem for point processes \citep[Theorem T16, p.\ 41]{Bremaud1981} implies that the counting processes $\tilde{N}^a_u \coloneqq N^a_{\tau(u)}$ and $\tilde{N}^b_u \coloneqq N^b_{\tau(u)}$ are Cox processes under $\tilde\filt$, with business-time intensities given by the chain rule applied to $u = \Lambda_t$.
Recall that for a Cox process whose intensity is specified in one time parameterization,
 switching to a different time parameterization rescales the intensity by the Jacobian of the change of variables. 
Here the wall-clock and business-time parameterizations are related by $u = \Lambda_t$ with inverse $t = \tau(u)$, so the relevant Jacobian is $\tau'(u)$. 
A Cox process with wall-clock intensity $\lambda^{a,b}_t$ therefore has business-time intensity
\begin{equation*}
    \tilde\lambda^{a,b}_u  =  \lambda^{a,b}_{\tau(u)} \cdot \tau'(u)  =  \frac{\lambda^{a,b}_{\tau(u)}}{\Lambda'_{\tau(u)}}  =  \frac{\lambda^{a,b}_{\tau(u)}}{\sigma_{\tau(u)}^2},
\end{equation*}
where the second equality uses $\Lambda'_t = \sigma_t^2$, and the third uses $\tau'(u) = 1/\Lambda'_{\tau(u)}$, 
the inverse-function derivative for $\tau = \Lambda^{-1}$.
A strategy $\pi = (\delta^a, \delta^b, C) \in \Pi$ maps under the time-change to 
$\tilde\pi \coloneqq (\tilde\delta^a, \tilde\delta^b, \tilde C)$ with $\tilde\delta^{a,b}_u \coloneqq \delta^{a,b}_{\tau(u)}$ and $\tilde C_u \coloneqq C_{\tau(u)}$; 
this map is a bijection between $\Pi$ (defined with respect to $\filt$ on $[0,T]$) 
and the analogous admissible strategy space with respect to $\tilde\filt$ on $[0, \Lambda_T]$, 
since the admissibility conditions transfer pointwise under the continuous time-change.
The liquidation-adjusted terminal wealth is preserved, 
$\widetilde{\WL}_{\Lambda_T}(\tilde\pi) = \WL_T(\pi)$.

\emph{Step 2: HJB in business-time.}
By Corollary~\ref{cor:clock-invariance}, the business-time reparametrization  $\widehat J_u(\pi) \coloneqq J_{\tau(u)}(\pi)$ has the entropic form on $\tilde\filt$ with the same parameter $\gamma$ as $J$, hence is stationary in $\Lambda$ in the sense that the five axioms of Section~\ref{sec:axioms}, together with the three derived properties, hold for $\widehat J$ on $\tilde\filt$ with $t$ replaced by $u$ throughout.
The HJB analysis of Section~\ref{subsec:hjb} therefore applies directly in business-time, with $\sigma$ replaced by $1$ (since $\tilde M$ is standard Brownian) and arrival intensities $\tilde\lambda^{a,b}$ as in Step 1.
Substituting the ansatz $\tilde V(u, X, q, \tilde S) = -\exp(-\gamma(X+q\tilde S))\, \tilde\psi(u, q)$ and passing to $\tilde\theta(u, q) = -\tfrac{1}{\gamma}\log\tilde\psi(u,q)$, the business-time Bellman equation reads
\begin{equation}
    \partial_u \tilde\theta - \frac{\gamma}{2}\, q^2 + \tilde{\mathcal J}^a[\tilde\theta] + \tilde{\mathcal J}^b[\tilde\theta]  =  0,
    \label{eq:theta-business}
\end{equation}
with running cost rate $\tfrac{\gamma}{2} \, q^2$ per unit of business-time.

\emph{Step 3: Pull-back to wall-clock.}
Set $\theta(t,q,\omega) \coloneqq \tilde\theta(\Lambda_t(\omega), q, \omega)$, where the $\omega$-dependence on the right enters through $\Lambda_t(\omega)$ and through the path-dependence of $\tilde\theta$ on the realized intensities $\tilde\lambda^{a,b}(\omega)$.
The chain rule gives $\partial_t \theta = \sigma_t^2 \, \partial_u \tilde\theta$.
The impulse operators pull back as follows.
By Step~1, the business-time intensities satisfy $\tilde\lambda^{a,b}_u = \lambda^{a,b}_{\tau(u)}/\sigma_{\tau(u)}^2$, equivalently $\sigma_t^2 \, \tilde\lambda^{a,b}_{\Lambda_t} = \lambda^{a,b}_t$.
Hence, the business-time impulse contribution $\tilde\lambda^{a,b}_{\Lambda_t} \cdot \tilde{\mathcal{J}}^{a,b}[\tilde\theta]$, when multiplied by $\sigma_t^2$, becomes $\lambda^{a,b}_t \cdot \mathcal{J}^{a,b}[\theta]$, the wall-clock impulse contribution, since $\tilde{\mathcal{J}}^{a,b}[\tilde\theta] = \mathcal{J}^{a,b}[\theta]$ at corresponding $(u, t)$ pairs (the impulse acts on the $q$-argument, not on time).
The $\sigma_t^2$ factor here is the BSDE-level consequence of clock-invariance noted in Remark~\ref{rem:clock-BSDE}.
Multiplying~\eqref{eq:theta-business} by $\sigma_t^2$ and using these identities yields
\begin{equation*}
    \partial_t \theta - \frac{\gamma \sigma_t^2}{2}\, q^2 + \lambda^a_t \, \mathcal{J}^a[\theta] + \lambda^b_t \, \mathcal{J}^b[\theta]  =  0,
\end{equation*}
which corresponds to the wall-clock  running cost rate~\eqref{eq:running-stochvol}.
\end{proof}

In words,  the running inventory cost is constant in business-time and proportional to $\sigma_t^2$ in wall-clock time.
 A volatility spike proportionally raises the maker's inventory penalty rate; a quiet regime keeps it low.
The axiomatic framework forces this empirically natural behavior to emerge.

\paragraph{CJ in stochastic volatility, and the connection to rough-volatility market making.}
The Cartea--Jaimungal running penalty $\int_0^T q_s^2 \, ds$, applied uncritically in a stochastic-volatility setting, is inconsistent with the clock-invariance of Corollary~\ref{cor:clock-invariance}.
The clock-coherent version is $\int_0^T q_s^2 \, d\langle S\rangle_s = \int_0^T \sigma_s^2 q_s^2 \, ds$, which scales the inventory penalty with the instantaneous variance rate. 
This is the form Proposition~\ref{prop:clock-corrected} forces.

This prescription connects directly to the multi-asset framework  of \citet{RosenbaumZhang2022}, where the inventory risk is expressed in the quadratic-variation clock of an SPX underlying under the quadratic rough Heston model.
In our notation, their objective specializes  Proposition~\ref{prop:clock-corrected} to the rough Heston price process; what our framework adds is an axiomatic foundation for that scaling.
The Dambis--Dubins--Schwarz theorem applies as soon as $S$ has a continuous local-martingale part.
 The forced relation $\phi_t =  \gamma\sigma_t^2/2$ therefore survives in rough-volatility models (Bergomi, Stein--Stein, log-modulated, and the like) without any model-specific re-derivation. 
It is a consequence of the axiom system, not of any particular volatility model.
Any alternative scaling (wall-clock CJ, ad-hoc volatility-dependent heuristic) would be inconsistent with the structural clock-invariance of Corollary~\ref{cor:clock-invariance}, hence outside the axiom class.

The risk-aversion parameter $ \gamma$ is a single positive scalar by   Theorem~\ref{thm:forced-uniqueness} and is constant across the realized variance path; a market maker who wishes to allow $\gamma$ to vary across regimes would need to enrich the framework by indexing the preference functional over a state variable, which we do not develop here.


\section{Discussion}
\label{sec:discussion}

\subsection{Operational Implications}
\label{subsec:operational}

In operational terms, the framework delivers three concrete  things to a practitioner.

First, the running-penalty coefficient $\phi$ in any CJ-style implementation is not a free parameter.
It should be set as $\phi_t = \gamma\sigma_t^2/2$, with $\sigma_t$ the instantaneous volatility, 
so that the running penalty against the quadratic-variation clock has the constant coefficient $\gamma/2$.
Calibrating $\phi$ independently of $\gamma$ is over-parametrizing the preference side of the model.

 Second, the implied risk aversion $\gamma$ is recoverable from observed quoted spreads via the calibration-inversion of Corollary~\ref{cor:calibration-inversion}: 
  $\gamma = 2\phi/\sigma^2$.
For a CJ-tradition desk, this inversion provides a sanity check, since the $\gamma$ implied by the desk's running-penalty calibration should match, up to noise, the $\gamma$ implied by an AS-style fit to the same data.
Persistent disagreement between the two is a signal of miscalibration or of a preference-functional inconsistency.

Third, the same forced relation has an implementation consequence in stochastic-volatility settings: 
the inventory penalty must be integrated against the quadratic-variation clock, not wall-clock time. 
The clock-coherent CJ running cost is $\tfrac{\gamma}{2}\int_0^T q_s^2 \, d\langle S\rangle_s = \tfrac{\gamma}{2}\int_0^T \sigma_s^2 q_s^2 \, ds$, 
not $\phi\int_0^T q_s^2 \, ds$ with constant $\phi$. 
A wall-clock penalty with constant coefficient systematically under-charges inventory on high-volatility days and over-charges it on low-volatility days.

These three prescriptions are derived entirely from the underlying axioms and hold independently of any additional modeling assumptions.
A desk that accepts the five axioms
 (cash-additivity, normalization, concavity, strong dynamic consistency,  and law-invariance)
  is committed to all three, 
  whether or not it finds the axiomatic framing intuitive. 
  Figure~\ref{fig:phi-gamma-sigma} summarizes the constraint geometry and the stochastic-volatility correction. 
  
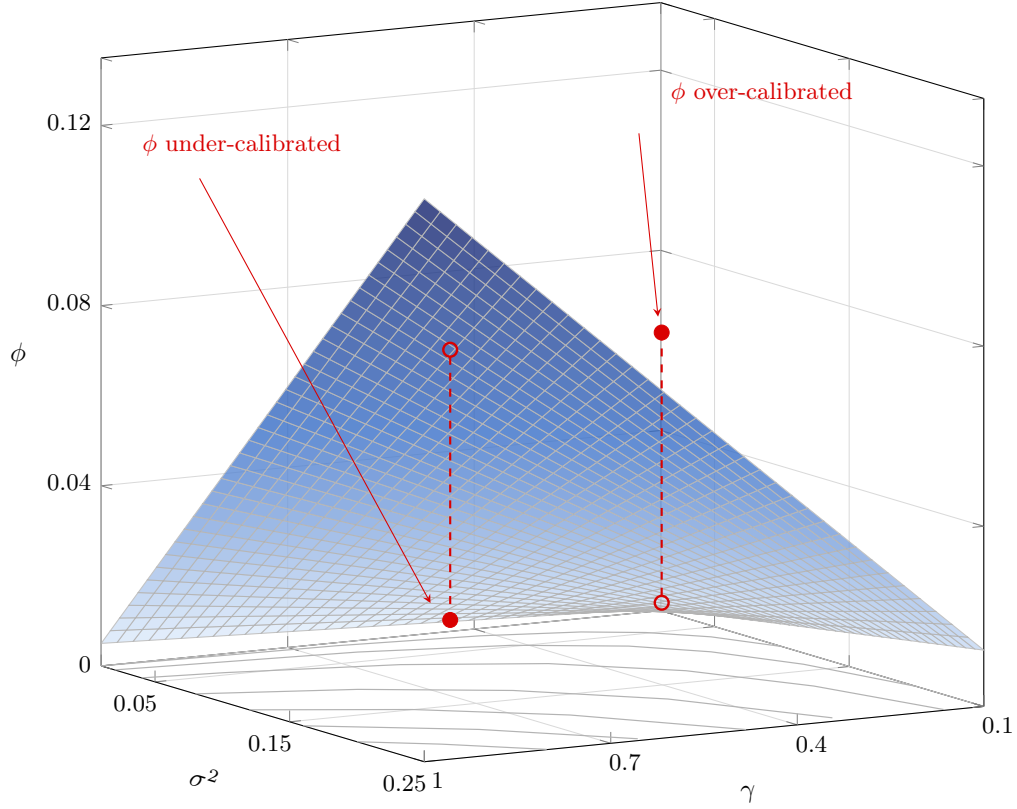
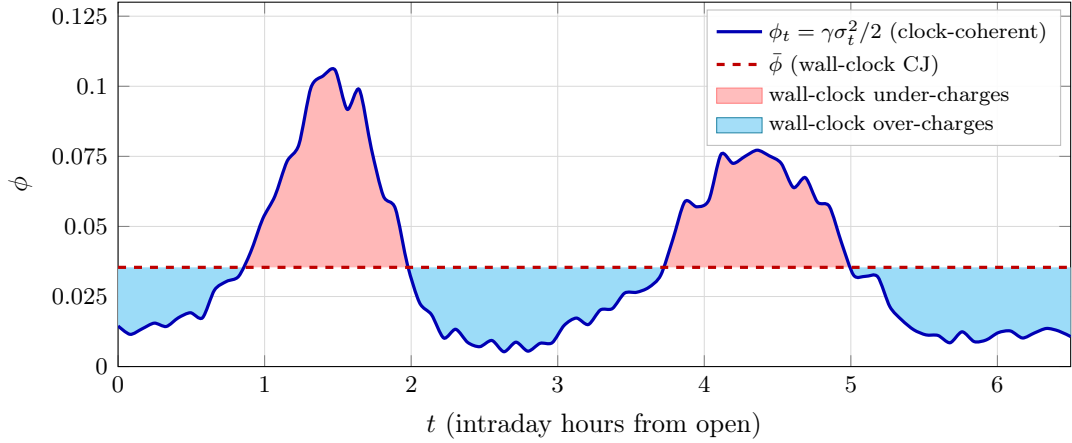
\begin{figure}[!htbp]
\centering
\begin{subfigure}[t]{\textwidth}
\centering
\begin{tikzpicture}
\begin{axis}[
    width=0.86\textwidth,
    height=7.6cm,
    view={150}{20},
    z post scale=2.0,
    xlabel={$\gamma$},
    ylabel={$\sigma^{2}$},
    zlabel={$\phi$},
    xlabel style={sloped, font=\small},
    ylabel style={sloped, font=\small},
    zlabel style={rotate=-90, font=\small},
    xmin=0.1, xmax=1.0,
    ymin=0.01, ymax=0.25,
    zmin=0, zmax=0.135,
    xtick={0.1,0.4,0.7,1.0},
    ytick={0.05,0.15,0.25},
    ztick={0,0.04,0.08,0.12},
    xticklabel style={/pgf/number format/.cd, fixed, precision=2, /tikz/.cd, font=\footnotesize},
    yticklabel style={/pgf/number format/.cd, fixed, precision=2, /tikz/.cd, font=\footnotesize},
    zticklabel style={/pgf/number format/.cd, fixed, precision=2, /tikz/.cd, font=\footnotesize},
    grid=both,
    major grid style={line width=.1pt,draw=gray!30},
    clip=false,
    colormap={blueheat}{rgb=(0.92,0.96,1.0) rgb=(0.20,0.42,0.78) rgb=(0.05,0.10,0.40)},
    point meta=z,
    domain=0.1:1.0,
    y domain=0.01:0.25,
    samples=36,
    samples y=36,
]
\addplot3[gray!60, line width=0.45pt, no marks] coordinates {(0.100,0.1000,0) (0.146,0.0686,0) (0.207,0.0484,0) (0.283,0.0353,0) (0.420,0.0238,0) (0.634,0.0158,0) (1.000,0.0100,0)};
\addplot3[gray!60, line width=0.45pt, no marks] coordinates {(0.100,0.2000,0) (0.146,0.1372,0) (0.207,0.0967,0) (0.283,0.0707,0) (0.420,0.0476,0) (0.634,0.0316,0) (1.000,0.0200,0)};
\addplot3[gray!60, line width=0.45pt, no marks] coordinates {(0.161,0.2484,0) (0.222,0.1802,0) (0.314,0.1276,0) (0.451,0.0887,0) (0.634,0.0631,0) (0.878,0.0456,0) (1.000,0.0400,0)};
\addplot3[gray!60, line width=0.45pt, no marks] coordinates {(0.329,0.2433,0) (0.451,0.1774,0) (0.603,0.1326,0) (0.786,0.1017,0) (1.000,0.0800,0)};
\addplot3[gray!60, line width=0.45pt, no marks] coordinates {(0.481,0.2493,0) (0.634,0.1893,0) (0.802,0.1497,0) (1.000,0.1200,0)};
\addplot3[gray!60, line width=0.45pt, no marks] coordinates {(0.649,0.2465,0) (0.802,0.1996,0) (1.000,0.1600,0)};
\addplot3[gray!60, line width=0.45pt, no marks] coordinates {(0.802,0.2495,0) (0.908,0.2201,0) (1.000,0.2000,0)};
\addplot3 [
    surf, opacity=0.78,
    faceted color=black!30,
    z buffer=sort,
    mesh/ordering=y varies,
] {x*y/2};
\addplot3 [only marks, mark=*, mark size=2.8pt, color=red!85!black]
    coordinates {(0.25, 0.08, 0.070)};
\addplot3 [only marks, mark=o, mark size=2.6pt, mark options={line width=0.9pt}, color=red!85!black]
    coordinates {(0.25, 0.08, 0.010)};
\addplot3 [color=red!85!black, dashed, line width=0.8pt]
    coordinates {(0.25, 0.08, 0.070) (0.25, 0.08, 0.010)};
\addplot3 [only marks, mark=*, mark size=2.8pt, color=red!85!black]
    coordinates {(0.85, 0.20, 0.025)};
\addplot3 [only marks, mark=o, mark size=2.6pt, mark options={line width=0.9pt}, color=red!85!black]
    coordinates {(0.85, 0.20, 0.085)};
\addplot3 [color=red!85!black, dashed, line width=0.8pt]
    coordinates {(0.85, 0.20, 0.025) (0.85, 0.20, 0.085)};
\node[red!85!black, font=\footnotesize, anchor=west]
    at (axis cs:0.10, 0.01, 0.115) {$\phi$ over-calibrated};
\draw[red!85!black, line width=0.4pt, ->, >=stealth]
    (axis cs:0.20, 0.04, 0.110) -- (axis cs:0.245, 0.075, 0.073);
\node[red!85!black, font=\footnotesize, anchor=west]
    at (axis cs:0.95, 0.01, 0.115) {$\phi$ under-calibrated};
\draw[red!85!black, line width=0.4pt, ->, >=stealth]
    (axis cs:0.95, 0.06, 0.112) -- (axis cs:0.86, 0.19, 0.028);
\end{axis}
\end{tikzpicture}
\caption{The forced relation $\phi  = \gamma\sigma^{2}/2$ as a surface in $(\gamma, \sigma^{2}, \phi)$ space, with the hyperbolic level curves $\phi = c$  projected onto the $(\gamma, \sigma^{2})$ floor.
Any axiom-consistent calibration lies on this surface.
The two filled red dots show off-surface desk calibrations, one with $\phi$ over-calibrated and one with $\phi$ under-calibrated; 
each is connected by a dashed segment to the corresponding open red circle on the surface, which marks the forced value.}
\label{fig:constraint-surface}
\end{subfigure}

\vspace{1.4em}

\begin{subfigure}[t]{\textwidth}
\centering
\begin{tikzpicture}
\begin{axis}[
    width=0.92\textwidth,
    height=6.4cm,
    xlabel={$t$ (intraday hours from open)},
    ylabel={$\phi$},
    label style={font=\small},
    xmin=0, xmax=6.5,
    ymin=0, ymax=0.13,
    xtick={0,1,2,3,4,5,6},
    ytick={0,0.025,0.05,0.075,0.10,0.125},
    xticklabel style={font=\footnotesize},
    yticklabel style={/pgf/number format/.cd, fixed, precision=3, /tikz/.cd, font=\footnotesize},
    grid=both,
    major grid style={line width=.1pt,draw=gray!30},
    legend style={
        at={(0.99,0.98)},
        anchor=north east,
        font=\scriptsize,
        draw=gray!50,
        fill=white,
        fill opacity=0.92,
        text opacity=1,
    },
    legend cell align=left,
]
\addplot[blue!70!black, very thick, smooth, name path=A] coordinates {
(0.000,0.0145) (0.082,0.0115) (0.165,0.0135) (0.247,0.0155) (0.329,0.0143) (0.411,0.0174) (0.494,0.0192) (0.576,0.0174) (0.658,0.0273) (0.741,0.0304) (0.823,0.0322) (0.905,0.0409) (0.987,0.0527) (1.070,0.0608) (1.152,0.0731) (1.234,0.0792) (1.316,0.0995) (1.399,0.1037) (1.481,0.1058) (1.563,0.0918) (1.646,0.0988) (1.728,0.0778) (1.810,0.0609) (1.892,0.0564) (1.975,0.0362) (2.057,0.0228) (2.139,0.0185) (2.222,0.0102) (2.304,0.0133) (2.386,0.0086) (2.468,0.0071) (2.551,0.0093) (2.633,0.0053) (2.715,0.0087) (2.797,0.0055) (2.880,0.0083) (2.962,0.0085) (3.044,0.0147) (3.127,0.0173) (3.209,0.0150) (3.291,0.0202) (3.373,0.0207) (3.456,0.0262) (3.538,0.0265) (3.620,0.0282) (3.703,0.0326) (3.785,0.0453) (3.867,0.0588) (3.949,0.0570) (4.032,0.0595) (4.114,0.0757) (4.196,0.0725) (4.278,0.0750) (4.361,0.0772) (4.443,0.0752) (4.525,0.0725) (4.608,0.0639) (4.690,0.0674) (4.772,0.0587) (4.854,0.0571) (4.937,0.0441) (5.019,0.0326) (5.101,0.0322) (5.184,0.0320) (5.266,0.0214) (5.348,0.0166) (5.430,0.0131) (5.513,0.0113) (5.595,0.0111) (5.677,0.0085) (5.759,0.0124) (5.842,0.0089) (5.924,0.0095) (6.006,0.0121) (6.089,0.0127) (6.171,0.0102) (6.253,0.0120) (6.335,0.0136) (6.418,0.0127) (6.500,0.0107)
};
\addlegendentry{$\phi_t = \gamma\sigma_t^{2}/2$ (clock-coherent)}
\addplot[red!75!black, very thick, dashed, name path=B] coordinates {
    (0,0.0354) (6.5,0.0354)
};
\addlegendentry{$\bar\phi$ (wall-clock CJ)}
\addplot[red!28, draw=none, forget plot] fill between[
    of=A and B, soft clip={domain=0:6.5}, split,
    every segment no 0/.style={fill=red!28},
    every segment no 2/.style={fill=red!28},
    every segment no 4/.style={fill=red!28},
    every segment no 6/.style={fill=red!28},
    every even segment/.style={fill=none},
];
\addplot[cyan!35, draw=none, forget plot] fill between[
    of=A and B, soft clip={domain=0:6.5}, split,
    every segment no 1/.style={fill=cyan!35},
    every segment no 3/.style={fill=cyan!35},
    every segment no 5/.style={fill=cyan!35},
    every segment no 7/.style={fill=cyan!35},
    every odd segment/.style={fill=none},
];
\addlegendimage{area legend, fill=red!28, draw=red!50}
\addlegendentry{wall-clock under-charges}
\addlegendimage{area legend, fill=cyan!35, draw=cyan!60!black}
\addlegendentry{wall-clock over-charges}
\end{axis}
\end{tikzpicture}
\caption{Clock-coherent inventory penalty rate $\phi_t  = \gamma\sigma_t^{2}/2$ on a stylized intraday volatility path with two bursts, against the wall-clock CJ constant $\bar\phi$ calibrated to the daily average.
The wall-clock constant under-charges inventory during high-volatility regimes and over-charges it during quiet periods, inconsistent with the clock-invariance of Corollary~\ref{cor:clock-invariance} and corrected by Proposition~\ref{prop:clock-corrected}.}
\label{fig:clock-correction}
\end{subfigure}

\caption{The forced relation $\phi = \gamma\sigma^{2}/2$ in two views.
(a) The constraint surface in parameter space; off-surface calibrations are over-parametrized.
(b) The stochastic-volatility version, $\phi_t = \gamma\sigma_t^{2}/2$, tracks realized variance pointwise in time, while a wall-clock constant systematically misprices inventory.}
\label{fig:phi-gamma-sigma}
\end{figure}

\subsection{Drawdown-Averse Market Makers}
\label{subsec:drawdown}

Proposition~\ref{prop:wealth-summary} commits us to preferences that depend on the strategy only through liquidation-adjusted terminal wealth $\WL_T$.
Path-functional preferences (aversion to maximum drawdown $\sup_{s \leq T} (\WL_0 - \WL_s)$, to time-average wealth, or to drawdown over a rolling window) are genuinely outside this framework and cannot be absorbed into $\WL_T$ by any reparametrization.

This is a real limitation.
Proprietary trading firms typically impose intraday drawdown limits on their market makers, 
whereby if cumulative losses since the start of the trading day exceed a threshold, the desk is stopped out and required to flatten its book.
Such limits are externally imposed constraints, not preferences in the usual sense, but the behavior they induce (more conservative quoting after a losing streak) is genuinely path-functional and is not captured by any $\WL_T$-based functional.
We note that an externally imposed intraday drawdown constraint, viewed as a hard rule on admissible strategies, changes the effective strategy space $\Pi$ and hence the optimization problem; it need not change the preference functional $J$ itself.
The limitation of our framework is therefore about \emph{preference modeling}, not about whether hard constraints can be imposed on the strategy space: 
such constraints can be added orthogonally to $\Pi$, but a market maker whose preferences are themselves drawdown-averse is outside the wealth-summary framework regardless.

 The result is best read as a uniqueness theorem within the class of $\WL_T$-based preferences, a substantive (but not universal) modeling commitment.
Note also the scoping distinction from optimal execution \citep{AlmgrenChriss2000}, which is a cost-minimization problem with fixed terminal target $q_T = 0$; our forced-uniqueness theorem does not apply to that setting.
The right axiomatic treatment of a genuinely drawdown-averse market maker is an open question. 
Relevant entry points are \citet{ChernyMadan2009} on path-dependent acceptability indices and \citet{CheriditoDelbaenKupper2006} on dynamic monetary risk measures for bounded discrete-time processes.

\subsection{Relation to Gu\'eant (2017)}
\label{subsec:gueant}

The strongest prior unification of the AS and CJ frameworks is that of \citet{Gueant2017}, who shows that in both his Model~A (AS with CARA expected-utility maximization) and Model~B (CJ with a running inventory penalty),  the value function $v(t,q)$ admits an ansatz reducing the HJB equation to a linear tridiagonal ODE system, and the two ODE systems are related by an explicit change of variables.
In this ODE-level sense the two frameworks are reconciled, in that the same machinery solves both, 
and a CJ-style problem with a given $\phi$ can be mapped to an AS-style problem with a corresponding $\gamma$.

Our contribution sits at a different structural level.
Gu\'eant unifies the two frameworks at the level of the value function and the ODE it satisfies,  by direct calculation in a specific HJB problem;
the present paper unifies them at the level of the preference functional itself, by axiomatic characterization.
The difference has three consequences.
First, parameter-forcing rather than parameter-mapping: 
Gu\'eant's result maps between two free-parameter model families, while ours identifies a constraint surface 
  $\{\phi = \gamma\sigma^2/2,\ \alpha = \tfrac{1}{2}L''(0)\}$ 
   on which any axiom-consistent preference must lie. 
A desk off the surface is over-parametrized.
Second, functional-level rather than HJB-level: Gu\'eant's reconciliation requires a specific HJB  problem with specific intensities, whereas the axiomatic result holds whenever the axioms apply (Proposition~\ref{prop:forced-phi} of Section~\ref{sec:supporting} is a specialization rather than the primary content).
Third, further structural consequences (the robust-optimization dual, position-size nonlinearity, CVaR incompatibility) are visible at the functional level and invisible at the ODE level.

\subsection{Further Structural Connections}
\label{subsec:further-connections}
Theorem~\ref{thm:forced-uniqueness} sits in a network of structural connections with adjacent literatures.
We record the multi-asset extension here. 
Further connections, to the prediction-market tradition and to minimal-entropy pricing measures, 
are noted in Section~\ref{subsec:future-work}.

\paragraph{Multi-asset market making with a single risk-aversion parameter.}
\label{par:multi-asset}
The forced uniqueness theorem is stated for a single-asset market maker, 
but its natural multi-asset extension  preserves the single-scalar-$\gamma$ structure.
Once one reduces, via the wealth-summary property (W) (Proposition~\ref{prop:wealth-summary}), to a functional of the scalar random variable $\WL_T \in \Linf(\F_T)$, the entire multi-asset structure is absorbed into the joint distribution of $\WL_T$, 
and the Kupper--Schachermayer representation continues to apply unchanged.
Asset-specific structure (correlations, asset-specific liquidation costs, asset-specific volatilities) enters through the joint distribution of $\WL_T$, not through any preference-side asymmetry. 
The formal statement and proof are in Appendix~\ref{app:multi-asset} (Theorem~\ref{thm:multi-asset}).

Operationally, this gives a sharp cross-asset falsifiability criterion. 
A desk that fits, for each of $K$ assets, its implicit $\gamma$ via the inversion 
$\gamma = 2\phi/\sigma^2$ of Corollary~\ref{cor:calibration-inversion} should obtain $K$ approximately equal values. 
Persistent cross-asset disagreement is evidence either of a violation of the multi-asset axiom system or of inconsistent calibration.

\subsection{Directions for Future Work}
 \label{subsec:future-work}

 We close with four directions which our framework suggests but does not pursue.

\paragraph{Empirical testability of the forced coefficient.}
The forced relation $\phi = \gamma\sigma^2/2$ is an empirically falsifiable prediction.
A market-making desk that calibrates $\gamma$ from AS-style quoted-spread data and $\phi$ from CJ-style P\&L or inventory-penalty data can compare the two implied values against $\gamma\sigma^2/2$ at the prevailing realized volatility. 
The calibration-inversion of Corollary~\ref{cor:calibration-inversion} gives the explicit formula $\gamma_{\phi} = 2\phi/\sigma^2$.
A desk for which $\gamma_{\phi}$ and the AS-implied $\gamma_{\text{AS}}$ disagree persistently is 
either over-parametrized or violating at least one of the five axioms.
We identify this as a productive direction for applied work with desk-internal or high-frequency data.

\paragraph{Strategy space: the limit order book.}
We have taken the strategy space to be the AS specification of two real-valued quote-distance processes (augmented with a cash-injection process).
Real limit order books operate under price-time priority, where modifying a quote loses queue priority; 
see \citet{LawViens2019}, \citet{LuAbergel2018}, and \citet{MoallemiYuan2016}.
Our preference-side axioms are independent of the strategy space; a natural extension would combine them with explicit queue-aware strategy-space axioms.
  An interesting question is whether the AS continuous-quote control arises as the small-tick limit of a queue-aware control.

\paragraph{Endogenous flow.}
Our assumption that the counterparty intensities $\lambda^a, \lambda^b$ are exogenous functions of quote distances is a simplification, since sophisticated counterparties (``skew sniffers'') read the market maker's position from her quote skew and adjust their arrivals \citep{BarzykinBergaultGueantLemmel2025}.
A natural extension would investigate whether equilibrium-rational counterparty behavior is consistent with Axioms~\ref{ax:J1}--\ref{ax:J5}; preliminary investigation suggests the framework absorbs endogenous flow as state-dependence of the intensities rather than as a violation of any axiom.
We leave the formal development to future work.

\paragraph{Convergence with the prediction-market tradition.}
The prediction-market tradition, originating with the market-scoring-rule mechanism of \citet{Hanson2003, Hanson2007} and developed by \citet{ChenPennock2007}, \citet{AbernethyChenVaughan2013}, and \citet{AbernethyFrongilloKutty2014},
identifies the entropic family as canonical under different axioms (path-independence, bounded loss, translation-invariance, monotonicity), with the cost function forced to be the conjugate of a generalized entropy \citep{AbernethyFrongilloKutty2014}.
Although the two axiom systems are structurally different, the shared emergence of the entropic family suggests a deeper connection through cash-additivity and a consistency requirement, whose unified treatment we leave to future work.

\paragraph{Connection to minimal-entropy pricing measures.}
The entropic functional is intimately connected to the minimal-entropy martingale measure of \citet{Frittelli2000}, which arises in utility-indifference pricing in incomplete markets \citep{FollmerSchweizer1991}, and identified via the dual representation of the entropic risk measure in the CARA case.
The connection is suggestive rather than exact, since our strategy space (quoting plus cash injection) differs from the self-financing space of indifference pricing; we leave the detailed analysis to future work.


\section{Conclusion}
\label{sec:conclusion}

We have proposed five core axioms for a dynamic preference functional of an inventory market maker (cash-additivity, normalization, concavity with strict concavity at $t=0$,  strong dynamic consistency, and law-invariance), and three additional properties ((M) monotonicity, (W) wealth-summary, and (R) right-continuity) that are derivable from them.
The five axioms force the market maker's preferences to be the entropic certainty-equivalent on liquidation-adjusted terminal wealth, parametrized by a single positive scalar $\gamma$ which is clock-invariant in the sense of Corollary~\ref{cor:clock-invariance}.

The result has several sharp consequences.
  The Avellaneda--Stoikov framework is the unique law-invariant dynamically-consistent  inventory market making model; the Cartea--Jaimungal objective is incompatible as a primitive functional but recoverable as a second-order approximation with a forced running-penalty coefficient  $\phi = \gamma\sigma^2/2$ and forced terminal coefficient $\alpha = \tfrac{1}{2}L''(0)$.
The forced relation $\phi = \gamma  \sigma^2/2$ is invertible, $\gamma = 2\phi/\sigma^2$, providing an audit instrument for desks operating in the CJ tradition.
The forced functional  admits a robust-optimization dual under a Knightian-ambiguity reading, is not positively homogeneous (so linear-in-size risk-capital allocation misprices large positions), and is structurally incompatible with dynamic CVaR.
The market maker's preference structure has one free scalar parameter on the open half-line, interpolating between risk-neutral and worst-case limits; the constant-absolute-risk-aversion property of the AS utility is a corollary of the deeper uniqueness, not a primitive choice.

The framework has a real limitation, namely that drawdown-averse market makers, for whom intraday loss limits induce path-functional preferences, are not covered.
Four directions for future work emerged from the analysis: (1) empirical testability of the forced coefficient, (2) the strategy-space side (limit-order-book microstructure), (3) endogenous counterparty flow, and (4) the connection to the prediction-market tradition.

The message is this. 
Since the publication of \citet{AvellanedaStoikov2008}, the inventory market making literature has carried two parallel traditions with the choice between them governed by tractability rather than principle.
Our result is that there are not two traditions but one,  and that the apparent freedom in choosing between them conceals a forced constraint, 
namely $\phi = \gamma\sigma^2/2$.
A practitioner who calibrates $\phi$ and $\gamma$ as independent parameters is leaving a falsifiable prediction on the table.


\appendix

\section{Derived Properties and Independence of Axioms}
 \label{app:indep}

This appendix establishes two complementary results about the axiom system of Section~\ref{sec:axioms}.
First, we prove the three derivability propositions mentioned in Section~\ref{subsec:derived-properties}: 
monotonicity (M) (Proposition~\ref{prop:monotonicity-app}), the wealth-summary property (W) (Proposition~\ref{prop:wealth-summary-app}), 
and right-continuity in time (R) (Proposition~\ref{prop:right-continuity-app}).
The relevance condition is likewise derivable from the five core axioms, 
as already established by Lemma~\ref{lem:relevance-derivable} in Section~\ref{subsec:relevance}.
Second, we prove that the five core axioms (J1)--(J5) are \emph{logically independent}:
 for each axiom, there exists a preference functional satisfying the other four but not that one (Theorem~\ref{thm:independence}).

\subsection{Derived Properties}
\label{subsec:derived-app}

Throughout this subsection, $J = (J_t)_{t \in [0,T]}$ is a dynamic preference functional on the strategy space $\Pi$ satisfying a subset of Axioms~\ref{ax:J1}--\ref{ax:J5} specified separately for each result.
We make repeated use of the standing assumption (Section~\ref{sec:setup}) that $\Pi$ is closed under bounded $\F_T$-measurable cash injection. 
That is, for any $\pi \in \Pi$ and any bounded $c \in \Linf(\F_T)$, the cash-injected strategy $\pi^{[c]}$ (defined to inject $c$ at time $T$) lies in $\Pi$, with $\WL_T(\pi^{[c]}) = \WL_T(\pi) + c$ almost surely (Lemma~\ref{lem:density} of Appendix~\ref{app:technical}).
We also use the existence of a strategy $\pi_0 \in \Pi$ with $\WL_T(\pi_0) = 0$ almost surely (the constant-zero strategy, also Lemma~\ref{lem:density}).

\begin{proposition}[Monotonicity (M), Appendix Version]
\label{prop:monotonicity-app}
Any map  $J$ satisfying Axioms~\ref{ax:J1} (cash-additivity),  \ref{ax:J2} (normalization), and~\ref{ax:J4} (strong dynamic consistency) also satisfies the monotonicity property (M) of Proposition~\ref{prop:monotonicity}: 
for all $\pi, \pi' \in \Pi$, if $\WL_T(\pi) \geq \WL_T(\pi')$ almost surely, then $J_t(\pi) \geq J_t(\pi')$ almost surely for every $t \in [0,T]$.
\end{proposition}

\begin{proof}
The argument has two steps.

\smallskip
\emph{Step 1: At terminal time $T$, the functional $J_T$ equals liquidation-adjusted terminal wealth.}
Fix $\pi \in \Pi$ and let $W \coloneqq \WL_T(\pi) \in \Linf(\F_T)$.
Consider the strategy $\pi^{[-W]} \in \Pi$ obtained by injecting the bounded $\F_T$-measurable cash amount $-W$ at time $T$.
Its liquidation-adjusted terminal wealth is  $\WL_T(\pi^{[-W]}) = \WL_T(\pi) + (-W) = W - W = 0$ almost surely.
By Axiom~\ref{ax:J2} (normalization) applied to $\pi^{[-W]}$, we have $J_T(\pi^{[-W]}) = 0$ almost surely.
On the other hand, Axiom~\ref{ax:J1} (cash-additivity at time $T$, with the  $ \F_T$-measurable cash injection $c = -W$) gives $J_T(\pi^{[-W]}) = J_T(\pi) + (-W) = J_T(\pi) - W$ almost surely.
Combining the two identities yields $J_T(\pi) = W = \WL_T(\pi)$ almost surely.
Since $\pi \in \Pi$ was arbitrary, we have established
\begin{equation}
    J_T(\pi)  =  \WL_T(\pi) \quad \text{almost surely, for every } \pi \in \Pi.
    \label{eq:JT-identity-app}
\end{equation}

\smallskip
\emph{Step 2: Propagate monotonicity backwards via dynamic consistency.}
Let $\pi, \pi' \in \Pi$ satisfy $\WL_T(\pi) \geq \WL_T(\pi')$ almost surely.
By~\eqref{eq:JT-identity-app}, $J_T(\pi) =  \WL_T(\pi) \geq \WL_T(\pi') = J_T(\pi')$ almost surely.
We then apply Axiom~\ref{ax:J4} (strong dynamic consistency) with the pair $(\pi, \pi')$
  at times $(s, t) = (t, T)$ for any $t \in [0,T]$: since $J_T(\pi) \geq J_T(\pi')$ almost surely,  the axiom yields $J_t(\pi) \geq J_t(\pi')$ almost surely.
This is the monotonicity property.
\end{proof}

\begin{proposition}[Wealth-Summary (W), Appendix Version]
\label{prop:wealth-summary-app}
Any $J$ satisfying monotonicity (M) of Proposition~\ref{prop:monotonicity-app} satisfies 
the wealth-summary property of Proposition~\ref{prop:wealth-summary}:
for all $\pi, \pi' \in \Pi$,  $\WL_T(\pi) = \WL_T(\pi')$ a.s.\ implies $J_t(\pi) = J_t(\pi')$ a.s.\ for every $t \in [0,T]$.
In particular, any $J$ satisfying Axioms~\ref{ax:J1}, \ref{ax:J2}, and~\ref{ax:J4} satisfies the wealth-summary property.
\end{proposition}

\begin{proof}
Suppose $\WL_T(\pi) = \WL_T(\pi')$ almost surely.
Then in particular $\WL_T(\pi) \geq \WL_T(\pi')$ almost surely and
 $\WL_T(\pi') \geq \WL_T(\pi)$ almost surely.
By Proposition~\ref{prop:monotonicity-app} applied in the first direction, $J_t(\pi) \geq J_t(\pi')$  almost surely; 
by Proposition~\ref{prop:monotonicity-app} applied in the second direction, $J_t(\pi') \geq J_t(\pi)$ almost surely.
The two inequalities together yield $J_t(\pi) = J_t(\pi')$ almost surely.
\end{proof}

\begin{proposition}[Right-Continuity in Time (R), Appendix Version]
\label{prop:right-continuity-app}
Any $J$ satisfying Axioms~\ref{ax:J1}--\ref{ax:J5} also satisfies the right-continuity property (R) of
  Proposition~\ref{prop:right-continuity}: 
  for every $\pi \in \Pi$ and every $t \in [0, T)$, the path $s \mapsto J_s(\pi)$ is right-continuous in probability at $s = t$.
\end{proposition}

\begin{proof}
The proof proceeds in  two stages.
Stage 1 establishes the entropic representation, 
$J_t(\pi) = -\tfrac{1}{\gamma} \log  \E[e^{-\gamma\WL_T(\pi)}  \mid \F_t]$ at every $t \in [0,T]$, 
using a Bellman + cash-additivity argument that, importantly, does not invoke right-continuity at any point.
Stage 2 then derives right-continuity in probability as a direct consequence of the right-continuity 
of conditional expectations under the usual conditions on the filtration.

\smallskip
By Proposition~\ref{prop:wealth-summary-app}, $J_t(\pi)$ depends on $\pi$ only through 
$W \coloneqq \WL_T(\pi)$, so we may work with the reduced functional $\tilde J_t \colon \Linf(\F_T) \to \Linf(\F_t)$ defined by $\tilde J_t(W) \coloneqq J_t(\pi)$ for any $\pi \in \Pi$ with $\WL_T(\pi) = W$.
By Lemma~\ref{lem:density}, every $W \in \Linf(\F_T)$ arises as some $\WL_T(\pi)$, so $\tilde J_t$ is  defined on all of $\Linf(\F_T)$.

\smallskip
\emph{Stage 1, Step 1.1: Bellman identity at every time.}
We claim that for every $t_0 \in [0,T]$ and every $W \in \Linf(\F_T)$,
\begin{equation}
    \tilde J_s\bigl(\tilde J_{t_0}(W)\bigr)  =  \tilde J_s(W) \quad \text{almost surely, for all } s \leq t_0.
    \label{eq:bellman-app}
\end{equation}
Fix $t_0$ and $W$.
 Let $\pi, \pi_0 \in \Pi$ satisfy $\WL_T(\pi) = W$ and $\WL_T(\pi_0) = 0$, respectively. 
We first verify that $\tilde J_{t_0}(W) \in \Linf(\F_{t_0})$,  with $\|\tilde J_{t_0}(W)\|_\infty \leq \|W\|_\infty$, so that the cash-injected strategy below is admissible.
Write $M_W \coloneqq \|W\|_\infty < \infty$ and let $\pi_0^{[\pm M_W]}$ denote the cash-injected strategies of Section~\ref{subsec:strategy} applied to $\pi_0$, which satisfy $\WL_T(\pi_0^{[\pm M_W]}) = \pm M_W$ almost surely.
By Axioms~\ref{ax:J1} and~\ref{ax:J2}, it follows that 
$J_{t_0}(\pi_0^{[\pm M_W]}) = \pm M_W$ almost surely.
Since we have  $\WL_T(\pi_0^{[-M_W]}) \leq \WL_T(\pi) \leq \WL_T(\pi_0^{[M_W]})$ almost surely, Proposition~\ref{prop:monotonicity-app} (monotonicity (M), derived earlier in this appendix from Axioms~\ref{ax:J1}, \ref{ax:J2}, \ref{ax:J4}) gives $-M_W \leq J_{t_0}(\pi) \leq M_W$ almost surely, so $\tilde J_{t_0}(W) = J_{t_0}(\pi) \in \Linf(\F_{t_0})$ with the claimed bound.
We define $\pi^* \coloneqq \pi_0^{[\tilde J_{t_0}(W)]}$, the strategy obtained by injecting at time $t_0$ the $\F_{t_0}$-measurable cash amount 
given by  $\tilde J_{t_0}(W)$.
Then $\WL_T(\pi^*) = 0 + \tilde J_{t_0}(W) = \tilde J_{t_0}(W)$ almost surely.
By Axiom~\ref{ax:J2} (normalization), $J_{t_0}(\pi_0) = 0$; by Axiom~\ref{ax:J1} (cash-additivity at time $t_0$),
\[
    J_{t_0}(\pi^*)  =  J_{t_0}(\pi_0) + \tilde J_{t_0}(W)  =  0 + \tilde J_{t_0}(W)  =  \tilde J_{t_0}(W)  =  J_{t_0}(\pi) \quad \text{a.s.}
\]
We apply Axiom~\ref{ax:J4} (strong dynamic consistency) in both directions to the pair $(\pi, \pi^*)$: 
since $J_{t_0}(\pi) = J_{t_0}(\pi^*)$ almost surely, we obtain $J_s(\pi) = J_s(\pi^*)$ almost surely for all $s \leq t_0$.
Translating via $\tilde J$,  
we get that $\tilde J_s(W) = \tilde J_s\bigl(\WL_T(\pi^*)\bigr) =  \tilde J_s\bigl(\tilde J_{t_0}(W)\bigr)$ a.s., which is~\eqref{eq:bellman-app}.

\smallskip
\emph{Stage 1, Step 1.2: Entropic form at dyadic times.}
Let  $\mathcal{D} \coloneqq \bigcup_{n \geq 1} \{kT/2^n \mid 0 \leq k \leq 2^n\}$ be the set of dyadic rationals in $[0,T]$.
By the Kupper--Schachermayer representation argument (the discrete-time content of  Sub-steps 3a--3b in the proof of Theorem~\ref{thm:forced-uniqueness}, which uses Axioms~\ref{ax:J1}--\ref{ax:J5} and relevance but does not use right-continuity),  there exists a unique $\gamma \in (0,\infty)$ such that for every dyadic $s \in \mathcal{D}$ and every $W \in \Linf(\F_T)$,
\begin{equation}
    \tilde J_s(W)  =  -\frac{1}{\gamma}\log \E[e^{-\gamma W} \mid \F_s] \quad \text{almost surely.}
    \label{eq:dyadic-app}
\end{equation}
At dyadic times we therefore have the desired entropic form.

\smallskip
\emph{Stage 1, Step 1.3: A conditional MGF identity below $t_0$.}
Fix an arbitrary $t_0 \in [0,T]$ (not necessarily dyadic) and $W \in \Linf(\F_T)$.
For any dyadic $s \in \mathcal{D}$ with $s < t_0$, combining~\eqref{eq:bellman-app}  and~\eqref{eq:dyadic-app} gives
\[
    -\frac{1}{\gamma}\log\E[e^{-\gamma W} \mid \F_s]  =  \tilde J_s(W)  =   \tilde J_s\bigl(\tilde J_{t_0}(W)\bigr)  =  -\frac{1}{\gamma}\log\E[e^{-\gamma \tilde J_{t_0}(W)} \mid \F_s].
\]
Since both sides are bounded random variables and the logarithm  is injective on the strictly positive bounded range,  we obtain
\begin{equation}
    \E[e^{-\gamma W} \mid \F_s]  =  \E[e^{-\gamma \tilde J_{t_0}(W)} \mid \F_s] \quad \text{almost surely, for every dyadic } s < t_0.
    \label{eq:mgf-dyadic-app}
\end{equation}

\smallskip
\emph{Stage 1, Step 1.4: Pass to the left-limit $s \uparrow t_0$.}
Let $(s_n)_{n \geq 1} \subset \mathcal{D}$ be a sequence of dyadic times with  $s_n < t_0$ and $s_n \uparrow t_0$, which exists since $\mathcal{D}$ is dense in $[0,T]$.
The $\sigma$-algebras $\F_{s_n}$ increase to $\F_{t_0^-} \coloneqq \sigma\bigl(\bigcup_{n} \F_{s_n}\bigr) = \sigma\bigl(\bigcup_{s < t_0} \F_s\bigr)$.
By the L\'evy upward martingale convergence theorem applied to the bounded random variables $e^{-\gamma W}$ 
and $e^{-\gamma\tilde J_{t_0}(W)}$, both sides of~\eqref{eq:mgf-dyadic-app} converge almost surely as $n \to \infty$, yielding
\begin{equation}
    \E[e^{-\gamma W} \mid \F_{t_0^-}]  =  \E[ e^{-\gamma \tilde J_{t_0}(W)} \mid \F_{t_0^-}] \quad \text{almost surely.}
    \label{eq:starstar-app}
\end{equation}

\smallskip
\emph{Stage 1, Step 1.5: Use cash-additivity to  interrogate  with $\F_{t_0}$-measurable test functions.}
We have established that~\eqref{eq:starstar-app} holds for every $W \in \Linf(\F_T)$.
Fix any bounded $\F_{t_0}$-measurable random variable $c \in \Linf(\F_{t_0})$ and apply~\eqref{eq:starstar-app} to $W + c$ in place of $W$ (which is valid since $W + c \in \Linf(\F_T)$).
By Axiom~\ref{ax:J1} (cash-additivity at time $t_0$), $\tilde J_{t_0}(W + c) = \tilde J_{t_0}(W) + c$ almost surely, 
  hence $e^{-\gamma\tilde J_{t_0}(W+c)} = e^{-\gamma c} \cdot e^{-\gamma\tilde J_{t_0}(W)}$ almost surely.
Substituting, we obtain that 
\begin{equation*}
    \E\!\bigl[e^{-\gamma c}\, e^{-\gamma\tilde J_{t_0}(W)} \,\big|\, \F_{t_0^-}\bigr]
     = 
    \E\!\bigl[e^{-\gamma c}\, e^{-\gamma W} \,\big|\, \F_{t_0^-}\bigr]
    \quad \text{almost surely.}
\end{equation*}
We then apply the tower property to the right-hand side,  using $\F_{t_0^-} \subseteq \F_{t_0}$ and the $\F_{t_0}$-measurability of $e^{-\gamma c}$, 
to get 
\[
    \E\!\bigl[e^{-\gamma c}\, e^{-\gamma W} \,\big|\, \F_{t_0^-}\bigr]
       = 
    \E\!\bigl[e^{-\gamma c}\, \E[e^{-\gamma W} \mid \F_{t_0}]  \,\big|\, \F_{t_0^-}\bigr].
\]
 Define
\[
    \Delta  \coloneqq  e^{-\gamma\tilde J_{t_0}(W)}  -  \E[e^{-\gamma W} \mid \F_{t_0}].
\]
Both terms are $\F_{t_0}$-measurable and bounded 
(the first because $\tilde J_{t_0}(W) \in \Linf(\F_{t_0})$,  
   the second by definition of conditional expectation), so $\Delta \in \Linf(\F_{t_0})$.
Subtracting, we obtain
\begin{equation}
    \E\!\bigl[e^{-\gamma c}\, \Delta \,\big|\, \F_{t_0^-}\bigr]  =  0 \quad \text{almost surely, for every } c \in \Linf(\F_{t_0}).
    \label{eq:dagger-app}
\end{equation}

\smallskip
\emph{Stage 1, Step 1.6: Conditional MGF argument forces $\Delta = 0$.}
Since $\Delta \in \Linf(\F_{t_0})$, for any $k \in \R$ the random variable $c_k  \coloneqq k\Delta/\gamma$ belongs to $\Linf(\F_{t_0})$. 
Upon substituting $c = c_k$ in~\eqref{eq:dagger-app}, we obtain that 
\begin{equation}
    \E\!\bigl[e^{-k\Delta}\, \Delta \,\big|\, \F_{t_0^-}\bigr]  =  0 \quad \text{almost surely, for every } k \in \R.
    \label{eq:k-test-app}
\end{equation}
Note that,  for each $k \in \R$,  
identity~\eqref{eq:k-test-app} holds outside a $\Prob$-null set that may  a priori depend on $k$.
To extract a statement uniform in $k$, we first restrict to a countable subset and then extend by continuity.
Let $\Q \subset \R$ denote the rationals.
There exists a single $\Prob$-null set $N \subset \Omega$ such that for every $\omega \notin N$ and every $k \in \Q$, $\E[e^{-k\Delta}\Delta \mid \F_{t_0^-}](\omega) = 0$.
 For each $k \in \R$ and $\omega \in \Omega$, define
$ \phi(k, \omega) \coloneqq \E[e^{-k\Delta} \mid \F_{t_0^-}](\omega). $
For each fixed $\omega$, the regular conditional distribution of $\Delta$ given $\F_{t_0^-}$ is a 
Borel probability measure on the bounded interval $[-\|\Delta\|_\infty, \|\Delta\|_\infty]$, and $\phi(\cdot, \omega)$ is
 the moment-generating function of that measure.
The moment-generating function of a compactly supported probability measure is entire on 
$\mathbb{C}$ 
(and in particular real-analytic on $\R$); this is a standard consequence of differentiation under the integral sign,
   justified by the dominating bound $|\Delta^n e^{-k\Delta}| \leq \|\Delta\|_\infty^n\, e^{|k|\,\|\Delta\|_\infty}$ valid for every $n \geq 0$, which gives an absolutely convergent Taylor series for $\phi(\cdot, \omega)$ on every bounded set in $\R$ \citep[Theorem 5.30]{Kallenberg2021}.
 Its derivative
\[
\partial_k \phi(k, \omega)  =  -\,\E\bigl[\Delta\, e^{-k\Delta}  \,\big|\, \F_{t_0^-}\bigr](\omega)
\]
is therefore continuous in $k$ for each fixed $\omega \notin N$, and vanishes on the dense subset $\Q$; by continuity, $\partial_k \phi(k, \omega) = 0$ for every $k \in \R$ and every $\omega \notin N$.
Therefore,  
$k \mapsto \phi(k, \omega)$ is constant on $\R$ for every $\omega \notin N$.
Evaluating at $k = 0$ gives $\phi(0, \omega) = 1$, so $\phi(k, \omega) = 1$ for every $k \in \R$ and every $\omega \notin N$.
Therefore, almost surely on $\Omega$, the conditional moment-generating function of $\Delta$ given $\F_{t_0^-}$ equals $1$ for every $k \in \R$.
Since $\Delta$ is bounded, its conditional law given $\F_{t_0^-}$ (a regular conditional distribution exists since $\Delta \in \Linf$ is real-valued, by \citet[Theorem 8.5]{Kallenberg2021}) is uniquely determined by this conditional MGF on any open interval containing $0$, in particular on $(-\|\Delta\|_\infty^{-1}, \|\Delta\|_\infty^{-1})$.
The unique Borel probability measure on $\R$ with MGF identically equal to $1$ on such an interval is the Dirac mass at $0$.
Hence, the conditional law of $\Delta$ given $\F_{t_0^-}$ is almost surely the Dirac mass at $0$, which implies $\E[\Delta^2 \mid \F_{t_0^-}] = 0$ almost surely, and consequently $\Delta = 0$ almost surely.
We have therefore established that 
\[
    e^{-\gamma\tilde J_{t_0}(W)}  =  \E[e^{-\gamma W} \mid \F_{t_0}] \quad \text{almost surely.}
\]
Taking $-\gamma^{-1}\log$ of both sides (valid since both sides are strictly positive and bounded away from zero by $e^{-\gamma\|W\|_\infty} > 0$), 
we obtain that 
\begin{equation}
    \tilde J_{t_0}(W)  =  -\frac{1}{\gamma}\log\E[e^{-\gamma W} \mid \F_{t_0}] \quad \text{almost surely.}
    \label{eq:entropic-all-t-app}
\end{equation}
This holds for every $t_0 \in [0,T]$ and every $W \in \Linf(\F_T)$.

\smallskip
\emph{Stage 2: Right-continuity in probability follows.}
Fix $\pi \in \Pi$, write $W \coloneqq \WL_T(\pi)$,  and let $t \in [0,T)$ with $t_k \downarrow t$, 
where $t_k \in [t, T]$.
By~\eqref{eq:entropic-all-t-app},
\[
    J_{t_k}(\pi)  =  \tilde J_{t_k}(W)   =   - \frac{1}{\gamma}\log\E[e^{-\gamma W} \mid \F_{t_k}].
\]
The process $M_u \coloneqq \E[e^{-\gamma W} \mid \F_u]$ is a uniformly integrable martingale; 
under the usual conditions on $\filt$ (right-continuity and completeness, as in Section~\ref{sec:setup}), 
  $M$ is almost surely right-continuous, and in particular $M_{t_k} \to M_t$ almost surely as $t_k \downarrow t$.
Since $M_u \geq e^{-\gamma\|W\|_\infty} > 0$ uniformly in $u$, the function $x \mapsto -\gamma^{-1}\log x$ is 
  continuous and bounded on the range of $M$, hence
\[
    J_{t_k}(\pi)  =  -\frac{1}{\gamma}\log M_{t_k} \longrightarrow -\frac{1}{\gamma}\log M_t  =  J_t(\pi) \quad \text{almost surely as } t_k \downarrow t.
\]
Almost-sure convergence implies convergence in probability.
This is the right-continuity property.
\end{proof}

\subsection{Independence of the Five Core Axioms}
\label{subsec:independence-app}

\begin{theorem}[Independence of (J1)--(J5)]
\label{thm:independence}
The five axioms (J1) (cash-additivity), (J2) (normalization), (J3) (concavity), (J4) (strong dynamic consistency),
and (J5) (law-invariance) form a logically independent system:
 for each $k \in \{1,2,3,4,5\}$,
 there exists a dynamic preference functional $J^{(k)}$ on  $\Linf(\F_T)$  satisfying every axiom in the list except~(J$k$).
\end{theorem}

\begin{proof}
In this  proof we work directly with the reduced functional 
 $\tilde J_t \colon \Linf(\F_T) \to \Linf(\F_t)$.
Each separating model below is defined as a functional of $W = \WL_T(\pi)$ directly, so the wealth-summary property (W) holds by construction in each case, without requiring the general derivation of Proposition~\ref{prop:wealth-summary-app} (which assumes J1, J2, and J4 simultaneously).
Fix any reference parameter $\gamma > 0$ and denote by $J^{\gamma}$ the entropic certainty-equivalent on $\Linf(\F_T)$, namely
\[
  J^{\gamma}_t(W)  \coloneqq   -\frac{1}{\gamma}\log\E[e^{-\gamma W}  \mid \F_t], \qquad W \in \Linf(\F_T),\; t \in [0,T].
\]
This functional satisfies all of~(J1)--(J5) (by the converse direction of Theorem~\ref{thm:forced-uniqueness});
  each separating model below is obtained as a small perturbation of $J^{\gamma}$ designed to violate exactly one axiom.
We treat the five cases in turn.

\smallskip
\noindent\emph{Model $M_1$: violates (J1), satisfies (J2)--(J5).}
Define $J^{(1)}_t(W) \coloneqq 2 J^{\gamma}_t(W)$.

\emph{(J1) fails.} 
For any nonzero $\F_t$-measurable bounded $c$,
\begin{equation*}
    J^{(1)}_t(W+c)  =  2 J^{\gamma}_t(W+c)  =  2 J^{\gamma}_t(W) + 2c \neq 2 J^{\gamma}_t(W) + c  =  J^{(1)}_t(W) + c.
\end{equation*}

\emph{(J2) holds.} 
$J^{(1)}_t(0) = 2 \cdot 0 = 0$.

\emph{(J3) holds.}
We verify J3a and J3b directly at the functional level, without invoking the strategy-level reduction
   (the strategy-level reduction would require J1 for $J^{(1)}$, which fails).
For J3a, fix $W, W' \in \Linf(\F_T)$ and any $\F_t$-measurable $\lambda \in [0,1]$.
By the (weak) concavity of $J^\gamma$ (verified independently in the converse direction of Theorem~\ref{thm:forced-uniqueness},  which uses only the explicit entropic form on  $\Linf$ and does not depend on the strategy-level J1),
we have 
\begin{align*}
    J^{(1)}_t(\lambda W + (1-\lambda)W')
    & =  2 J^\gamma_t(\lambda W + (1-\lambda)W') \\
    & \geq  2\bigl[\lambda J^\gamma_t(W) + (1-\lambda) J^\gamma_t(W')\bigr] \\
    & =  \lambda J^{(1)}_t(W) + (1-\lambda) J^{(1)}_t(W').
\end{align*}
For J3b, take $W, W'$ with $W - W'$ not a.s.\ constant and deterministic $\lambda \in (0,1)$.
The strict concavity of $J^\gamma$ at $t = 0$ on such pairs (verified independently  of J1, since the explicit entropic functional $J^\gamma$ has J3b for $\gamma > 0$) gives $J^\gamma_0(\lambda W + (1-\lambda)W') > \lambda J^\gamma_0(W) + (1-\lambda) J^\gamma_0(W')$.
Multiplying by $2 > 0$ preserves the strict inequality, hence $ J^{(1)}_0$ satisfies J3b.

\emph{(J4) holds.}
Since $J^{(1)}_t(\pi) = 2 J^\gamma_t(\pi)$ at the strategy level (by definition of $J^{(1)}$ as twice the entropic functional), the equivalence $J^{(1)}_t(\pi) \geq J^{(1)}_t(\pi')$ a.s.\ if{f} $J^{\gamma}_t(\pi) \geq J^{\gamma}_t(\pi')$ a.s.\ holds at the strategy level as well as at the reduced-functional level.
Strong dynamic consistency of $J^\gamma$ in the strategy-level form (Axiom~\ref{ax:J4}) therefore transfers directly to $J^{(1)}$.

\emph{(J5) holds.}  
Law-invariance is immediate: 
$W \eqdist W'$ implies $\E[e^{-\gamma W}] = \E[e^{-\gamma W'}]$, so $J^{(1)}_0(W) = 2J^{\gamma}_0(W) = 2J^{\gamma}_0(W') = J^{(1)}_0(W')$.

\smallskip
\noindent\emph{Model $M_2$: violates (J2), satisfies (J1) and (J3)--(J5).}
Define $J^{(2)}_t(W) \coloneqq J^{\gamma}_t(W) + 1$.

\emph{(J2) fails.} 
$J^{(2)}_t(0) = 0 + 1 = 1 \neq 0$.

\emph{(J1) holds.} 
$J^{(2)}_t(W + c) = J^{\gamma}_t(W + c) + 1 = J^{\gamma}_t(W) + c + 1 = J^{(2)}_t(W) + c$ for any $\F_t$-measurable bounded $c$.

\emph{(J3) holds.} 
Adding a constant preserves (strict) concavity, in both clauses.

\emph{(J4) holds.}
  Adding the same constant to two functionals does not change orderings: 
  $J^{(2)}_t(W) \geq J^{(2)}_t(W')$ if{f} $J^{\gamma}_t(W) \geq J^{\gamma}_t(W')$, and dynamic consistency of $J^{\gamma}$ transfers.

\emph{(J5) holds.}
  If $W \eqdist W'$ then $J^{(2)}_0(W) = J^{\gamma}_0(W) + 1 = J^{\gamma}_0(W') + 1 = J^{(2)}_0(W')$.

\smallskip
\noindent\emph{Model $M_3$: violates (J3), satisfies (J1), (J2), (J4), (J5).}
Define $J^{(3)}_t(W) \coloneqq \E[W \mid \F_t]$, the conditional expectation.

\emph{(J3) fails.} 
For any $W, W' \in \Linf(\F_T)$ and any $\F_t$-measurable $\lambda \in [0,1]$,
\begin{equation*}
    J^{(3)}_t(\lambda W + (1-\lambda) W')  =  \lambda\,\E[W \mid \F_t] + (1-\lambda)\,\E[W' \mid \F_t]  =  \lambda\, J^{(3)}_t(W) + (1-\lambda)\, J^{(3)}_t(W').
\end{equation*}
Equality holds identically, so the weak concavity inequality (J3a) is satisfied with equality everywhere, and the \emph{strict}-concavity clause (J3b) fails for every pair $(W, W')$ with $W - W'$ nonconstant.
Hence (J3), taken to require both clauses, is violated.

\emph{(J1) holds.} 
For $c \in  \Linf(\F_t)$, we have  $\E[W + c  \mid \F_t] = \E[W \mid \F_t] + c$.

\emph{(J2) holds.} 
$\E[0 \mid \F_t]  = 0$.

\emph{(J4) holds.}
  If $\E[W \mid \F_t] \geq \E[W' \mid \F_t]$ almost surely and $s \leq t$, then by the tower property,
\begin{equation*}
    \E[W \mid \F_s]  =  \E[\E[W \mid \F_t] \mid \F_s]  \geq  \E[\E[W' \mid \F_t]  \mid \F_s]  =  \E[W' \mid \F_s] \quad \text{almost surely.}
\end{equation*}

\emph{(J5) holds.} 
 If $W \eqdist W'$ then $\E[W] = \E[W']$;  at $t = 0$ with $\F_0$ trivial, 
it holds that $J^{(3)}_0(W) = \E[W] = \E[W'] = J^{(3)}_0(W')$.

\smallskip
\noindent\emph{Model $M_4$: violates (J4), satisfies (J1), (J2), (J3), (J5).}
 We construct a time-dependent entropic functional with distinct risk-aversion parameters on a partition of $[0,T]$.
Choose any two distinct parameters $0 <  \gamma_1  \neq \gamma_2  < \infty$ and define
$$
    J^{(4)}_t(W)  \coloneqq 
    \begin{cases}
        J^{\gamma_1}_t(W) & \text{if } t \in [0, T/2), \\[4pt]
        J^{\gamma_2}_t(W) & \text{if } t \in [T/2, T].
    \end{cases}
$$
That is, the functional uses risk-aversion $\gamma_1$ on the first half of the horizon and risk-aversion $\gamma_2$ on the second half.

\emph{(J4) fails.}
We exhibit $W, W' \in \Linf(\F_T)$ and times  $s = 0$, $t = T/2$ for which it holds that 
 $J^{(4)}_t(W) \geq J^{(4)}_t(W')$ almost surely yet $J^{(4)}_s(W) < J^{(4)}_s(W')$.
The construction uses that the entropic certainty-equivalent $J^{\gamma}_0$, viewed as a functional on $\Linf$, induces different orderings for different $\gamma$.
Without loss of generality, $\gamma_1 < \gamma_2$.
Take $W' \in \Linf(\F_T)$ a nondegenerate symmetric bounded random variable, independent of $\F_{T/2}$
 (e.g., $W'$ takes the values $+a$ and $-a$ each with probability $1/2$ for some $a > 0$ large enough; 
  existence of such $W'$ follows from the non-atomicity of $(\Omega, \F_T, \Prob)$, 
  in turn implied by $\langle S\rangle_T > 0$ almost surely).
Let $W \coloneqq w$ be a deterministic constant in the open interval $(J^{\gamma_2}_0(W'), J^{\gamma_1}_0(W'))$.
This interval is nonempty: 
by the strict monotonicity of $\gamma \mapsto J^{\gamma}_0(W')$  on  $(0, \infty)$ for any nondegenerate $W'$ (established in Sub-step 3b of the proof of Theorem~\ref{thm:forced-uniqueness}, equation $f'(\gamma) < 0$), we have $J^{\gamma_2}_0(W') < J^{\gamma_1}_0(W')$.
Since $W'$ is independent of $\F_{T/2}$,
\[
    J^{(4)}_{T/2}(W')  =  J^{\gamma_2}_{T/2}(W')  =  -  \frac{1}{\gamma_2}\log\E[e^{-\gamma_2 W'} \mid \F_{T/2}]  =  -\frac{1}{\gamma_2}\log\E[e^{-\gamma_2 W'}]  =  J^{\gamma_2}_0(W') \quad \text{a.s.,}
\]
the constant value $J^{\gamma_2}_0(W')$.
Meanwhile, $J^{(4)}_{T/2}(W) = w$ (since $W$ is deterministic).
By our choice of $w$, it holds that 
$w > J^{\gamma_2}_0(W') = J^{(4)}_{T/2}(W')$ almost surely, so $J^{(4)}_{T/2}(W) \geq J^{(4)}_{T/2}(W')$ almost surely (the hypothesis of (J4)).
At $s = 0$, however, $J^{(4)}_0(W)  = w < J^{\gamma_1}_0(W') = J^{(4)}_0(W')$, contradicting the conclusion of (J4).

\emph{(J1) holds.} 
At each $t$, $J^{(4)}_t$ is an entropic functional with parameter $ \gamma_1$ (if $t < T/2$) or $\gamma_2$ (if $t \geq T/2$), and entropic functionals satisfy cash-additivity.

\emph{(J2) holds.} 
$J^{\gamma_i}_t(0) = 0$ for each $i \in \{1,2\}$.

\emph{(J3) holds.} 
Each $J^{\gamma_i}_t$ is concave (and strictly so at $t = 0$ on nonconstant pairs).
Since the strict-concavity clause (J3b) is at $t = 0$, where $J^{(4)}$ uses $\gamma_1$, and $\gamma_1 > 0$ gives strict concavity, (J3b) holds.

\emph{(J5) holds.}
 $J^{(4)}_0 = J^{\gamma_1}_0$, which is law-invariant.

\smallskip
\noindent\emph{Model $M_5$: violates (J5), satisfies (J1)--(J4).}
Let $\Q$ be a probability measure on $(\Omega, \F)$ equivalent to $\Prob$ with Radon--Nikod\'ym density $\varphi \coloneqq d\Q/d\Prob \in \Linf(\F_T)$, $\varphi > 0$ almost surely, $\E[\varphi] = 1$, and $\varphi$ not $\Prob$-almost-surely constant.
Define
\[
    J^{(5)}_t(W)  \coloneqq  -\frac{1}{\gamma}\log\E^{\Q}[e^{-\gamma W}  \mid \F_t],
\]
the entropic certainty-equivalent computed under $\Q$.

\emph{(J5) fails.}
We construct $W, W' \in \Linf(\F_T)$ with $W \eqdist W'$ under $\Prob$ but $J^{(5)}_0(W) \neq J^{(5)}_0(W')$.
Take $Z \in \Linf(\F_T)$ with $\E[Z] = 0$ and $\E[Z^2] > 0$, and write $\varphi = c_0 + \varepsilon Z$ for $c_0 \coloneqq \E[\varphi] = 1$ and $\varepsilon > 0$ small enough that $\varphi > 0$.
By the non-atomicity of $(\Omega, \F_T, \Prob)$, there exists a random variable $W' \in \Linf(\F_T)$ with $W' \stackrel{d}{=} Z$ under $\Prob$ but independent of $\varphi$
 (constructed on the same probability space using an auxiliary independent uniform random variable).
Set $W \coloneqq Z$.
Then $W \eqdist W'$ under $\Prob$, but
\begin{align*}
    \E^{\Q}[e^{-\gamma W}] &= \E[\varphi\, e^{-\gamma Z}]  =  c_0\, \E[e^{-\gamma Z}] + \varepsilon\,\E[Z\, e^{-\gamma Z}], \\
    \E^{\Q}[e^{-\gamma W'}] &= \E[\varphi]\,\E[e^{-\gamma W'}]  =  \E[e^{-\gamma Z}]  =  c_0\,\E[e^{-\gamma Z}],
\end{align*}
where the second line uses the independence of $W'$  and $\varphi$ together with the fact that  $\E[\varphi] = c_0 = 1$.
The difference is $\varepsilon\,\E[Z\, e^{-\gamma Z}]$, which is nonzero whenever $Z$ is not deterministic.
To see this, let $h(k) \coloneqq \E[e^{-\gamma k Z}]$ for $k \in \R$.
The function $h$ is strictly convex on $\R$ for nondegenerate bounded $Z$ (by strict convexity of $k \mapsto e^{-\gamma k z}$ for each $z$ and the nondegeneracy of the distribution of $Z$).
Its derivative is $h'(k) = -\gamma\,\E[Z\,e^{-\gamma k Z}]$, so the claim $\E[Z\,e^{-\gamma Z}] \neq 0$ is equivalent to $h'(1) \neq 0$.
Strict convexity of $h$ implies $h'$ is strictly increasing, hence has at most one zero on $\R$.
By $\E[Z] = 0$, the derivative at $k = 0$ is $h'(0) = -\gamma\,\E[Z] = 0$, 
and therefore  $k = 0$ is the unique zero of $h'$.
Hence, $h'(1) \neq 0$, i.e., $\E[Z\,e^{-\gamma Z}] \neq 0$.
We get that  $J^{(5)}_0(W) \neq J^{(5)}_0(W')$,  violating (J5).

\emph{(J1) holds.} 
Let $c \in \Linf(\F_t)$. 
Then we have $\E^{\Q}[e^{-\gamma(W+c)} \mid \F_t] = e^{-\gamma c}\,   \E^{\Q}[e^{-\gamma W} \mid \F_t]$
  (since $c$ is $\F_t$-measurable, it factors out under $\Q$-conditional expectation as well). 
We conclude that 
   $J^{(5)}_t(W + c) = J^{(5)}_t(W) + c$.

\emph{(J2) holds.}
 $J^{(5)}_t(0) = -\gamma^{-1}\log\E^{\Q}[1 \mid \F_t] = - \gamma^{-1}\log 1 = 0$.

\emph{(J3) holds.} 
The map $W \mapsto -\gamma^{-1}\log\E^{\Q}[e^{-\gamma W} \mid \F_t]$ is strictly concave in $W$ 
by the strict convexity of $x \mapsto e^{-\gamma x}$ and Jensen's inequality applied under $\Q$; 
the same argument that gives concavity of the entropic functional under $\Prob$ applies under $\Q$.

\emph{(J4) holds.} 
By the tower property under $\Q$, for $s \leq t$,
\begin{equation*}
    \E^{\Q}\bigl[e^{-\gamma J^{(5)}_t(W)}  \,\big|\, \F_s\bigr]  =  \E^{\Q}\bigl[\E^{\Q}[e^{-\gamma W} \mid \F_t] \,\big|\, \F_s\bigr]  =  \E^{\Q}[e^{-\gamma W} \mid \F_s],
\end{equation*}
which is the Bellman identity for $J^{(5)}$.  
The Bellman identity, combined with~(J1) and~(J2), implies strong dynamic consistency.
Indeed,  suppose $J^{(5)}_t(W) \geq J^{(5)}_t(W')$ a.s. 
Then it follows that 
 $e^{-\gamma J^{(5)}_t(W)} \leq e^{-\gamma J^{(5)}_t(W')}$ a.s., and taking $\Q$-conditional expectation given $\F_s$ preserves the inequality.
The tower identity then gives $\E^{\Q}[e^{-\gamma W} \mid \F_s] \leq \E^{\Q}[e^{-\gamma W'} \mid \F_s]$ a.s., and applying $-\gamma^{-1}\log$ yields $J^{(5)}_s(W) \geq J^{(5)}_s(W')$ almost surely. 

\smallskip
This exhibits, for each $k \in \{1,2,3,4,5\}$, a functional satisfying every axiom except~(J$k$), establishing logical independence.
\end{proof}

\begin{remark}[On the Derived Properties]
The three derived properties ((M) monotonicity, (W) wealth-summary, and (R) right-continuity in time) are, by definition, consequences of (J1)--(J5).
Each of the five separating models above satisfies certain of the derived properties despite violating its respective core axiom.
The conditional-expectation model $M_3$ (violating (J3)) satisfies monotonicity (M) by direct calculation, satisfies wealth-summary (W) since the expectation depends only on $W$, and in fact satisfies right-continuity (R) as well, since the process $t \mapsto \E[W \mid \F_t]$ is a uniformly integrable martingale and therefore admits an almost surely right-continuous modification under the usual conditions on $\filt$.
It is important to note the logical distinction here.
The \emph{proof} of (R) in Proposition~\ref{prop:right-continuity-app} uses all of (J1)--(J5), since it first establishes the entropic representation and then derives right-continuity from martingale regularity.
That does not mean every functional violating one of the axioms must fail (R); whether a specific nonentropic functional happens to be right-continuous depends on the functional,   as $M_3$ illustrates.
A precise analysis of which derived properties  survive each individual axiom violation is straightforward but tangential, and we omit it.
\end{remark}


\section{Technical Conditions on the Strategy Space}
\label{app:technical}

In this appendix we record the technical conditions on the strategy space $\Pi$ which justify the reduction step in the proof of Theorem~\ref{thm:forced-uniqueness}.
The conditions are essentially standard for the Avellaneda--Stoikov framework and its extensions, but we have not seen them collected in one place in the form we need, so we make them explicit here.

A reader content with informal statements may safely skip this appendix.
Its key conclusion, which we shall use in Step 1 of the proof of Theorem~\ref{thm:forced-uniqueness}, is the following.

\begin{lemma}[Well-Definedness of the Reduced Functional]
\label{lem:well-defined}
Under the admissibility conditions of Definition~\ref{def:admissible} below, for every $\pi \in \Pi$ the liquidation-adjusted terminal wealth $\WL_T(\pi)$ is in $\Linf(\F_T)$, and there is a functional $\tilde{J}_t \colon \Linf(\F_T) \to \Linf(\F_t)$ such that $J_t(\pi) = \tilde{J}_t(\WL_T(\pi))$ for every $\pi \in \Pi$.
The functional $\tilde{J}_t$ satisfies the dynamic risk measure axioms checked in Section~\ref{subsec:proof}, Step 2, if and only if $J$ satisfies Axioms~\ref{ax:J1}--\ref{ax:J5}.
\end{lemma}

The proof occupies the rest of this appendix and proceeds in three short steps.

\subsection{Admissibility}
\label{app:admissibility}

We work with the setup of Section~\ref{sec:setup}, which we briefly recall.
The mid-price is $dS_t = \sigma_t \, dB_t$ with $\sigma$ predictable and locally bounded; we set $d\langle S\rangle_t = \sigma_t^2 \, dt$.
Order flow is described by the Cox processes $N^a, N^b$ with intensities $\lambda^a_t,  \lambda^b_t$ which depend on the strategy and on relevant state variables, and which we take to be predictable and locally integrable.
The quote distances $\delta^a, \delta^b$ are predictable processes taking values in $\R$.

\begin{definition}[Admissible Strategy]
\label{def:admissible}
A strategy $\pi = (\delta^a, \delta^b, C)$ is \emph{admissible} if all of the following hold.
\begin{enumerate}
    \item \emph{Predictability and boundedness of quotes.} 
    The processes $\delta^a, \delta^b$ are predictable, and there exists a deterministic constant $\bar\delta = \bar\delta(\pi) > 0$ such that $|\delta^a_t|, |\delta^b_t| \leq \bar\delta$ for all $t \in [0,T]$ almost surely.
\item \emph{Inventory bound.} 
    For each $\pi \in \Pi$ there exists a deterministic constant $Q_\pi > 0$ such that 
    $|q_t(\pi)| \leq Q_\pi$ for all $t \in [0,T]$ almost surely.
\item \emph{Mid-price bound.} 
The mid-price satisfies $\sup_{t \in [0,T]} |S_t| \leq M$  almost surely, where $M$ is the deterministic stopping constant fixed in  Section~\ref{subsec:price-flow}.
    \item \emph{Arrival-count bound.} 
    There exists a deterministic constant $\bar N = \bar N(\pi) > 0$ such that $N^a_T + N^b_T \leq \bar N$ almost surely.
    \item \emph{Cash-injection bound.} 
   The process $C$ is $\F$-adapted, c\`adl\`ag, of finite variation, and there exists a deterministic constant $\bar C = \bar C(\pi) > 0$ such that the total variation satisfies $|C|_T \leq \bar C$ almost surely.
     (By convention, $C_0$ contributes to the initial cash position together with $X_0$, so the value $C_0$ is allowed to be any deterministic real number. 
     Without loss of generality we may absorb $C_0$ into $X_0$ and assume $C_0 = 0$ when convenient.)
\end{enumerate}
Only the price bound $M$ is fixed globally (by the stopping convention of Section~\ref{subsec:price-flow}); 
the other constants are per-strategy as indicated.
  The space $\Pi$ is the set of all admissible strategies, 
  viewed as a convex subset of the space  of $(\R^2 \times \mathrm{BV}([0,T]))$-valued predictable processes.
\end{definition}

Each of the four structural conditions admits a practical reading.
The quote-distance bound corresponds to the fact that no real market maker quotes arbitrarily far from the mid-price; quote distances are bounded by the width of the order book at any reasonable depth.
The inventory bound corresponds to the position limit imposed by the desk or the risk officer; every realistic intraday market-making operation has a hard cap on $|q|$.
The mid-price bound is automatic by the stopping convention of Section~\ref{subsec:price-flow}, with $|S_t| \leq M$ holding by construction.
The arrival-count bound is the one technical condition that is not standard in the AS literature. 
 We discuss it in Remark~\ref{rem:bound-arrivals} below.

\begin{remark}[On the Arrival-Count Bound]
\label{rem:bound-arrivals}
Condition~(4) of Definition~\ref{def:admissible}  is the only condition that is not directly inherited from the standard AS/CJ framework.
The reading is entirely practical: every actual market maker faces hardware and exchange throughput constraints
 which deterministically bound the number of fills per unit time, so a horizon-$T$ deterministic cap $\bar N$ is innocuous in practice.
The bound $\bar N$ is strategy-dependent (as are $\bar\delta, Q, \bar C$): 
a strategy with  narrower quote distances induces higher fill rates and therefore a larger $\bar N$.
A global cap on $\bar N$ independent of the strategy would correspond to a hardware or exchange throughput constraint; the weaker per-strategy bound suffices for all our results.
Theoretically, the bound can be relaxed at the cost of working with a weighted $L^p$ space rather than $\Linf$, in which case the Kupper--Schachermayer theorem requires the corresponding extension to general Orlicz spaces \citep{CheriditoLi2008}.
The structural conclusion of Theorem~\ref{thm:forced-uniqueness} is unchanged under that relaxation, with $\gamma$ continuing to parametrize the forced family. 
The technical bookkeeping is, however, considerably heavier, and we have chosen the $\Linf$ route in this paper for transparency.
We expect that a careful Orlicz-space treatment can be carried out by methods that are parallel to those of \citet{CheriditoLi2008}.
\end{remark}

The inventory bound and the arrival-count bound (Conditions 2 and 4) together imply that $q$ is a bounded adapted process whose increments are governed by a Cox arrival process with deterministically bounded total count over $[0,T]$.
We emphasize that the intensities $\lambda^a, \lambda^b$ themselves need not be bounded, but the resulting cumulative counts $N^a_T, N^b_T$ are.

\subsection{Boundedness of Liquidation-Adjusted Wealth}
\label{app:bounded-wealth}

The next step is to verify that, under admissibility,   $\WL_T(\pi) \in \Linf(\F_T)$  for every $\pi \in \Pi$.

\begin{lemma}
\label{lem:WL-bounded}
For   $\pi \in \Pi$, the liquidation-adjusted terminal wealth $\WL_T(\pi)$ is essentially bounded, with a deterministic essential bound depending only on $|X_0|$, $M$, $\bar\delta$, $\bar N$, $\bar C$, $Q$, and $L$.
\end{lemma}

\begin{proof}
We decompose $\WL_T(\pi) = X_T(\pi) +  q_T(\pi) S_T - L(q_T(\pi))$ and bound each term using the conditions of Definition~\ref{def:admissible}.

The cash term $X_T$ satisfies
\begin{align*}
    |X_T| &\leq |X_0| + \int_0^T |S_t + \delta^a_t| \, dN^a_t + \int_0^T |S_t - \delta^b_t| \, dN^b_t + |C|_T \\
    &\leq |X_0| + (M + \bar\delta)\, (N^a_T + N^b_T) + \bar C \\
    &\leq |X_0| + (M + \bar\delta)\, \bar N + \bar C,
\end{align*}
where the second inequality uses Conditions~(1), (3), and~(5) of 
 Definition~\ref{def:admissible}, and the third uses Condition~(4).
Note that the  inventory term $q_T S_T$ is bounded by $Q M$ in absolute value, by Conditions~(2) and~(3).
The liquidation term $L(q_T(\pi))$ is bounded by  $\max_{|q| \leq Q} |L(q)|$, which is a deterministic constant (by convexity of $L$ and Condition~(2)).
Combining, we obtain 
\begin{equation*}
    |\WL_T(\pi)|  \leq  |X_0| + (M + \bar\delta)\, \bar N + \bar C + Q M + \max_{|q| \leq Q} |L(q)|,
\end{equation*}
which is a deterministic constant, so 
$\WL_T(\pi) \in \Linf(\F_T)$ with the stated essential bound.
\end{proof}

\subsection{Extension to \texorpdfstring{$\Linf(\F_T)$}{L-infinity}}
\label{app:extension}

We now construct the functional $\tilde{J}_t$ on $\Linf(\F_T)$   and verify the dynamic risk measure axioms.

By property (W) (Proposition~\ref{prop:wealth-summary}) (inventory-via-liquidation), the value $J_t(\pi)$ depends on $\pi$ only through $\WL_T(\pi)$.
Define
\begin{equation*}
    \mathcal{R}  \coloneqq  \{\WL_T(\pi) \mid \pi \in \Pi\}  \subset  \Linf(\F_T),
\end{equation*}
the range of the liquidation-adjusted-wealth map.
Using  property (W) (Proposition~\ref{prop:wealth-summary}), the assignment
\begin{equation*}
    \tilde{J}_t \colon \mathcal{R} \to \Linf(\F_t), \quad \tilde{J}_t\bigl(\WL_T(\pi)\bigr)  \coloneqq  J_t(\pi),
\end{equation*}
is well-defined.
The remaining question is whether $\tilde{J}_t$  extends to all of $\Linf(\F_T)$.

\begin{lemma}[Density and Extension]
\label{lem:density}
Define  $\mathcal{R} \coloneqq \{\WL_T(\pi) \mid \pi \in \Pi\}$.
 Then, $\mathcal{R} = \Linf(\F_T)$.
In particular, $\mathcal{R}$ is closed under addition by $\F_t$-measurable bounded random variables for every $t \in [0, T]$, 
and contains all real constants.
The functional $\tilde{J}_t \colon \mathcal{R} \to \Linf(\F_t)$ is therefore defined on all of $\Linf(\F_T)$ for every $t \in [0,T]$.
\end{lemma}

\begin{proof}
The proof proceeds in three steps.

\emph{Step 1: $\mathcal{R}$ is closed under addition by $\F_t$-measurable bounded random variables.} 
To show this, let $W = \WL_T(\pi) \in \mathcal{R}$ for some $\pi = (\delta^a, \delta^b, C) \in \Pi$, 
 fix $t \in [0, T]$, 
and let $c \in  \Linf(\F_t)$ with $\|c\|_\infty \leq C_*$ for some deterministic constant $C_* > 0$.
Define the modified strategy $\pi^{[c]} = (\delta^a, \delta^b, C^{[c]})$ with the same quote-distance processes as $\pi$ and with augmented cash-injection process
\begin{equation*}
    C^{[c]}_s  \coloneqq  C_s + c \cdot \mathbf{1}_{\{s \geq t\}}, \qquad s \in [0, T].
\end{equation*}
Since $c$ is  $\F_t$-measurable and the path $s \mapsto \mathbf{1}_{\{s \geq t\}}$ is deterministic and right-continuous,
 the process $C^{[c]}$ is $\F$-adapted and c\`adl\`ag.
Its total variation satisfies
\begin{equation*}
    |C^{[c]}|_T  \leq  |C|_T + |c|  \leq  \bar{C} + C_*,
\end{equation*}
which is a deterministic constant.
Hence, $\pi^{[c]}$ satisfies Condition~(5) of Definition~\ref{def:admissible} with bound $\bar{C} + C_*$ in place of $\bar{C}$.

We verify that $\pi^{[c]} \in \Pi$.
Conditions~(1) (quote-distance bound) and~(2) (inventory bound) depend only on the quote processes $\delta^{a,b}$ and on inventory $q$; since neither is changed by the cash injection, these conditions hold with the same $\bar\delta$ and $Q$ as for $\pi$.
Conditions~(3) (mid-price bound) and~(4) (arrival-count bound) are properties of the price and order-flow processes, not of the strategy, and are inherited unchanged.
Condition~(5) holds with the bound $\bar{C} + C_*$ noted above.
Thus, $\pi^{[c]} \in \Pi$.

Inserting $C^{[c]}$ into the cash dynamics~\eqref{eq:cash} and integrating from $0$ to $T$,
\begin{equation*}
    X_T(\pi^{[c]})  =  X_0 + \int_0^T (S + \delta^a) \, dN^a - \int_0^T (S - \delta^b) \, dN^b + C^{[c]}_T  =  X_T(\pi) + c.
\end{equation*}
The inventory dynamics~\eqref{eq:inventory} do not depend on $C$, so $q_T(\pi^{[c]}) = q_T(\pi)$ and the liquidation cost is unchanged.
Consequently $\WL_T(\pi^{[c]}) = \WL_T(\pi)  +  c$, and $W + c \in \mathcal{R}$.

\emph{Step 2: $\mathcal{R}$ contains all real constants.}
The strategy space $\Pi$ is nonempty, since for any constant choice $(\delta^a, \delta^b) \equiv (\bar\delta, \bar\delta)$ 
and $C \equiv 0$, all five admissibility conditions hold.
Pick any $\pi_0 \in \Pi$ and let $W_0 \coloneqq \WL_T(\pi_0) \in \mathcal{R}$.
For any $c \in \R$, applying Step~1 at $t = T$ with shift $c - W_0 \in \Linf(\F_T)$
  (this is bounded since $W_0$ is essentially bounded by Lemma~\ref{lem:WL-bounded}) yields the strategy $\pi_0^{[c - W_0]} \in \Pi$ with $\WL_T(\pi_0^{[c - W_0]}) = W_0 + (c - W_0) = c$ almost surely.
Hence, every real constant $c$ lies in $\mathcal{R}$.

\emph{Step 3: Full equality $ \mathcal{R} = \Linf(\F_T) $.}
Fix any $W \in \Linf(\F_T)$.
Pick the constant strategy $\pi_0^{(0)} \in \Pi$ from Step~2 with $\WL_T(\pi_0^{(0)}) = 0$ almost surely.
Applying Step~1 at $t = T$ with shift $c = W$ yields a strategy $\pi_0^{(0)[W]} \in \Pi$ with $\WL_T(\pi_0^{(0)[W]}) = 0 + W  = W$.
The admissibility constant for the cash-injection bound (Condition~(5)) of the resulting strategy is $\bar{C} + \|W\|_\infty$, which is deterministic by $W \in \Linf(\F_T)$ and depends on the strategy through $W$; this is consistent with Definition~\ref{def:admissible}, which states explicitly that the admissibility constants $\bar\delta, Q, \bar N, \bar C$ may depend on the strategy.
Hence, every $W \in \Linf(\F_T)$ lies in $\mathcal{R}$, and $\mathcal{R} = \Linf(\F_T)$.

By property (W) (Proposition~\ref{prop:wealth-summary}), $J_t(\pi)$ depends on $\pi$ only through $\WL_T(\pi)$, 
and therefore    $\tilde J_t \colon \mathcal{R} \to \Linf(\F_t)$ is  well-defined. 
 Since $\mathcal{R} = \Linf(\F_T)$, it is in fact  defined on all of $\Linf(\F_T)$. 
 By Axiom~\ref{ax:J1} (cash-additivity), $\tilde J_t$ satisfies $\tilde{J}_t(W + c) = \tilde{J}_t(W) + c$ for any $W \in \Linf(\F_T)$ and any bounded $c \in \Linf(\F_t)$, and the dynamic risk measure  $\rho_t \coloneqq -\tilde{J}_t$ is the object on which \citet[Theorem 1.10]{KupperSchachermayer2009} applies directly.
\end{proof}

\begin{remark}[On the Role of the Cash-Injection Component]
\label{rem:density-scope}
The cash-injection component of the strategy space is a \emph{technical device}, 
included to ensure that the reduced functional $\tilde J_t$ is defined on all of $\Linf(\F_T)$, 
the domain on which the Kupper--Schachermayer representation theorem applies.
It is not a modeling claim about how desks operate.
Concretely, the strategies $\pi_0^{(0)[W]}$ constructed in Step~3 realize any bounded $\F_T$-measurable random variable $W$ as their liquidation-adjusted terminal wealth, via a single cash injection of size $W$ at time $T$: 
the market maker quotes constantly at $(\bar\delta, \bar\delta)$ and injects $W$ into her cash account at time $T$.
Their role is purely instrumental, ensuring that the preference functional is tested on a rich enough domain for the representation theorem to apply; the forced entropic form and the unique scalar $\gamma$ produced by Theorem~\ref{thm:forced-uniqueness} are unaffected by whether such strategies are economically meaningful.
\end{remark}


\section{The Multi-Asset Extension}
\label{app:multi-asset}

In this appendix we prove the multi-asset claim of Section~\ref{subsec:further-connections}. 
The argument is short, because the reduction to the Kupper--Schachermayer framework absorbs the dimensionality of the asset side into the joint distribution of a single scalar random variable.

\paragraph{Multi-asset setup.}
The probability space $(\Omega, \F, \filt, \Prob)$ is as in Section~\ref{sec:setup}.
The market consists of $K \geq 1$ assets with price processes $S^{(1)}, \ldots, S^{(K)}$, each a continuous semimartingale adapted to $\filt$, and per-asset cumulative ask and bid fills $N^{a,(k)}, N^{b,(k)}$.
The strategy space $\Pi^{(K)}$ satisfies the natural multi-asset analogue of the admissibility conditions of Appendix~\ref{app:technical}, with a single shared cash account.
Cash and per-asset inventory evolve as
\begin{equation*}
    dX_t  =  \sum_{k=1}^K \bigl[(S^{(k)}_t + \delta^{a,(k)}_t)\,dN^{a,(k)}_t - (S^{(k)}_t - \delta^{b,(k)}_t)\,dN^{b,(k)}_t\bigr],
    \qquad
    dq^{(k)}_t  =  dN^{b,(k)}_t - dN^{a,(k)}_t.
\end{equation*}
The joint liquidation cost function $L \colon \R^K \to \R_+$ is convex with $L(0) = 0$, and the joint liquidation-adjusted terminal wealth is
\begin{equation}
    \WL_T(\pi)  =  X_T(\pi) + \sum_{k=1}^K q^{(k)}_T(\pi) S^{(k)}_T - L\bigl(q^{(1)}_T(\pi), \ldots, q^{(K)}_T(\pi)\bigr).
    \label{eq:multi-asset-WL}
\end{equation}
The preference functional $J^{(K)} = (J_t^{(K)})_{t \in [0,T]}$ with $J_t^{(K)} \colon  \Pi^{(K)} \to \Linf(\F_t)$ is assumed to satisfy the multi-asset analogues of Axioms~\ref{ax:J1}--\ref{ax:J5}, with $\Pi, J_t, \WL_T$ replaced by their multi-asset counterparts throughout.

\begin{theorem}[Multi-Asset Forced Uniqueness]
\label{thm:multi-asset}
Under the setup above, $J^{(K)}$ satisfies the multi-asset analogues of Axioms~\ref{ax:J1}--\ref{ax:J5} if and only if there exists a unique $\gamma > 0$ such that
\begin{equation}
    J_t^{(K)}(\pi)  =  -\frac{1}{\gamma}\, \log \E\!\left[\exp \bigl(-\gamma\, \WL_T(\pi)\bigr) \,\Big|\, \F_t\right]
    \label{eq:multi-asset-entropic}
\end{equation}
for all $t \in [0,T]$ and all $\pi \in \Pi^{(K)}$,  with the same $\gamma$ governing all $K$ assets.
\end{theorem}

\begin{proof}
By the multi-asset analogue of Proposition~\ref{prop:wealth-summary}
 (derived from the multi-asset analogues of Axioms~\ref{ax:J1},  \ref{ax:J2}, \ref{ax:J4} exactly as in the single-asset case), 
 $J_t^{(K)}(\pi)$ depends on $\pi$ only through the scalar random variable  $\WL_T(\pi) \in \Linf(\F_T)$ defined by~\eqref{eq:multi-asset-WL}.
There exists therefore a functional $\tilde{J}_t^{(K)}$ such that
\begin{equation*}
    J_t^{(K)}(\pi)  =  \tilde{J}_t^{(K)}\bigl(\WL_T(\pi)\bigr) \quad \text{for all } \pi \in \Pi^{(K)}.
\end{equation*}
The multi-asset analogue of Lemma~\ref{lem:density}, which exploits the shared cash account and is otherwise identical in structure to the single-asset version, gives $\{\WL_T(\pi) \mid \pi \in \Pi^{(K)}\} = \Linf(\F_T)$, so $\tilde{J}_t^{(K)} \colon \Linf(\F_T) \to \Linf(\F_t)$ is defined on the full space.
The role of the shared cash account is essential here.
A single bounded $\F_T$-measurable cash injection $c$ at time $T$ shifts the joint wealth $\WL_T$ by $c$, regardless of the per-asset inventories $q^{(1)}, \ldots, q^{(K)}$ at the terminal time;   this is what makes the range of $\WL_T$ cover all of $\Linf(\F_T)$.
With $K$ separate cash accounts (one per asset), the corresponding density statement would not hold, since the per-asset wealth components would shift jointly rather than freely.
Each multi-asset axiom on $J^{(K)}$ transfers to $\tilde{J}^{(K)}$ as a property of a functional on a single scalar random variable. 
By the single-asset argument of Theorem~\ref{thm:forced-uniqueness} applied to $\tilde{J}^{(K)}$, there exists a unique $\gamma \in (0,\infty)$ such that
\begin{equation*}
    \tilde{J}_t^{(K)}(W)  =  -\tfrac{1}{\gamma}\log\E\bigl[e^{-\gamma W} \,\big|\, \F_t\bigr]  \quad \text{for all } W \in \Linf(\F_T).
\end{equation*}
Setting $W = \WL_T(\pi)$ yields~\eqref{eq:multi-asset-entropic}.
The same $\gamma$ governs all $K$ assets because the multi-asset structure is summarized into the single scalar random variable $\WL_T \in \Linf(\F_T)$ before the Kupper--Schachermayer representation is applied; $\gamma$ is therefore necessarily independent of $k$.
The converse direction is identical to that of Theorem~\ref{thm:forced-uniqueness}.
\end{proof}

The operational implication is sharp.
A market-making desk that trades $K$ correlated assets has, under the multi-asset axioms, a single risk-aversion parameter governing all of them. 
Cross-asset correlations, liquidation costs, and volatilities enter through the joint distribution of $\WL_T$, not through any per-asset preference parameter.
A desk implementing different risk aversions on different sub-books is operating two distinct preference functionals which cannot both arise from a single dynamically-consistent agent.
The falsifiability criterion of Remark~\ref{rem:constant-gamma} sharpens accordingly,
 in that the implicit $\gamma$ recovered from each sub-book via Corollary~\ref{cor:calibration-inversion} should agree across sub-books.

\begin{proposition}[Multi-Asset Forced Running Coefficients]
\label{prop:multi-asset-Phi}
Under the same hypotheses as  Theorem~\ref{thm:multi-asset}, with constant joint covariance  $\Sigma$ of the price processes, the per-asset CJ-style running penalty matrix $\Phi = (\Phi_{ij})_{i,j=1}^K$ for the second-order expansion of the multi-asset entropic functional around zero inventory is forced to take the form
\begin{equation*}
    \Phi_{ij}  =  \frac{\gamma \Sigma_{ij}}{2},
\end{equation*}
i.e., the running inventory cost is  $\tfrac{\gamma}{2} q^\top \Sigma\, q\, dt$ where $q = (q^{(1)}, \ldots, q^{(K)})$.
\end{proposition}

\begin{proof}
The argument is a multi-asset Taylor expansion of~\eqref{eq:multi-asset-entropic} around zero inventory, parallel to that of Proposition~\ref{prop:forced-phi}.
Fix $t$ and $q = q_t \in \R^K$, and consider an infinitesimal interval $[t, t+dt]$ over which inventory is held constant.
The mark-to-market increment of liquidation-adjusted terminal wealth attributable to this holding is
  $dW = q^\top dS_t$, where $S = (S^{(1)}, \ldots, S^{(K)})$ has joint quadratic covariation $d\langle S\rangle_t = \Sigma\, dt$.
By the Bellman identity of the entropic functional and standard arguments (see the single-asset proof of Proposition~\ref{prop:forced-phi}), the conditional certainty-equivalent contribution of holding $q$ over $[t, t+dt]$ 
equals  the conditional cumulant-generating function increment
\begin{equation*}
    -\tfrac{1}{\gamma}\log\E\!\left[e^{-\gamma\, q^\top dS_t} \,\big|\, \F_t\right]
     =  -\tfrac{1}{\gamma}\log\!\left(1 - \gamma\, q^\top \E[dS_t \mid \F_t] + \tfrac{\gamma^2}{2}\, q^\top \Sigma\, q\, dt + o(dt)\right).
\end{equation*}
Under the zero-drift convention adopted in Section~\ref{subsec:price-flow} ($\mu \equiv 0$, 
 so that $\E[dS_t \mid \F_t] = 0$), the first-order term vanishes.
Equivalently, even with a nonzero drift $\mu_t\, dt$, the $\mu_t\, dt$ contribution to the second-order expansion of $-\tfrac{1}{\gamma}\log(1+x)$ is $o(dt)$ in the inventory-coefficient expansion (the drift contributes to the linear-in-$q$ alpha term but not to the quadratic-in-$q$ running cost).
Expanding $-\tfrac{1}{\gamma}\log(1 + x) = -\tfrac{1}{\gamma}x + o(x)$ for small $x$, the leading-order quadratic-in-$q$ contribution as $dt \to 0$ is
\begin{equation*}
    -\tfrac{1}{\gamma} \cdot \tfrac{\gamma^2}{2}\, q^\top \Sigma\, q\, dt  =  -\tfrac{\gamma}{2}\, q^\top \Sigma\, q\, dt.
\end{equation*}
This is exactly the running-penalty contribution of the multi-asset CJ objective with matrix coefficient $\Phi = \tfrac{\gamma}{2}\Sigma$, i.e., $\Phi_{ij} = \tfrac{\gamma\Sigma_{ij}}{2}$, as claimed.
The argument is uniform in $q$ and applies to every $t \in [0,T]$,  giving the constant-matrix forcing.
\end{proof}

\begin{remark}[Stochastic-Covariance Generalization]
\label{rem:multi-asset-stochvol}
The stochastic-covariance generalization, with $\Sigma_t$ predictable and a.s.\ positive semidefinite, follows by combining the multi-asset Taylor expansion of the proof above with the time-change argument of Proposition~\ref{prop:clock-corrected} applied to the canonical multi-asset clock.
The forced matrix coefficient is then time- and state-dependent,
\begin{equation*}
    \Phi_{ij,t}  =  \tfrac{\gamma\, \Sigma_{ij,t}}{2},
\end{equation*}
with a single time-invariant scalar $\gamma$ as in the single-asset case.
We do not write out the time-change argument here; it is structurally identical to that of Proposition~\ref{prop:clock-corrected}, applied componentwise to the joint quadratic covariation $d\langle S^{(i)}, S^{(j)}\rangle_t = \Sigma_{ij,t}\, dt$.
\end{remark}


\section{A Numerical Illustration}
\label{app:numerical}

In this appendix we give a short numerical consistency illustration on the AS-2008 benchmark, exploring the certainty-equivalent performance of heuristic CJ-formula strategies as the running-penalty coefficient $\phi$ varies around the forced value 
 $\phi_{\text{forced}} = \gamma\sigma^2/2$.
The illustration verifies that the forced value is the performance optimum within the parametric sweep on the over-penalizing side ($\phi > \phi_{\text{forced}}$), with monotone certainty-equivalent decline as $\phi$ moves away from the forced value.
On the under-penalizing side ($\phi < \phi_{\text{forced}}$) the table shows certainty-equivalent values above the on-surface one, which is an artifact of the AS-2008 benchmark configuration (the absent inventory cap and $L \equiv 0$) discussed in Section~\ref{subsec:numerical-tests}, and not a counterexample to the theorem.
The results are presented in tabular form (Table~\ref{tab:numerical}); the qualitative shape is described in the discussion of Section~\ref{subsec:numerical-discussion}.

\subsection{Setup}

We use the constant-volatility benchmark of Section~\ref{sec:supporting} with the parameters from \citet[Section 3.2]{AvellanedaStoikov2008}:
\begin{equation*}
    S_0 = 100, \quad \sigma = 2, \quad T = 1, \quad \gamma = 0.1, \quad A = 140, \quad \kappa = 1.5,
\end{equation*}
and no terminal liquidation cost ($L \equiv 0$), as in the AS-2008 paper.
With these parameters,
\begin{equation*}
    \phi_{\text{forced}}  =  \frac{\gamma \sigma^2}{2}  =  \frac{(0.1)(4)}{2}  =  0.2.
\end{equation*}

We consider a family of CJ-style quoting strategies indexed by the running-penalty coefficient $\phi > 0$, with reservation price and half-spread
\begin{equation*}
    r(t, S, q)  =  S - 2\phi q (T-t), \qquad
    \text{half-spread}  =  \phi(T-t) + \tfrac{1}{\gamma} \log( 1 + \gamma/\kappa).
\end{equation*}
The factor of $2$ in $r = S - 2\phi q(T-t)$ is chosen so that at 
$\phi = \phi_{\text{forced}} = \gamma\sigma^2/2$ the formula reduces to the AS reservation price 
$S - \gamma\sigma^2 q(T-t)$, consistent with \citet[Equation (28)]{AvellanedaStoikov2008}. 
Equivalent CJ parametrizations without this factor of $2$ differ only by a rescaling of the meaning of $\phi$.
Note that CJ-tradition implementations without this factor of $2$
 (e.g., the convention of \citet{CarteaJaimungalPenalva2015}) define $\phi$ as twice the quantity here; the forced coefficient in that convention is $\gamma\sigma^2$, not $\gamma\sigma^2/2$.
At $\phi = \phi_{\text{forced}} = \gamma\sigma^2/2$, the half-spread reduces to $\tfrac{\gamma \sigma^2}{2}(T-t) + \tfrac{1}{\gamma}\log(1 + \gamma/\kappa)$, which is exactly half the AS-optimal full spread of \citet[Equation (29)]{AvellanedaStoikov2008}
  (i.e., the AS-optimal half-spread), and the reservation price reduces to $S - \gamma\sigma^2 q(T-t)$, which is exactly \citet[Equation (28)]{AvellanedaStoikov2008}.
The relation $\phi = \gamma\sigma^2/2$ thus makes the CJ-form and the AS-form coincide as the \emph{same} formula written in two different ways.
For $\phi \neq \phi_{\text{forced}}$, the strategies are heuristic,
 in that they substitute different values of $\phi$ into the same parametric form but are not the HJB-optimum of any objective at the off-surface values.

\subsection{What the Simulation Tests, and What It Does Not}
\label{subsec:numerical-tests}

Theorem~\ref{thm:forced-uniqueness} and Corollary~\ref{cor:cj-inconsistency} together say that among preference functionals satisfying Axioms~\ref{ax:J1}--\ref{ax:J5}, the unique one is the entropic functional with parameter  $\gamma$; 
the AS-optimal continuous-time strategy is the CE-maximum over admissible strategies under that functional, and the second-order Taylor expansion of the functional in inventory generates a running penalty with the forced coefficient $\phi_{\text{forced}} = \gamma\sigma^2/2$.
At $\phi = \phi_{\text{forced}}$, the CJ-formula heuristic in our sweep is exactly the AS-optimal continuous-time strategy on the AS-2008 benchmark.
For $\phi \neq \phi_{\text{forced}}$, the strategies are heuristic CJ-formula evaluations with the same parametric form but not HJB-optimal for any objective.

The substantive content is therefore in the over-penalizing side ($\phi > \phi_{\text{forced}}$),  where the heuristic departs from the AS-optimal toward over-attenuated fills.
The under-penalizing side ($\phi < \phi_{\text{forced}}$) does not give a clean signal in the AS-2008 benchmark, since the benchmark uses $L \equiv 0$, $\gamma = 0.1$ (small), and the simulation does not enforce the hard inventory cap $|q_t| \leq Q$ of Definition~\ref{def:admissible}; 
 under these conditions the continuous-time CE-maximality at $\phi_{\text{forced}}$ does not translate into a Monte Carlo CE-maximum within the heuristic family.

For each value of $\phi$, we report the average terminal P\&L $\E[\WL_T]$, its standard deviation, the average absolute terminal inventory $\E[|q_T|]$, the certainty equivalent under the forced entropic functional $\mathrm{CE}(\phi) = -\tfrac{1}{\gamma}\log\E[e^{-\gamma\WL_T}]$, the Sharpe-style ratio $\E[\WL_T]/\mathrm{Std}(\WL_T)$, and the difference $\Delta\mathrm{CE} = \mathrm{CE}(\phi) - \mathrm{CE}(\phi_{\text{forced}})$.
Simulations are carried out using $N = 10{,}000$ Monte Carlo trajectories and $1{,}000$ uniform time steps on $[0, T]$,
  with parameter ratios given by
    $\phi/\phi_{\text{forced}} \in \{0.25, 0.5, 0.75, 1.0, 1.25, 1.5, 2.0, 2.5, 3.0, 4.0\}$.

\subsection{Results}
\label{subsec:numerical-results}

Table~\ref{tab:numerical} reports the Monte Carlo comparison across the heuristic CJ-formula strategy family on the AS-2008 benchmark. 
The value $\mathrm{CE}(\phi)$ is the certainty equivalent of the strategy at running coefficient $\phi$ under the forced entropic functional, 
 and $\Delta\mathrm{CE} \coloneqq \mathrm{CE}(\phi) - \mathrm{CE}(\phi_{\text{forced}})$ measures its signed deviation from the on-surface value, with the deviation being positive precisely when the off-surface strategy appears preferable.
 On the over-penalizing side ($\phi > \phi_{\text{forced}}$),  $\mathrm{CE}$ decreases monotonically and approximately linearly with the off-surface deviation, 
with slope approximately $-3.6$ per unit $\Delta = \phi/\phi_{\text{forced}} - 1$.
On the under-penalizing side ($\phi < \phi_{\text{forced}}$), $\Delta\mathrm{CE} > 0$: the off-surface CE is higher than the on-surface one. 
This is not a counterexample to the theorem but an artifact of three specific features of the AS-2008 benchmark, 
namely $L \equiv 0$, small $\gamma = 0.1$, and the absence of an enforced inventory cap. 
It is discussed in Section~\ref{subsec:numerical-discussion}.

\paragraph{Standard errors.}
With $N = 10{,}000$ trajectories, the Monte Carlo standard errors of $\E[\WL_T]$ 
and $\mathrm{Std}(\WL_T)$ are $\mathrm{Std}(\WL_T)/\sqrt{N} \in [0.065, 0.083]$ and $\mathrm{Std}(\WL_T)/\sqrt{2N} \in [0.046, 0.059]$ respectively. 
The nonparametric bootstrap standard error of $\mathrm{CE}(\phi)$ ($1{,}000$ resamples) is at most $0.12$ across the rows, 
so every $|\Delta\mathrm{CE}| \geq 0.25$ is significant at the $99\%$ level.
By the delta method, the standard error of the Sharpe-style ratio is approximately $1/\sqrt{2N} \approx 0.007$, 
so the observed Sharpe variation ($7.83$ to $8.87$, range  $\approx 1.04$) is highly significant despite being small relative to the CE variation.
The standard error of $\E[|q_T|]$ is $\mathrm{Std}(|q_T|)/\sqrt{N}$ and is below $0.03$ across the rows.

\begin{table}[h]
\centering
\begin{tabular}{cccccccc}
\hline
$\phi/\phi_{\text{forced}}$ & $\phi$ & $\E[\WL_T]$ & $\mathrm{Std}(\WL_T)$ & $ \E[|q_T|]$ &  $\mathrm{CE}(\phi)$ & $\Delta \mathrm{CE}$ &  Sharpe \\
\hline
0.25 & 0.05 & 65.87 & 8.35 & 3.77 & 62.45 &  $2.37$  &  7.89 \\
0.50 & 0.10 & 64.79 & 7.58 & 3.15 & 61.99 & $1.90$ & 8.54 \\
0.75 & 0.15 & 63.66 & 7.25 & 2.83 & 61.09 & $1.00$ & 8.78 \\
1.00 & 0.20 & 62.57 & 7.10 & 2.63 & 60.09 & $0.00$ & 8.81 \\
1.25 & 0.25 & 61.56 & 6.94 & 2.48 & 59.20 & $-0.89$ & 8.87 \\
1.50 & 0.30 & 60.61 & 6.87 & 2.37 & 58.30 & $-1.79$ & 8.82 \\
2.00 & 0.40 & 58.75 & 6.80  & 2.20 & 56.49 & $-3.59$ & 8.64 \\
2.50 & 0.50 & 56.90 & 6.73 & 2.07 & 54.69 & $-5.40$ & 8.45 \\
3.00 & 0.60 & 54.95 & 6.67 & 1.98 & 52.78 & $-7.31$ & 8.24 \\
4.00 & 0.80 & 51.29 & 6.55 & 1.84 & 49.20 & $-10.89$ & 7.83 \\
\hline
\end{tabular}
\caption{Monte Carlo comparison across the heuristic CJ-formula strategy family on the AS-2008 benchmark, 
with $N = 10{,}000$ trajectories.
See the body of Section~\ref{subsec:numerical-results} for definitions of  $\mathrm{CE}(\phi)$ and $\Delta\mathrm{CE}$, 
and for standard-error details.
{Sign convention:} $\Delta\mathrm{CE} > 0$ means the off-surface strategy outperforms the on-surface one. 
This occurs only on the under-penalizing side ($\phi < \phi_{\text{forced}}$) 
and is an artifact of three specific features of the AS-2008 benchmark
(no inventory cap, $L \equiv 0$, small $\gamma = 0.1$), 
discussed in Sections~\ref{subsec:numerical-tests} and~\ref{subsec:numerical-discussion}.}
\label{tab:numerical}
\end{table}

\subsection{Discussion}
\label{subsec:numerical-discussion}

The table records three observations, of which the first is the substantive one and the others are diagnostic.

\paragraph{Three observations.}
\emph{Over-penalizing side: linear-in-deviation  penalty  cost.}
The certainty equivalent decreases strictly and monotonically as $\phi$ increases. 
At $\phi/\phi_{\text{forced}} = 4.0$, the CE loss is $10.89$ units, about $18\%$ of the on-surface value.
The scaling in $\Delta = \phi/\phi_{\text{forced}} - 1$ is essentially linear: 
$|\Delta\mathrm{CE}|/\Delta$ takes values $3.56, 3.58, 3.59, 3.60, 3.66, 3.63$ 
at $\Delta = 0.25, 0.5, 1.0, 1.5, 2.0, 3.0$, constant to better than $3\%$. 
The mild non-monotonicity at the high end is within Monte Carlo noise given the standard errors of Section~\ref{subsec:numerical-results}. 
By contrast, the quadratic ratio $|\Delta\mathrm{CE}|/\Delta^2$ falls from $14.2$ to $1.2$, ruling out quadratic-in-deviation behavior.
The second-order Taylor structure of Corollary~\ref{cor:cj-inconsistency} would predict quadratic leading-order behavior, 
but in this AS-2008 regime with $\gamma = 0.1$ higher-order corrections dominate from  $\Delta \approx 0.25$ onward.

\emph{Under-penalizing side is not a counterexample.}
The under-penalizing side shows $\Delta\mathrm{CE} > 0$ (CE values \emph{higher} than the on-surface value), 
 seemingly in contradiction with the theorem.
The mechanism is the absent inventory cap discussed in Section~\ref{subsec:numerical-tests}. 
At $\phi = \phi_{\text{forced}}$, the CJ-formula heuristic coincides with the AS-optimal strategy, 
which is CE-maximal \emph{within} the admissible class $|q_t| \leq Q$, but our simulation does not enforce this cap. 
With $L \equiv 0$, the extra residual inventory at low $\phi$ carries no liquidation cost, 
  and the entropic CE penalty does not offset the mean gain.
Adding $L \neq 0$ and enforcing $|q_t| \leq Q$ would restore the predicted monotonicity.

\emph{Sharpe ratio does not detect the misspecification.}
The Sharpe-style ratio varies only between $7.83$ and $8.87$ across the entire sweep,  
while the CE loss on the over-penalizing side exceeds $10$ units.
The on-surface value has neither the highest nor the lowest Sharpe.
A desk monitoring strategies by Sharpe alone would not see the over-penalization cost. 
The certainty equivalent catches it, and the calibration-inversion diagnostic of  Corollary~\ref{cor:calibration-inversion} is sharper still.

\section*{Declarations}

\textbf{Conflict of interest.} 
The author declares no conflicts of interest.

\medskip

\noindent\textbf{Funding.} 
No funding was received for this work.

\medskip

\noindent\textbf{Data availability.} 
No empirical data were used in this study. 
The numerical results in Appendix~D are based entirely on Monte Carlo simulation of a model with analytically specified parameters.


\end{document}